\documentclass[imslayout,preprint]{imsart}
\usepackage{pdfpages}
\RequirePackage[OT1]{fontenc}

\usepackage{graphicx}
\usepackage{array}
\usepackage{enumitem}

\usepackage[authoryear]{natbib}
\usepackage{amsmath}
\usepackage{amsthm}
\usepackage{amssymb}
\usepackage{amsfonts}
\usepackage{bbm}
\usepackage[margin = 1in]{geometry}
\newtheorem{theorem}{Theorem}[section]
\newtheorem{lemma}[theorem]{Lemma}
\newtheorem{assumption}[theorem]{Assumption}

\newtheorem{proposition}[theorem]{Proposition}
\newtheorem{remark}[theorem]{Remark}
\usepackage{listings}
\usepackage{hyperref}
\usepackage[hypcap]{caption}
\usepackage{booktabs}
\usepackage{fancyvrb}
\usepackage{comment}
\usepackage{multirow}

\usepackage{algorithm}
\usepackage[noend]{algpseudocode}

\usepackage{color}

\makeatletter
\def\BState{\State\hskip-\ALG@thistlm}
\makeatother

\numberwithin{equation}{section}

\newcommand{\argmin}{\operatornamewithlimits{argmin}}
\newcommand{\argmax}{\operatornamewithlimits{argmax}}

\newcommand{\E}{\mathbb{E}}

\newcommand{\cF}{\mathcal{F}}

\newcommand{\floor}[1]{\lfloor #1 \rfloor}
\newcommand{\1}[1]{\boldsymbol{1}\{#1\}}

\newtheorem{corollary}[theorem]{Corollary}

\begin{document}
	
	\begin{frontmatter}
	\title{\bf Dating the Break in High-dimensional Data}
	\runtitle{Dating the Break in High-dimensional Data}
	
	\begin{aug}
		\author{\fnms{Runmin} \snm{Wang}\thanksref{t1,m1}\ead[label=e1]{rwang52@illinois.edu}},
		\and
		\author{\fnms{Xiaofeng } \snm{Shao}\thanksref{t1,m1}
			\ead[label=e3]{xshao@illinois.edu}}
		
		\thankstext{t1}{Runmin Wang is Ph.D. Candidate and Xiaofeng Shao is  Professor, Department of  Statistics, University of Illinois at Urbana-Champaign,  Champaign, IL 61820 (e-mail: rwang52, xshao@illinois.edu).}
		\runauthor{Wang and Shao}
		
		\affiliation{University of Illinois at Urbana-Champaign\thanksmark{m1}}
		
		\address{R.Wang\\Department of Statistics\\University of Illinois at Urbana-Champaign\\
			725 South Wright Street \\
			Champaign, IL 61820\\
			USA\\
			\printead{e1}}
		
		\address{X.Shao\\Department of  Statistics\\ University of Illinois at Urbana-Champaign \\ 725 South Wright Street \\
			Champaign, IL 61820\\
			USA\\
			\printead{e3}}
		
	\end{aug}
	
	\begin{abstract}
		This paper is concerned with estimation and inference for the location of a change point in the mean of independent high-dimensional data. Our change point location estimator maximizes  a new U-statistic based objective function, and its  convergence rate and asymptotic distribution after suitable  centering and normalization are obtained under mild assumptions. Our estimator turns out to have better efficiency as compared to the least squares based counterpart in the literature. Based on the asymptotic theory, we construct a confidence interval by plugging in consistent estimates of several quantities in the normalization. We also provide a bootstrap-based confidence interval and state its asymptotic validity under suitable conditions. Through simulation studies, we demonstrate favorable  
	finite sample performance of the new change point location estimator as compared to its least squares based counterpart,  and our bootstrap-based confidence intervals, as compared to several existing competitors.   The asymptotic theory based on high-dimensional U-statistic is substantially different from those developed in the literature and is of independent interest.
	\end{abstract}
	\begin{keyword}[class=MSC]
		\kwd[Primary ]{62G05}
		\kwd{62G20}
		\kwd[; secondary ]{60G15 }
		\kwd{60G42}
		\kwd{60G51}
	\end{keyword}
	
	\begin{keyword}
		\kwd{Change Point Detection}
		\kwd{High Dimension}
        \kwd{Structural Break}
		\kwd{U-statistic}		
	\end{keyword}	
	
\end{frontmatter}

	\section{Introduction}
	
	Advances in science and technology have led to an explosion of data of high dimension. Examples of high-dimensional data include
	fMRI imaging data in neuroscience, genomic data in biological science, panel time series data from economics and finance, and  
	spatio-temporal data from climate science, among others. Often statisticians assume some kind of homogeneity assumptions in analyzing such data, such as 
	i.i.d. (independent and identically distributed) or stationarity with weak serial dependence for a sequence of high-dimensional data.  
	The validity of methodology they develop can be sensitive with respect to such assumptions. In this paper, we shall focus on  a particular type of non-homogeneity, a change point
	in the mean of an otherwise i.i.d. sequence of high-dimensional data. That is, we assume that our observed data follows the one-change point model, 
	\[CP1:~X_t=\mu_t+Z_t,t=1,\cdots,n,\]
	where $Z_t$ are i.i.d. $p$-dimensional data with zero mean, and $\mu_t=\mu_1 {\bf 1}(1\le t\le k_0)+(\mu_1+\delta){\bf 1}(k_0+1\le t\le n)$. In this model, both parameters $\delta$ and $k_0=n\tau_0,$ where $\tau_0\in (0,1)$, are unknown. 
	Our main goal is to provide a new estimator of break point $\tau_0$ and
	 a confidence interval, which works in the high-dimensional setting that allows $p>>n$ and also dependence within $p$ components. 
	  To achieve this, we develop a new $U$-statistic based objective function and propose to use its maximizer as our location estimator.  Under some mild assumptions, this new estimator is shown to be  consistent with suitable convergence rate and
	  asymptotic distribution upon centering and normalization. It is also shown to be superior to the least squares based counterpart in \cite{bai2010common} and \cite{bhattacharjee2019change},  for some specific models of interest. Furthermore, we provide a bootstrap-based confidence interval that can be adaptive to the magnitude of change, theoretically justified and works well in finite sample.

	The literature on change point testing and estimation for high-dimensional data has been growing at a fast pace lately, as stimulated by the practical needs 
	of analyzing high-dimensional data with change points. From the viewpoint of mathematical statistics, the high dimensionality
	can bring substantial methodological and theoretical challenges, as many classical estimation and testing procedures developed for low-dimensional data
	 may not work in the high-dimensional setting. 
	This also brings interesting opportunities to the mathematical statistics community, as there is a great need to develop new estimation and testing methods that can accommodate 
	high dimensionality and dependence within components and over time. 
	As there is a vast literature on retrospective change point testing and estimation in the low-dimensional setting, we  refer the readers to several excellent review papers and books, see e.g., 
	\cite{csorgo1997limit}, \cite{perron2006}, and \cite{auehorvath2013}	for many references.

	Below we shall	 provide a brief review of the more recent literature on high-dimensional change point inference. 
	For testing a change point in the mean of high-dimensional data,  \cite{horvath2012change} considered an $l_2$ aggregation of one-dimensional
	 CUSUM statistics, which targets at dense alternative. For a sequence of Gaussian vectors,  \cite{enikeeva2013high} proposed a new test to detect the presence of a change point in mean and established
	 the detection boundary in different regimes that allow the dimension to approach the infinity. Their test was formed on the basis of a combination of a linear statistic and a scan statistic, which can capture both sparse and dense alternatives, but their 
	 critical values were obtained under strong Gaussian and independent components assumptions. \cite{cho2015multiple} proposed a sparse binary segmentation algorithm for detecting multiple change points in the second order structure of a high-dimensional time series, by aggregating the low-dimensional CUSUM statistics  that pass a certain threshold. 
	    \cite{cho2016change} developed a double CUSUM statistic that can be viewed as an interesting extension of the ideas in \cite{enikeeva2013high} and \cite{cho2015multiple}, and her test was shown to be consistent in estimating the change points with binary segmentation allowing for dependence over time and across cross-sections.  \cite{jirak2015uniform} considered an $l_{\infty}$ aggregation of CUSUM statistics, which aims for sparse alternatives. The ``INSPECT'' method proposed in \cite{wang2018high} was based on  sparse projection method for a single change point and it has been  extended to multiple change point estimation by combining with the wild binary segmentation [\cite{fryzlewicz2014wild}]. \cite{wang2019a} proposed a U-statistic based approach to test for change points in independent high-dimensional data via self-normalization, as an extension of \cite{shaozhang2010}, and also provided 
	    an segmentation algorithm by using the wild binary segmentation. 
	    \cite{chen2019inference} proposed an $l_{\infty}$ based statistic to test for change points in trends for high-dimensional time series with a consistent estimator of the long-run covariance matrix. Also see \cite{yu2017finite} for another $l_{\infty}$ based test for one change point alternative in mean.

	For the estimation and confidence interval construction of the break point $\tau_0$ (or $k_0$), there have been many papers written on this topic when the dimension $p$ is low and fixed; see early work by \cite{hinkley70}, \cite{hinkley72}, \cite{picard85}, \cite{bhattacharya87}, \cite{yao87} and \cite{bai1994least}, among others. There have been  extensions to the change point problems in linear regression and multivariate time series; see \cite{bai97review}, \cite{baiperron98} and \cite{bai98restud}, but all these works focused on the low-dimensional case.  Relatively little is done in the high-dimensional setting.  
	 \cite{bai2010common} considered a least square estimator for the change point location in the panel data with independent cross section units and weak dependence over time, and obtained the asymptotic distribution of break date estimator.   
	Recently \cite{bhattacharjee2019change} extended the least squares method in \cite{bai2010common} to high-dimensional time series and their setting allowed for both cross-sectional and serial dependence.  \cite{bai2010common}, \cite{bhattacharjee2019change} and \cite{chen2019inference} provided an asymptotic distribution for the suitable centered and normalized break date estimator and constructed a confidence interval for the break date. Since both \cite{bai2010common} and \cite{bhattacharjee2019change} tackled the one-change point model, 	we shall provide a detailed comparison with these two papers in theory and numerical simulations later. It is worth noting that the temporal independence assumption is often assumed in change point analysis for genomic data; see 
	\cite{zhang2010detecting}, \cite{jeng10} and \cite{zhang12} among others.  

A word on notations.  For a vector $a\in R^{p}$, $\|a\|$ is the Euclidean norm. For a matrix $A\in R^{p\times p}$, denote $\|A\|$ and $\|A\|_F$ as the spectral norm and the Frobenious norm respectively, and denote $tr(A)$ as the trace of $A$. Define $\floor{\cdot}$ as the floor function. We use $cum(X_1,X_2,...,X_n)$ to represent the joint cumulants of random variables $X_1,...,X_n$. Throughout the paper, all asymptotic results are stated under the regime $n \wedge p \rightarrow \infty$.
	
The rest of the article is organized as follows. Section~\ref{sec:estimation} presents a new method to estimate $\tau_0$ based on the U-statistic, and contains all the asymptotic results. Section~\ref{sec:CI} introduces several methods of constructing confidence intervals for $\tau_0$, including a bootstrap-based method and its theoretical justification. Section~\ref{sec:simulation} gathers all simulation results. Section~\ref{sec:conclusions} concludes and mentions a few future research topics. All technical proofs are relegated to Sections~\ref{sec:appendix} and \ref{sec:appB}. 
	
	\section{Estimation Method and Asymptotic Theory }
	\label{sec:estimation}
	
Under the one-change point model CP1:  $X_t = \mu_1 {\bf 1}(1\le t\le k_0)+(\mu_1+\delta){\bf 1}(k_0+1\le t\le n) + Z_t,$ for $t = 1,2,...,n$, where $\{Z_t\}_{t = 1}^n$ are $p$-dimensional i.i.d. random vectors with mean $0$ and variance matrix $\Sigma$.  Here we follow the convention in the change point literature and denote the true unknown location of the change point as $k_0 = \tau_0 n$, $\tau_0 \in (0,1)$, i.e., $k_0$ is a fixed positive fraction of the sample size $n$. Without loss of generality, assume $\mu_1 =  0$ as our estimation method is invariant to the value of $\mu_1$. For notational convenience, we shall not use the double-array notation $X_{t,n}$, $Z_{t,n}$, etc. 
	
	Consider the statistic $G_n(k)$ such that for all $k = 2,3,...,n-2$,
	$$G_n(k) = \frac{1}{k(n-k)}\sum_{i_1,i_2 = 1, i_1\neq i_2}^{k}\sum_{j_1,j_2 = k+1, j_1 \neq j_2}^{n}(X_{i_1} - X_{j_1})^T(X_{i_2} - X_{j_2}).$$
	We define 
	$$\hat{k}_U = \argmax_{k = 2,...,n-2}G_n(k)$$ 
	as the estimate of the change point location (or break date) $k_0$. This is a natural estimator since $\E[G_n(k)]$ achieves its maximum when $k$ is the true change point location, as shown in the lemma below. We define $\hat{\tau}_U = \hat{k}_U/n$ as the estimate of the relative position $\tau_0$. Let $a_n = n^2\|\delta\|^4/\|\Sigma\|_F^2$, which is the rate of convergence for $\hat{\tau}_U$ to be shown later.  
	
\begin{lemma} \label{lem:max}
	$\E[G_n(k)] = (k-1)(n-k_0)(n-k_0-1)\|\delta\|^2/(n-k)$ when $k \leq k_0$ and $\E[G_n(k)] = (n-k-1)k_0(k_0-1)\|\delta\|^2/k$ when $k \geq k_0$. Hence $\E[G_n(k)]$ achieves its maximum at $k = k_0$.
\end{lemma}

In \cite{bai2010common}, the location of change point is estimated by minimizing 
	a least squares criterion, that is 
	\begin{equation}\label{eq:LS}
	\hat{k}_{LS} = \argmin_{k = 1,2,...,n-1}SSR(k) := \sum_{i = 1}^{k}\|X_i - \bar{X}_{1:k}\|^2 + \sum_{i = k+1}^{n}\|X_i - \bar{X}_{(k+1):n}\|^2,
	\end{equation}
	where $\bar{X}_{a:b}$ is the sample average based on the subsample $\{X_{a},X_{a+1},...,X_{b}\}$, for any $1 \leq a \leq b \leq n$. The least squares method is natural in the low-dimensional setting; see \cite{bai1994least}, \cite{bai1997et}, \cite{baiperron98}, among others in either one break or multiple break model with or without covariates. 
	
	In the high-dimensional setting, the use of U-statistic was first initiated by \cite{chenqin10} in the two sample testing for the equality of means. Recently, \cite{wang2019a} extended the U-statistic based approach to high-dimensional change point testing, coupled with the idea of self-normalization [\cite{shao10},  \cite{shaozhang2010}, \cite{shao15}]. In this paper, we further advance the U-statistic based approach to the estimation of change point location in the one change-point model, and our proof techniques are substantially different from that in \cite{bai2010common} and \cite{bhattacharjee2019change} due to the use of a different objective function, and also very different from that in \cite{wang2019a} due to the focus on the asymptotic behavior of the estimator. We shall compare the performance of our location estimator with the least squares based counterpart in theory and simulations later. 
	
	To investigate the asymptotic properties of our location estimator $\hat{\tau}_U$, we introduce the following assumptions. 
	
\begin{assumption}\label{ass}
	\begin{enumerate}[label=(\alph*)]
		\item $tr(\Sigma^4) = o(\|\Sigma\|_F^4)$.\label{ass:tr}
		
		\item There exists a positive constant $C$ independent of $n$ such that 
		$$\sum_{j_1,...j_h = 1}^{p}cum^2(Z_{0,j_1},...Z_{0,j_h}) \leq C\|\Sigma\|_F^h,$$
		for $h = 2,3,4,5,6$. \label{ass:cum}
		
		\item $\sqrt{\delta^T\Sigma\delta} = o\left(\sqrt{n}\|\delta\|^2/{\sqrt{a_n}}\right) = o({\|\Sigma\|_F}/{\sqrt{n}})$. \label{ass:delta}

			\item $a_n = o(n)$ and $\log(n) = o(a_n)$. \label{ass:rate}

	\end{enumerate}
\end{assumption}

\begin{remark}[Discussion of Assumptions]
    \rm Assumption \ref{ass}\ref{ass:tr} and \ref{ass}\ref{ass:cum} are identical to the assumptions used in \cite{wang2019a}, where a U-statistic based approach was developed for change point testing in the high-dimensional setting. These assumptions essentially impose weak dependence among the $p$ components,  which can be verified for AR type correlation or models with banded componentwise dependence, but are violated when the variance-covariance matrix is compound symmetric; see detailed discussion in Remark 3.2 of \cite{wang2019a}.

    Assumption \ref{ass}\ref{ass:delta} guarantees the change point signal dominates the noise in the U-statistic and it is equivalent to $\frac{n\delta^T\Sigma\delta}{\|\Sigma\|_F^2}=o(1)$. In the special case $\Sigma=I_p$, it is reduced to 
    $\|\delta\|^2=o(p/n)$. 
    Assumption \ref{ass}\ref{ass:rate}  defines the particular regime we are considering. Note that when $a_n=o(1)$, the U-statistic based test developed in \cite{wang2019a} delivers trivial power asymptotically; see Theorem 3.5 therein. 
    This suggests that the restriction $\log(n)=o(a_n)$ is almost necessary in order for $\tau_0$ to be consistently estimated; see Theorem~\ref{thm:rate} below. The condition $a_n=o(n)$ represents a  particular regime under which a meaningful asymptotic distribution for the suitably centered and normalized location estimator can be obtained. 
    In the case that $\Sigma=I_p$,  Assumption \ref{ass}\ref{ass:rate} is equivalent to $\{\log(n)p\}^{1/2}/n=o(\|\delta\|^2)$ and $\|\delta\|^2=o((p/n)^{1/2})$
    We offer more discussions about what happens in other regimes later; see Theorem \ref{thm:3regime}.
    
\end{remark}


\begin{theorem}[Rate of Convergence]\label{thm:rate}
	Suppose that Assumptions \ref{ass} hold. Then for any $\epsilon > 0$, there exists $M >0$, such that for  large enough $n$, 
	$$P(\hat{k}_U \in \Omega_n(M)) < \epsilon,$$
	where $\Omega_n(M) = \{k: |k-k_0| > nM/a_n\}$.
\end{theorem}

Theorem \ref{thm:rate} implies that $\hat{\tau}_U$ is a rate-$a_n$ consistent estimator for $\tau_0$. Following the conventional argument in studying the limiting behavior of $M$-estimator [\cite{vanweak}], we reparametrize and define 
$\gamma=a_n(\tau-\tau_0)$. Then 
$$\hat{\gamma}_n:=a_n(\hat{\tau}_U-\tau_0)=\argmin_{\gamma\in R} ~L_n(\gamma;\tau_0)$$
where $L_n(\gamma;\tau_0):=\frac{\sqrt{2}\sqrt{a_n}}{n\|\Sigma\|_F}\{G_n(n\tau_0)-G_n(\lfloor n\tau_0+n\gamma/a_n\rfloor)\}.$

To proceed, we define 
	$$G_n^Z(k) = \frac{1}{k(n-k)}\sum_{i_1,i_2 = 1, i_1\neq i_2}^{k}\sum_{j_1,j_2 = k+1, j_1 \neq j_2}^{n}(Z_{i_1} - Z_{j_1})^T(Z_{i_2} - Z_{j_2})$$
as an analog of $G_n(k)$. Let $\l_{\infty}([-M,M])$ denote the set of essentially bounded measurable functions on $[-M,M]$.

\begin{theorem}\label{thm:weak}
	Under the Assumptions \ref{ass}\ref{ass:tr}-\ref{ass}\ref{ass:cum}, for any {$b_n \rightarrow \infty$ and $b_n = o(n)$},
	$$H_n(\gamma) := \frac{\sqrt{2}\sqrt{b_n}}{n\|\Sigma\|_F}\left\{G_n^Z(\floor{n\tau_0}) - G_n^Z(\floor{n\tau_0 + n\gamma/b_n})\right\}\rightsquigarrow \frac{2\sqrt{2}}{\sqrt{\tau_0(1-\tau_0)}}W^*{(\gamma)}$$
	in $l_{\infty}([-M,M])$ for any fixed $M > 0$, where $W^*(\gamma)$ is a two-sided Brownian motion. That is, when $\gamma < 0$, $W^*(\gamma) = W_1(-\gamma)$ and when $\gamma \geq 0$, $W^*(\gamma) = W_2(\gamma)$, where $W_1,W_2$ are independent standard Brownian motions defined on $[0,+\infty)$.		
\end{theorem}

Theorem \ref{thm:weak} gives a process convergence result for the properly normalized increment of the process $G_n^Z(\cdot)$ around the true change point in a shrinking neighborhood. By directly applying argmax continuous mapping theorem [Theorem 3.2.2, \cite{vanweak}], we can get the following corollary.

\begin{corollary}[Asymptotic Distribution]\label{cor:asymdist}
	Under Assumptions \ref{ass}, we can show that for any $M>0$, 
\[L_n(\gamma;\tau_0)\rightsquigarrow L(\gamma;\tau_0) := \sqrt{2}|\gamma| + \frac{2\sqrt{2}}{\sqrt{\tau_0(1-\tau_0)}}W^*(\gamma)\]
in $l_{\infty}([-M,M])$. Consequently, 
	$$a_n(\hat{\tau}_U - \tau_0)\overset{\mathcal{D}}{\rightarrow} \xi(\tau_0):=\argmin_{\gamma\in(-\infty,\infty)}L(\gamma;\tau_0) $$
\end{corollary}

\begin{remark}[Discussion of $\xi(\tau_0)$]\label{rmk:xi}\rm
     In fact, the distribution of $\xi({\tau_0})$ has been well studied in the literature. According to Proposition 1 in \cite{stryhn1996location}, the probability density function of $\xi{(\tau_0)}$, denoted as $f(t)$, is
    $$f(t) = \frac{3}{2}\tau_0(1-\tau_0)e^{\tau_0(1-\tau_0)|t|}\Phi\left(-\frac{3}{2}\sqrt{\tau_0(1-\tau_0)|t|}\right) - \frac{1}{2}\tau_0(1-\tau_0)\Phi\left(-\frac{1}{2}\sqrt{\tau_0(1-\tau_0)|t|}\right),$$
    where $\Phi(\cdot)$ is the cumulative distribution function of a standard normal random variable. It is straightforward to see that $f(t)$ is symmetric, i.e. $f(t) = f(-t)$, for all $t \in \mathbb{R}$, and the densities of $\xi(\tau_0)$ and $\xi(1-\tau_0)$ are identical. Furthermore $f(t)$ achieves its unique maximum at $t = 0$ and $f(0) = \tau_0(1-\tau_0)/2$. In addition, the tail of the distribution is exponential and $Var(\xi(\tau_0)) \propto (\tau_0(1-\tau_0))^{-2}$. 
    
    To approximate the distribution of $\xi(\tau_0)$ and its critical values, we approximate the standard Brownian motion by standardized sum of i.i.d. standard normal random variables, and generate $10^5$ Monte-Carlo replicates of $\xi(\tau_0)$ for $\tau_0 = 0.01,...,0.99$. Then we plot their densities and critical values over $\tau_0 \in [0.01,0.5]$ in Figure \ref{fig:dist}.
\end{remark}
\[\text{Please insert Figure~\ref{fig:dist} here!}\]

    From Figure \ref{fig:dist}, we see that as $\tau_0$ moves from $0.5$ to $0.1$, the density is less concentrated around $0$, indicating the relative difficulty of accurately estimating $\tau_0$ when $\tau_0$ is close to $0$ or $1$. Correspondingly, the critical calues increase as $\tau_0$ goes from $0.5$ to $0.01$.
    
    If Assumptions \ref{ass}\ref{ass:tr}, \ref{ass}\ref{ass:cum} and \ref{ass}\ref{ass:delta} hold, there are indeed three regimes that correspond to different rates of $a_n$. Given $a_n/\log(n)\rightarrow \infty$, if $a_n/n \rightarrow 0$, this situation is covered by Corollary \ref{cor:asymdist}. There are two more regimes, under which the behavior of our estimator is discussed in the following theorem.
\begin{theorem}\label{thm:3regime}Under Assumptions \ref{ass}\ref{ass:tr} and \ref{ass}\ref{ass:cum},
 \begin{enumerate}[label = (\alph*)] 
        \item $a_n/n \rightarrow c\in (0,\infty)$: if Assumption \ref{ass}\ref{ass:delta} also holds, our change point location estimator still works in the sense that
        $$\hat{k}_U - k_0 \overset{\mathcal{D}}{\rightarrow} \argmin_{\gamma \in \mathbb{Z}}L(\gamma,\tau_0)$$
        and $(\hat{k}_U - k_0) = O_p(1)$. 
        \item $a_n/n \rightarrow \infty$: if $\sqrt{\delta^T\Sigma\delta} = o(\|\delta\|^2)$, then we have $P(\hat{k}_U \neq k_0) \rightarrow 0$.
    \end{enumerate}
\end{theorem}


We shall offer some comparison with the methods, theory and assumptions in \cite{bai2010common} and \cite{bhattacharjee2019change}, as the latter two papers both addressed change point estimation in the one change point model. To elaborate the differences, we shall separate our discussions into several categories as follows. 

{\bf 1}: Model assumptions and estimation methods. Although all three papers assumed one change point in mean for a sequence of  high-dimensional observations, there are substantial differences. In particular, \cite{bai2010common} assumed componentise independence (or so-called cross-sectional independence) but allowed weak temporal dependence for each series; \cite{bhattacharjee2019change} relaxed the cross-sectional independence assumption in \cite{bai2010common} and allowed weak dependence over time and also within components. By contrast, we require the data to be independent over time but allow for weak componentwise dependence. This makes a direct comparison of the three papers quite challenging, so we shall focus on some specific cases only. Note that in both \cite{bai2010common} and \cite{bhattacharjee2019change}, an infinite order vector moving average process (i.e., VMA($\infty$))
was assumed.  

Denote the break date estimator in \cite{bai2010common} as $\hat{\tau}_{LS} = \hat{k}_{LS}/n$, which is obtained as the minimizer of a  least squares criterion, see (\ref{eq:LS}). Let $\hat{\tau}_{BBM}$ denote the estimator used in \cite{bhattacharjee2019change} where the minimum is taken over $k \in [\floor{c^*n},\floor{(1-c^*)n}]$ for some prespecified $c^* \in (0,0.5)$. It should be expected that $P(\hat{\tau}_{LS}\not=\hat{\tau}_{BBM})\rightarrow 0$ if $\tau_0\in (c^*, 1-c^*)$. By contrast, our break date estimator $\hat{\tau}_U$ is the maximizer of a new U-statistic based objective function.

{\bf 2}: Asymptotic framework, regimes and technical assumptions. In  \cite{bai2010common}, he studied two asymptotic frameworks, $n$ fixed and $n\rightarrow\infty$ as $p\rightarrow\infty$ (\cite{bai2010common} used $N$ for our $p$, and $T$ for our $n$ in his paper). We shall only focus on a comparison with his result in the latter case, i.e., $\min(n,p)\rightarrow\infty$. 
To make a fair comparison, we shall discuss the case for which both theories are expected to work, which is the case of cross-sectional and temporal independence. To facilitate the comparison, we further assume   $\Sigma = I_p$. Under this condition, our Assumptions \ref{ass}\ref{ass:tr} and \ref{ass}\ref{ass:cum} are automatically satisfied.  
    
    Under the assumption that $\|\delta\|^2 \rightarrow \infty$, Theorem 3.2 in \cite{bai2010common} stated that if $\log(\log(n))p/n \rightarrow 0$, then $P(\hat{k}_{LS} \neq k_0) \rightarrow 0$. This corresponds to our third regime, where $a_n/n \rightarrow \infty$ and $\sqrt{\delta^T\Sigma\delta} = \|\delta\| = o(\|\delta\|^2)$, which implies that $P(\hat{k}_U \neq k_0) \rightarrow 0$ as well, see Theorem \ref{thm:3regime}.  Under the assumption that $\|\delta\|^2 \rightarrow C \in (0,\infty)$, $\log(\log(n))p/n \rightarrow 0$ and $\delta^T\Sigma\delta \rightarrow C'\in(0,\infty)$,  Theorem 4.2 of \cite{bai2010common} asserted the asymptotic distribution for $\hat{k}_{LS}-k_0$. These assumptions imply $a_n/n \rightarrow \infty$, however $\sqrt{\delta^T\Sigma\delta}$ is no longer $o(\|\delta\|^2)$. This does not belong to any one of our three regimes stated early. But interestingly, under this specific setting our estimator $\hat{k}_U$ converge to the same distribution as $\hat{k}_{LS}$, and we shall prove this result below in Proposition~\ref{prop:compareBai}.

    In addition to these two cases, our theory also uncovers an important and interesting regime \cite{bai2010common} did not consider, that is $a_n=o(n)$. In this case, as we showed earlier, there is a very nice interplay between the order of $\|\delta\|^2$ and $(n,p)$ that allows $p$ to diverge faster than $n$, such that there is an asymptotic distribution for $a_n(\hat{\tau}_U-\tau_0)$. Thus in a sense, our theory is more complete than the one in \cite{bai2010common} for this specific model. 
    


\cite{bhattacharjee2019change} extended the method and theory in  \cite{bai2010common} to accommodate both cross-sectional and temporal (serial) dependence. They also used  the $VMA(\infty)$ model with a mean shift, but to accommodate the cross-sectional dependence,  many additional assumptions were imposed. For example, \cite{bhattacharjee2019change} required finite fourth moment, whereas \cite{bai2010common} did not. Also they required that the number of nonzero elements in $\delta$ cannot vary with $n$; see assumption (A4) in their paper. Such requirement is not needed in our technical analysis. Different from the conditions used in \cite{bai2010common}, there is no explicit restriction on the relative relationship between $p$ and $n$ in \cite{bhattacharjee2019change}. See Remark 2.10 for additional explanations.


{\bf 3}: Convergence rates and efficiency comparison. To compare with the theory in \cite{bhattacharjee2019change}, we shall focus on the following model,
$$X_t = \delta\1{t > k_0} + A\epsilon_t, \qquad t = 1,2,...,n,$$
where $\{\epsilon_t\}$ are i.i.d. $p$-dimensional random vectors with zero mean and identity covariance matrix,
$A$ is $p \times p$ real-valued matrix and $\delta \in \mathbb{R}^p$ is the vector of mean change. In our setting, we can set $A = \Sigma^{1/2}$, and hence $\Sigma = AA^T$. 
This model allows cross-sectional dependence but enforces temporal independence, so is included in our framework.
It can also be viewed as a special case of the VMA$(\infty)$ model with a mean shift in \cite{bhattacharjee2019change}
as basically we let $A_j=A$ when $j=0$ and $A_j=0$ for $j\ge 1$ in their VMA representation. According to Theorem 2.1 and Theorem 2.2 of 
\cite{bhattacharjee2019change}, to guarantee the consistency of $\hat{\tau}_{BBM}$, the signal-to-noise ratio (SNR) has to grow to infinity, i.e. 
$SNR = \frac{n\|\delta\|^2}{p\|A\|^2} \rightarrow \infty.$ Under this condition, $\hat{\tau}_{BBM}$ is consistent with the rate of convergence $a_n^{(BBM)} = p \times SNR = n\|\delta\|^2/\|A\|^2$. Notice that the rate of convergence for $\hat{\tau}_U$ is $a_n = n^2\|\delta\|^4/\|\Sigma\|_F^2$. When $SNR \rightarrow \infty$,
\begin{align*}
    \frac{a_n}{a_n^{(BBM)}} = \frac{n^2\|\delta\|^4}{\|\Sigma\|_F^2}\cdot \frac{\|A\|^2}{n\|\delta\|^2} = \frac{n\|\delta\|^2}{p\|A\|^2}\cdot\frac{p\|A\|^4}{\|\Sigma\|_F^2} \geq SNR \cdot \frac{p\|A\|^4}{p\|AA^T\|^2} \geq SNR \rightarrow \infty,
\end{align*}
where the first inequality is due to the fact that $\|\Sigma\|_F^2 \leq p\|\Sigma\|^2$, and the second inequality in the above derivation is because $\|AA^T\| \leq \|A\|^2$. This is a significant finding as it means that for the above specific model, if assumptions for both methods are satisfied, the convergence rate corresponding to our U-statistic based estimator is faster.

To ensure the consistency of our break point estimator $\hat{\tau}_U$,  we also require a signal-to-noise condition, i.e., $a_n/\log(n) \rightarrow \infty$. Since $a_n \geq a_n^{(BBM)}SNR = p\cdot SNR^2$ as shown in the above display,  $a_n/\log(n) \rightarrow \infty$ provided that (a) $p = O(\log(n))$ and $SNR\rightarrow\infty$ or (b) $\log (n)=o(p)$ and $SNR$ is fixed. This implies that our U-statistic based estimator  can be consistent with suitable convergence rate under a weaker signal setting as compared to the  least squares counterpart. In other words, $\hat{\tau}_U$ is a consistent estimator of $\tau_0$ under much weaker conditions than $\hat{\tau}_{BBM}$. We shall provide some theoretical explanation for this phenomenon in Remark~\ref{rem:explain}. 

Furthermore, similar to our Theorem \ref{thm:3regime} where we have described two additional regimes according to different orders of $(a_n/n)$, the asymptotic behavior of $\hat{\tau}_{BBM}$ also has three regimes depending on the order of $(a_n^{(BBM)}/n)$. Under certain assumptions, if $a_n^{(BBM)}/n \rightarrow \infty$, the asymptotic distribution of $\hat{k}_{BBM}-k_0$ is degenerate at zero. In this case, since $a_n/a_n^{(BBM)} \geq SNR \rightarrow \infty$, the limiting distribution of $\hat{k}_U-k_0$ is also degenerate at zero, as $a_n^{(BBM)}/n \rightarrow \infty$ implies that $a_n/n \rightarrow \infty$ under the assumption that $SNR\rightarrow\infty$.  This suggests that the regime that corresponds to the degenerate limiting distribution for $\hat{k}_{BBM}$ is well included in  the regime that corresponds to degenerate limiting distribution for $\hat{k}_U$. 

 Of course, the results in \cite{bhattacharjee2019change} are generally applicable to the temporal dependent case, so the slower convergence rate relative to our estimator, which is tailed to the independent high-dimensional data, is probably not surprising. 
Nevertheless, it shows that our new location estimator can bring substantial efficiency gain relative to the least squares based counterpart in the case of independent high-dimensional data. 

 \begin{proposition}\label{prop:compareBai}
    Both estimators $\hat{k}_{LS}$ and $\hat{k}_U$ converge to the same limiting distribution if
    \begin{enumerate}[label = (\alph*)]
        \item $\Sigma = I_p$,
        \item $\|\delta\|^2 \rightarrow c \in (0,\infty)$
        \item $\log(\log(n))p/n \rightarrow 0$.
    \end{enumerate}
\end{proposition}
\begin{remark}
\label{rem:explain}
{\rm 
    The main reason why \cite{bai2010common} only provided theories under the setting $p = o(n)$ is that the objective function he used is least squares based and it contains extra diagonal terms that need to be controlled under certain restriction on the growth rate of $p$ as a function of $n$. To be specific, as we showed in the proof of Proposition~\ref{prop:compareBai}, the first three terms of $SSR(k)-SSR(k_0)$ (i.e., $I_j$, $j=1,2,3$) are of form $\sum_{i = a}^{b}Z_i^TZ_i$ (up to a multiplication constant) for some $1 \leq a < b \leq n$ that needs to be of smaller magnitude than the leading terms in his theoretical derivation, so these three terms have to be controlled under the assumption $p=o(n)$ since if $n=o(p)$ the three diagonal terms can dominate the others. In contrast, in our U-statistic based objective function, we essentially remove the diagonal terms, so no growth rate assumption as $p=o(n)$ is required and we are able to cover the  ``large $p$ small $n$'' case automatically.
    
    The advantage of U-statistic over the least squares counterpart was in fact stated in \cite{chenqin10} under the two sample testing framework. Compared to an important early paper by \cite{bai1996effect}, which involves a least squares term in the test statistic, \cite{chenqin10} used a U-statistic to remove the diagonal terms, which are not useful in the testing and incur unnecessary 
    growth rate constraints in the theoretical analysis. Therefore in a sense, our U-statistic based approach inherits this advantage from \cite{chenqin10} and allows our theory to cover the interesting "large $p$ small $n$" case (i.e., $p>>n$). 

   It is worth noting that \cite{bhattacharjee2019change} considered almost the same estimator as \cite{bai2010common} but extended the theory to a more general setting including the large $p$ small $n$ case. In \cite{bhattacharjee2019change}, the control of diagonal terms in the least squares based objective function was explicitly done by imposing some extra conditions. In particular their condition (A4) was used to control the order of the diagonal terms. However, these additional assumptions seem hard to verify in practice.  In comparison, our four assumptions are relatively more transparent and interpretable. 
  } 

\end{remark}

\section{Confidence interval construction}\label{sec:CI}

Given the asymptotic theory presented in Section~\ref{sec:estimation}, we shall first describe a way of constructing a confidence interval for $\tau$ based on asymptotic approximation. Note that the normalizing constant $a_n$ depends on 
two unknown quantities, $\Delta = \|\delta\|^2$ and $\|\Sigma\|_F^2$. Fortunately, their consistent estimators have been provided by \cite{chenqin10} in the two sample testing context, and we can easily adapt them to our setting. 
Algorithm \ref{alg:CI} describes the procedure for the plug-in approach below. 
\begin{algorithm}[h!]
	\begin{enumerate}
		\item Estimate $k_0$: $\hat{k}_U = \argmax_{k}G_n(k)$.
		\item Estimate $\Delta = \|\delta\|^2$ : $\widehat{\Delta} = \frac{1}{(\hat{k}_U-1)(n-\hat{k}_U-1)}G_n(\hat{k}_U)$.
		\item Estimate $\|\Sigma\|_F^2$ : $\widehat{\|\Sigma\|_F^2}(\hat{k}_U)$.
		\item Estimate $a_n$: $\hat{a}_n = n^2\widehat{\Delta}^2/\widehat{\|\Sigma\|_F^2}(\hat{k}_U)$.
		\item $(1-\alpha)$ confidence interval for $\tau_0$: $[\hat{\tau}_U-q_{1-\alpha/2}(\xi(\hat{\tau}_U))/\hat{a}_n, \hat{\tau}_U- q_{\alpha/2}(\xi(\hat{\tau}_U))/\hat{a}_n]$, where $q_{\alpha}(\xi({\tau}))$ denotes the $\alpha$-quantile of the distribution of $\xi({\tau})$.
	\end{enumerate}
	\caption{Algorithm for constructing a confidence interval for $\tau_0$}\label{alg:CI}
\end{algorithm}

In the Algorithm \ref{alg:CI}, we have used the jackknife type estimator $\widehat{\|\Sigma\|_F^2}(\hat{k}_U)$ introduced in Chen and Qin (2010),
to estimate $\|\Sigma\|_F^2$, i.e., 
\begin{align*}
\widehat{\|\Sigma\|_F^2}(\hat{k}_U) =  (\hat{k}_U(n -\hat{k}_U))^{-1}tr &\left\{\sum_{i = 1}^{\hat{k}_U}(X_{i} - \bar{X}_{(1:\hat{k}_U,i)})(X_{i} - \bar{X}_{(1:\hat{k}_U,i)})^T\right.\\
&\left.\cdot\sum_{j = \hat{k}_U+1}^{n}(X_{j} - \bar{X}_{(\hat{k}_U+1 : n, j)})(X_{j} - \bar{X}_{(\hat{k}_U+1 : n, j)})^T\right\},
\end{align*}
where $\bar{X}_{(a:b,i)}$ is the sample average of $X_a,...,X_{b}$ excluding $X_{i}$.This slightly differs from the one used in \cite{chenqin10} in that we removed two terms that correspond to the two double sums within the pre-break sample and post-break sample. Simulation suggests that there is little impact on the coverage and length of intervals.

After preliminary simulations, we realize that there is considerable amount of coverage error for the above plug-in 
based confidence interval since this only covers the regime described in Corollary \ref{cor:asymdist}. There can be other regimes which have different convergence rates and in reality we may not be able to know which regime the data generating process falls into. This motivates us to propose the following bootstrap-based interval. 

\begin{algorithm}[h!]
	\begin{enumerate}
		\item  Estimate $k_0$ by $\hat{k}_U = \argmax_kG_n(k)$ and $\hat{\tau}_U = \hat{k}_U/n$.
		
		\item Estimate $\Delta$ by $\widehat{\Delta} = \frac{1}{(\hat{k}_U-1)(n - \hat{k}_U-1)}G_n(\hat{k}_U)$, and let $\hat{\delta} = \boldsymbol{1}_p\sqrt{\widehat{\Delta}/p}$,where $\boldsymbol{1}_p$ is a $p$-dimensional vector with all elements equal to 1.
		
		\item Estimate $\Sigma$ by some positive semi-definite estimator $\hat{\Sigma}_X$.
		
		\item Generate random vectors $\epsilon_1$,...,$\epsilon_n$ in $\mathbb{R}^p$ from the distribution $\mathcal{N}(0,\hat{\Sigma}_X)$.

		\item Generate ${X}_t^* = \epsilon_t$ if $t \leq \hat{k}_U$ and ${X}_t^* = \hat{\delta} + \epsilon_t$ if $t > \hat{k}_U$.
		
		\item Calculate the bootstrap estimate $\hat{k}_U^*$ by $\hat{k}_U^* = \argmax_{k = 2,...,n-2}G_n^{(X^*)}(k),$ where $G_n^{(X^*)}(k)$ denotes the value of $G_n(k)$ calculated based on $\{X_t^*\}$, and calculate the bootstrap estimate of the proportion by $\hat{\tau}_U^* = \hat{k}_U^*/n$.
		
		\item Repeat step 4-6 for $B$ times to generate $\hat{\tau}^*_{U,1}$,...,$\hat{\tau}^*_{U,B}$, and 95\%  bootstrap CI for $\tau_0$ is $[\hat{\tau}_U-q^*_{0.975}, \hat{\tau}_U-q^*_{0.025}]$, where $q^*_{0.025}$ and $q^*_{0.975}$ are the sample $2.5\%$ and $97.5\%$ quantiles for $\{\hat{\tau}^*_{U,i} - \hat{\tau}_U\}_{i = 1}^{B}$.
	\end{enumerate}
	\caption{Bootstrap algorithm for constructing a confidence interval for $\tau_0$ }\label{alg:CI_para}
\end{algorithm}

In Algorithm \ref{alg:CI_para} we consider to use a uniform vector with squared norm equal to $\widehat{\Delta}$ to estimate the mean change vector, regardless of the sparsity of the truth $\delta$. The reason why this works is because the limiting distribution only depends on the norm of the mean change, not the vector of the mean change itself. To verify this we have also tried variants of the above algorithm by imposing different sparsity on $\hat{\delta}$ while maintaining the same norm, the finite sample performance turns out to be stable.  


\begin{theorem}[Bootstrap Consistency]\label{thm:boots}
Suppose Assumption \ref{ass} holds. Further, we assume that
\begin{enumerate}[label = (\alph*)]
    \item $tr({\hat{\Sigma}_X}^4)/\|\hat{\Sigma}_X\|_F^4 = o_p(1)$,
    \item $\|\hat{\Sigma}_X\| = o_p(\max(n{\|\delta\|}^2/a_n,\|\delta\|^2))$,
    \item $\hat{a}_n/a_n \rightarrow_p 1$.
\end{enumerate}
Given the data, the conditional distribution of $\{a_n(\hat{\tau}_U^* - \hat{\tau}_U)\}$ weakly converges to that of $\xi(\tau_0)$ in probability, i.e.
$$a_n(\hat{\tau}_U^* - \hat{\tau}_U) \overset{\mathcal{D}}{\rightarrow} \xi(\tau_0) \text{ in P}.$$
Thus the bootstrap interval described in Algorithm \ref{alg:CI_para} has desired coverage probability asymptotically, i.e.
$$P(\tau_0 \in [\hat{\tau}_U - q^*_{1-\alpha/2},\hat{\tau}_U - q^*_{\alpha/2}]) \rightarrow 1-\alpha.$$
\end{theorem}

\begin{remark}{\rm
    For the other two regimes besides the regime covered by Assumption \ref{ass}, the bootstrap estimator $\hat{\tau}_U^* - \hat{\tau}_U$ has the same asymptotic behavior as $\hat{\tau}_U - \tau_0$ if the three conditions in the above theorem are satisfied. The proof is very similar  to the proof of Theorem \ref{thm:boots} so we skip the details. The verification of the two conditions (a) and (b) in Theorem \ref{thm:boots} depend on what type of positive definite estimator $\hat{\Sigma}_X$ we adopt, and it requires a case-by-case analysis. Hence details are omitted. } 
\end{remark}

\begin{remark}
{\rm \cite{bhattacharjee2019change} required a stronger signal-to-noise condition for the bootstrap consistency result. Specifically, instead of the $SNR \rightarrow \infty$, they required $SNR\text{-}ADAP = \frac{n\|\delta\|^2}{p\log(p)\|A\|^2} \rightarrow \infty$. By contrast, the requirement for the signal-to-noise ratio in our bootstrap consistency result is identical to that for the consistency of our original estimator [cf. Theorem \ref{thm:weak}], which is $a_n/\log(n) \rightarrow \infty$. 
}
\end{remark}

\section{Simulation studies}
\label{sec:simulation}

In this section, we study the finite sample performance of our proposed estimator and confidence intervals. We consider the single change point model (CP1) with the change point located at $n\tau_0$, i.e.
\begin{equation}
X_t = \delta\1{t > n\tau_0} + \epsilon_t, \qquad t = 1,2,...,n,
\end{equation}
where $\{\epsilon_t\}$ are i.i.d. multivariate normal random vectors with mean zero and covariance $\Sigma$,  $\delta$ is the mean change vector, and $\tau_0 = 0.2$ or $0.5$. To study the impact of componentwise dependence on the performance, we include four different models for $\Sigma$: (1) Identity (ID: $\Sigma = I_p$); (2) AR(1) (AR: $\Sigma(i,j) = 0.8^{|i-j|}$); (3) Banded (BD: $\Sigma(i.j) = 0.5^{|i-j|}\1{|i-j| <= 2}$); (4) Compound Symmetric (CS: $\Sigma = 0.5I_p + 0.5\boldsymbol{1}\boldsymbol{1}^T$). 

The sample size $n$ is chosen from $\{50,100,200\}$ and the dimension $p$ is chosen from $\{50,150\}$. Furthermore we consider two cases for the sparsity of $\delta$. One is dense change where $\delta$ is formed by $p$ i.i.d. random values generated from Uniform distribution $Unif[-0.5,0.5]$. The other is sparse change, where we first generate a $p$-dimensional random vector by the same procedure as what we have done for dense change and record its norm as $\|\delta\|$, and then generate $\delta = \|\delta\|(1/\sqrt{5}, 1/\sqrt{5}, 1/\sqrt{5}, 1/\sqrt{5}, 1/\sqrt{5},0,...,0)^T$ as the sparse vector. We fix $\delta$ for all Monte-Carlo replicates with the same $(n,p)$ combination. 

We conduct two simulation studies under the above settings to evaluate the performance of the point estimators  and confidence intervals, and we comment on the results below for these two studies respectively.

\subsection{Finite sample performance of location estimators}

We examine the finite sample performance of the location estimators, including our U-statistic based estimator (denoted as $\hat{\tau}_U$) and the least squares based estimator described in \cite{bai2010common} (denoted as $\hat{\tau}_{LS}$) by $20000$ Monte-Carlo replicates. The bias, variance and  mean squared error (MSE) are summarized in Table \ref{tab:point0.2} for $\tau_{0} = 0.2$ and in  Table \ref{tab:point0.5} for $\tau_0 = 0.5$. As we can observe, $\hat{\tau}_U$ outperforms $\hat{\tau}_{LS}$ for almost all settings in terms of the MSE. There are two settings for $\tau_0 = 0.2$ with banded covariance structure ($(n,p) = (200,50)$ for sparse change and $(n,p) = (100,150)$ for dense change), where $\hat{\tau}_U$ has slightly larger MSE, and this could be due to random Monte Carlo errors.
As we break the MSE criterion into (squared) bias and variance, we spot an interesting pattern. The biases for $\hat{\tau}_U$ and $\hat{\tau}_{LS}$ are mostly comparable, with no one dominating the other. 
The advantage of $\hat{\tau}_U$ in MSE is mostly attributed to its smaller variance, which may be explained by the usage of a U-statistic based objective function, as U-statistic has the well-known minimal variance property in estimation.

It can also be seen that  when comparing the results for $\tau_0=0.2$ and $\tau_0=0.5$,  there are substantially smaller bias and variance for all settings when $\tau_0=0.5$, which is consistent with our intuition that  estimation is easier when $\tau_0 = 0.5$. Additionally, both methods exhibit a larger MSE for the compound symmetric case, as compared to other covariance structures. This is not surprising since the compound symmetric covariance matrix corresponds to  strong componentwise dependence and violates the weak cross-sectional dependence  assumptions necessary for both methods. Nevertheless,  as sample size gets  larger, the MSE gets smaller for all cases. A direct comparison between  the dense change and the sparse charge shows that the results for both cases are very similar for all combinations of $(n,p)$ and models. This is quite reasonable since the performance of both methods essentially depends on the $l_2$ norm of the mean change, which we hold at the same level. Thus the sparsity of the mean change is not the critical factor in determining the finite sample performance of both estimators. Overall, our new estimator enjoys the efficiency gain over the least squares counterpart in almost all settings and should be  preferred  in the high-dimensional environment.

\subsection{Finite sample performance of confidence intervals}
In this section we evaluate the finite sample performance for confidence intervals. For each setting, we generate 3000 Monte-Carlo replicates and construct 7 different 95\% confidence intervals for $\tau_0$ including:

\begin{enumerate}
\item Oracle U-statistic based CI ($U_1$): Constructed by Algorithm \ref{alg:CI} with the true value for $a_n$ replacing $\hat{a}_n$.

\item U-statistic based CI ($U_2$): Constructed by Algorithm \ref{alg:CI}.

\item U-statistic based parametric bootstrap I ($U_3$): Constructed by modifying Algorithm \ref{alg:CI_BBM}. We estimate the location by $\hat{\tau}_U$, and follow steps 2-5 to generate bootstrap samples. For each bootstrap sample we used our U-statistic based method instead of least squares based method to estimate the location. 

\item U-statistic based parametric bootstrap II ($U_4$): Constructed by Algorithm \ref{alg:CI_para}.

\item U-statistic based nonparametric bootstrap ($U_5$):  We sample with replacement based on the pre-break sample and post-break sample separately to generate bootstrap data, where the break date is estimated by $\hat{\tau}_U$. The break date for the bootstrap sample is estimated by our method.

\item CI in \cite{bai2010common} ($LS_1$): Constructed by Algorithm \ref{alg:CI_Bai}.
\item Adaptive CI in \cite{bhattacharjee2019change} ($LS_2$): Constructed by Algorithm \ref{alg:CI_BBM}, which is a modified version of the one in \cite{bhattacharjee2019change} due to the temporal independence model we assume here. This main difference between this one and $U_3$ is that the least squares based approach was used for the break date estimation for both original and bootstrap sample, whereas the U-statistic based approach was used in the construction of $U_3$. 

\end{enumerate}

One thing worth pointing out is that in Algorithm \ref{alg:CI_BBM}, \cite{bhattacharjee2019change}  originally used banded autocovariance matrix to generate bootstrap samples. However we found during our simulation that the banded covariance matrix may not be positive semi-definite and the sample covariance matrix  itself is not a good estimate when the dimension is high. To solve this issue,  we use the R package "PDSCE" to get a positive definite estimate for the high-dimensional covariance matrix.
 Specifically, denote $S$ as the sample covariance matrix and $R$ as the sample correlation matrix. Denote $S^+$ as the diagonal matrix with the same diagonal as $S$ and $S^- = S - S^+$.  Then the correlation matrix estimator is constructed as 
$$\hat{\Theta} = \argmin_{\Theta\succ 0}(\|\Theta - R\|_F^2/2 - \lambda_1\log|\Theta| + \lambda_2|\Theta^-|_1),$$
where $\lambda_1$ is a fixed small positive constant, $\lambda_2$ is a non-negative tuning parameter,  $|\cdot|$ is the determinant and $|\cdot|_1$ is the $l_1$ norm of the vectorized matrix. The tuning parameters are selected via  a default cross-validation step in  ``PDSCE''. Then the estimated covariance matrix is constructed as $\hat{\Sigma} = (S^+)^{1/2}\hat{\Theta}(S^+)^{1/2}$. See \cite{rothman2012positive} for more details about this methodology. Note that this is the covariance matrix estimate we used for our Algorithms \ref{alg:CI_para} and  \ref{alg:CI_BBM}, i.e., in the construction of $U_3, U_4$ and $LS_2$. 


The results are summarized in Tables \ref{tab:CI_ID}-\ref{tab:CI_CS}. We calculate the sample coverage probability as well as the average length for each CI. Each table corresponds to a particular covariance model with all results for both sparse and dense changes, all combinations of $(n,p)$s and two cases  $\tau_0=0.2,0.5$. It is apparent for some combinations of $(n,p)$, the coverages for all seven intervals are far below the nominal coverage level $95\%$, indicating the difficulty of constructing an interval with proper coverage, especially when $n$ is small and the dependence among components is strong.

The three intervals ($U_1$, $U_2$ and $LS_1$) are based on asymptotic approximation, which seem quite coarse as all these intervals exhibit serious undercoverage. It appears that in most cases $U_1$ and $U_2$ have better coverage than $LS_1$ when $p>n$ but have worse coverage than $LS_1$ when $p<n$, which is consistent with the asymptotic theory. Note that the undercoverage for 
$U_1$ and $U_2$ are tied to the fact that  we always use $a_n$ for $U_1$ and $\hat{a}_n$ for $U_2$ as the rate of convergence, which is only correct for our regime $a_n/n \rightarrow 0$. However for other regimes the rates of convergence can be  slower than $a_n$. This leads to an undercoverage. For the same reason, $LS_1$ fails to achieve the desired coverage probability since the theory is only valid for a specific regime and requires the cross-sectional independence.

The four bootstrap-based intervals ($U_3$, $U_4$, $U_5$ and $LS_2$) have overall almost uniform better coverages than the three counterparts based on asymptotic approximation. Among these four intervals, the ranking appears to be (in the order of preferences) $U_4>U_5>U_3>LS_2$. 
We can see that $U_4$ and $U_5$ have very comparable results for all settings. Both have about 95\% coverage probability even for small $(n,p)$ when $\tau_0 = 0.5$. When $\tau_0 = 0.2$, the problem gets harder but they can still have a desired coverage when $n$ and $p$ are large, sometimes even too conservative. We have no theoretical justification for the nonparametric bootstrap procedure yet, but the simulation results are quite encouraging.  We have also tried variants of $U_4$ by using a sparse estimate of $\delta$ and keeping the same $l_2$ norm, the results turn out to be  similar to what we have here. 

As a comparison, $LS_2$ is another least squares based interval and it has better coverage comparing to $LS_1$ by using a bootstrap procedure. But it still cannot achieve the desired confidence level  for most settings when $\tau_0 = 0.2$. When $\tau_0 = 0.5$, it has decent coverage probability when the sample size is large. As we have discussed before,  the rate of convergence of our estimator is faster than the least squares based estimator. Hence our methods ($U_3$, $U_4$ and $U_5$) can achieve the desired coverage with smaller $(n,p)$. Furthermore we observe that for large $p$ small $n$ situation, our methods still provide a good coverage whereas $LS_2$ cannot. This may be related to the fact that our methods work for both $n/p \rightarrow \infty$ and $n/p \rightarrow 0$ in theory, but the least squares based method only works for the case $p<<n$.  Note that a  higher coverage probability is usually associated with a longer interval.


Among other observations, we mention that for the same setting, the results for $\tau_0 = 0.5$ are always comparable or better than the results (in terms of more coverage and shorter interval length) for $\tau_0 = 0.2$ due to the fact  the estimation problem for $\tau_0 = 0.5$ is easier. For most settings, the results for sparse and dense changes are similar, which is consistent with the fact that the performance is mainly determined by the $l_2$ norm of the mean change rather than the mean change vector itself.
For the same $\tau_0$, there are relatively less differences between the results for ``ID'', ``AR'' and ``BD'' covariance models, compared to their differences from the "CS" case. This can be explained by the fact that the former three cases belong to the class of weakly dependent components model whereas the compound symmetric covariance structure implies strong dependence.



 Among all methods, the U-statistic based bootstrap procedures ($U_4$, $U_5$) perform the best,  achieving the desired coverage level in most settings even with a small sample size and moderate dependence within components.

 \subsection{Impact of $\delta^T\Sigma\delta$ on the finite sample coverage}
 
As shown in the previous subsection,  the confidence intervals constructed by the asymptotic approximation have substantially lower coverage probabilities than the desired 95\%. One possible explanation is that we used the limiting distribution corresponding to a particular regime to approximate the finite sample distribution of $\hat{\tau}_U-\tau_0$ regardless of which regime the data generating process falls into, which yields large approximation errors in many cases. Another plausible explanation is that 
the finite sample approximation error is related to the magnitude of $\delta^T\Sigma\delta$, which controls the amount of noise in the U-statistic based objective function. This can be seen from our theoretical derivation, as our objective function contains   terms as $\delta^TZ_i$, for $i = 1,2,...,n$. These terms are asymptotically negligible under Assumption \ref{ass}\ref{ass:delta}, but in finite sample, these interaction terms can have a substantial impact on the finite sample coverage. 
Theoretically the order of these interaction terms is proportional to $\delta^T\Sigma\delta=Var(\delta^TZ_i)$. To examine the impact of these interaction terms, as measured by the magnitude of $\delta^T\Sigma\delta$, we shall  design a small simulation experiment as follows. 
 
In the following experiment,  the sample size $n$ is still chosen from $\{50,100,200\}$ and $p$ is selected from $\{50,150\}$. The change occurs at $\tau_0 = 0.2$. We set $\Sigma$ as a diagonal matrix with elements $\Sigma(i,i) = 0.1$ if $i \leq p/2$ and $\Sigma(i,i) = 1$ for $i > p/2$.  We consider three cases for $\delta$ to represent different strength of the interaction terms. We fix $\|\delta\|^2 = 4$, and  set (1) Weak:  $\delta \propto (\boldsymbol{1}_{p/2}^T,\boldsymbol{0}_{p/2}^T)^T$;(2) Moderate:  $\delta = \boldsymbol{1}_p$;(3) Strong: $\delta \propto (\boldsymbol{0}_{p/2}^T,\boldsymbol{1}_{p/2}^T)^T$. It is easily seen that the magnitude of $\delta^T\Sigma\delta$  gradually increases as we move from case (1) to (2) and to (3). The results are summarized in Table \ref{tab:CI_int}.
 
 As we fix $\|\delta\|^2 = 4$, the signal of the problem is fixed. When we increase the strength of the interaction as quantified by $\delta^T\Sigma\delta$, we essentially increase the level of the noise, so the (finite sample) signal to noise ratio decreases. Consequently, for a fixed sample size and dimension combination,  the coverage probabilities for all methods decrease as the interaction gets stronger.
 For all intervals based on asymptotic approximations ($U_1$, $U_2$ and $LS_1$), it is interesting to observe that while the average length does not change much, the coverage drops as the strength of interaction terms moves from weak to moderate and then to strong. When $\delta^T\Sigma\delta$ is small, we see that $U_1$ and $U_2$ can indeed achieve a coverage of more than $90\%$ for $p=50, 150$, when $n=100$ and $200$, which corroborates our asymptotic theory. There is some noticeable impact on the coverage of bootstrap-based intervals ($U_3$, $U_4$, $U_5$ and $LS_2$) but 
 compared to the impact on asymptotic approximation based intervals ($U_1$, $U_2$ and $LS_1$), 
the strength of interaction terms plays a less significant role in the finite sample coverage. This might be due to the adaptive nature of the bootstrap method. A good theoretical explanation for this adaptiveness presumably involves second-order edgeworth expansion of the distribution of both $\hat{\tau}_U-\tau_0$ and $\hat{\tau}_{U}^*-\hat{\tau}_{U}$, which seems very challenging. 
 Overall $U_4$ and $U_5$ have the best coverage probabilities among all methods, although they appear to be conservative (i.e., over-coverage) in a few settings.
 
 


\section{Conclusions}
\label{sec:conclusions}

In this article, we introduce a new estimation method for the change point location in the mean of independent high-dimensional data.  The new U-statistic based objective function is natural given its unbiased and minimal variance property in classical estimation problems, and brings substantial efficiency gain to the change point location estimation in the high-dimensional setting, as demonstrated in both theory and simulations.
The convergence rate and asymptotic distribution of the location estimate are obtained under mild assumptions using new technical arguments that involve some nontrivial asymptotic theory for the high-dimensional U-statistics. A bootstrap-based approach was also proposed to construct a confidence interval, which seems to work well in all simulation settings. Our theoretical results and numerical  findings are significant as they suggest that (i) the U-statistic based point estimator is preferred to the least squares based counterpart in break date estimation, especially when $p>>n$; (ii) 
Bootstrap-based interval is fairly adaptive to different magnitude of change, and should be preferred to the asymptotic plug-in approach. In addition, U-statistic based estimation approach is recommended to couple with either nonparametric bootstrap or parametric bootstrap with a suitably estimated covariance matrix in constructing an interval for the break date.  

To conclude, the work we present in this article opens up several new directions for future research.  
The assumption of independence (over time) is crucial for the formulation of our U-statistic based objective function and derivation of asymptotic property of our estimator. It would be desirable to extend our methodology and theory to cover the temporally dependent case, as in practice many high-dimensional  time-ordered data have weak dependence over time. In view of recent work of \cite{wangshao19}, some trimming might be needed in forming the U-statistic based objective function. 
In addition, nonparametric bootstrap-based confidence interval performs well in simulation but our theory can only cover the parametric bootstrap. A complete theoretical justification for nonparametric bootstrap would be interesting. At last, our method and theory are limited to the relatively simple model with only one change point.  For the linear  regression model with low-dimensional covariates, see \cite{baiperron98} for a suite of least squares based procedures 
for the estimation of change point locations and the construction of tests that allow inference to be made about the presence of structural change and the number of change points. It would be certainly interesting to extend our U-statistic based approach to the model with multiple change points in mean and also to high-dimensional regression setting.
We leave these important topics for future investigation.

\newpage

\begin{table}
\footnotesize
	\centering

	\begin{tabular}{ccccccc|cccccc}
\hline
	&\multicolumn{6}{c|}{Sparse}&\multicolumn{6}{c}{Dense}\\\cline{2-13}
	$\Sigma$ & $p$   & $n$   &             & Bias   & Variance & MSE    & $p$   & $n$   &             & Bias   & Variance & MSE    \\ \hline
\multirow{12}{*}{ID}       & \multirow{6}{*}{50}  & \multirow{2}{*}{50}  & $\hat{\tau}_U$    & 792.8  & 309.6    & 372.4  & \multirow{6}{*}{50}  & \multirow{2}{*}{50}  & $\hat{\tau}_U$    & 796.9  & 316.6    & 380.1  \\
         &       &       & $\hat{\tau}_{LS}$ & 763.5  & 401.6    & 459.9  &       &       & $\hat{\tau}_{LS}$ & 771.4  & 411.1    & 470.6  \\ \cline{4-7} \cline{10-13} 
         &       & \multirow{2}{*}{100} & $\hat{\tau}_U$    & 86.0   & 31.1     & 31.8   &       & \multirow{2}{*}{100} & $\hat{\tau}_U$    & 86.9   & 31.6     & 32.3   \\
         &       &       & $\hat{\tau}_{LS}$ & 69.8   & 41.0     & 41.5   &       &       & $\hat{\tau}_{LS}$ & 71.8   & 38.9     & 39.4   \\ \cline{4-7} \cline{10-13} 
         &       & \multirow{2}{*}{200} & $\hat{\tau}_U$    & 4.9    & 0.8      & 0.8    &       & \multirow{2}{*}{200} & $\hat{\tau}_U$    & 4.6    & 0.8      & 0.8    \\
         &       &       & $\hat{\tau}_{LS}$ & 3.0    & 0.7      & 0.7    &       &       & $\hat{\tau}_{LS}$ & 2.6    & 0.8      & 0.8    \\ \cline{2-13} 
         & \multirow{6}{*}{150} & \multirow{2}{*}{50}  & $\hat{\tau}_U$    & 77.1   & 22.1     & 22.7   & \multirow{6}{*}{150} & \multirow{2}{*}{50}  & $\hat{\tau}_U$    & 71.0   & 19.9     & 20.4   \\
         &       &       & $\hat{\tau}_{LS}$ & 48.4   & 22.3     & 22.5   &       &       & $\hat{\tau}_{LS}$ & 48.4   & 22.5     & 22.8   \\ \cline{4-7} \cline{10-13} 
         &       & \multirow{2}{*}{100} & $\hat{\tau}_U$    & 4.6    & 0.6      & 0.6    &       & \multirow{2}{*}{100} & $\hat{\tau}_U$    & 5.8    & 0.6      & 0.6    \\
         &       &       & $\hat{\tau}_{LS}$ & 1.9    & 0.6      & 0.6    &       &       & $\hat{\tau}_{LS}$ & 2.4    & 0.6      & 0.6    \\ \cline{4-7} \cline{10-13} 
         &       & \multirow{2}{*}{200} & $\hat{\tau}_U$    & 0.6    & 0.1      & 0.1    &       & \multirow{2}{*}{200} & $\hat{\tau}_U$    & 0.7    & 0.1      & 0.1    \\
         &       &       & $\hat{\tau}_{LS}$ & 0.1    & 0.1      & 0.1    &       &       & $\hat{\tau}_{LS}$ & 0.3    & 0.1      & 0.1    \\ \hline
\multirow{12}{*}{AR}      & \multirow{6}{*}{50}  & \multirow{2}{*}{50}  & $\hat{\tau}_U$    & 1847.9 & 629.2    & 970.6  & \multirow{6}{*}{50}  & \multirow{2}{*}{50}  & $\hat{\tau}_U$    & 1871.0 & 626.9    & 977.0  \\
         &       &       & $\hat{\tau}_{LS}$ & 1871.3 & 860.4    & 1210.6 &       &       & $\hat{\tau}_{LS}$ & 1926.5 & 861.1    & 1232.1 \\ \cline{4-7} \cline{10-13} 
         &       & \multirow{2}{*}{100} & $\hat{\tau}_U$    & 920.7  & 379.7    & 464.4  &       & \multirow{2}{*}{100} & $\hat{\tau}_U$    & 821.0  & 345.9    & 413.3  \\
         &       &       & $\hat{\tau}_{LS}$ & 1011.4 & 578.5    & 680.7  &       &       & $\hat{\tau}_{LS}$ & 924.1  & 538.5    & 623.9  \\ \cline{4-7} \cline{10-13} 
         &       & \multirow{2}{*}{200} & $\hat{\tau}_U$    & 128.6  & 55.6     & 57.2   &       & \multirow{2}{*}{200} & $\hat{\tau}_U$    & 46.0   & 18.4     & 18.6   \\
         &       &       & $\hat{\tau}_{LS}$ & 134.0  & 86.1     & 87.9   &       &       & $\hat{\tau}_{LS}$ & 49.4   & 31.4     & 31.6   \\ \cline{2-13} 
         & \multirow{6}{*}{150} & \multirow{2}{*}{50}  & $\hat{\tau}_U$    & 907.7  & 356.7    & 439.1  & \multirow{6}{*}{150} & \multirow{2}{*}{50}  & $\hat{\tau}_U$    & 804.7  & 320.3    & 385.0  \\
         &       &       & $\hat{\tau}_{LS}$ & 909.4  & 477.2    & 559.9  &       &       & $\hat{\tau}_{LS}$ & 799.0  & 419.9    & 483.7  \\ \cline{4-7} \cline{10-13} 
         &       & \multirow{2}{*}{100} & $\hat{\tau}_U$    & 161.4  & 66.0     & 68.6   &       & \multirow{2}{*}{100} & $\hat{\tau}_U$    & 81.4   & 32.0     & 32.6   \\
         &       &       & $\hat{\tau}_{LS}$ & 152.6  & 90.9     & 93.3   &       &       & $\hat{\tau}_{LS}$ & 75.0   & 44.9     & 45.5   \\ \cline{4-7} \cline{10-13} 
         &       & \multirow{2}{*}{200} & $\hat{\tau}_U$    & 7.0    & 1.6      & 1.6    &       & \multirow{2}{*}{200} & $\hat{\tau}_U$    & 1.8    & 0.3      & 0.3    \\
         &       &       & $\hat{\tau}_{LS}$ & 4.1    & 1.6      & 1.6    &       &       & $\hat{\tau}_{LS}$ & 1.3    & 0.3      & 0.3    \\ \hline
\multirow{12}{*}{BD}   & \multirow{6}{*}{50}  & \multirow{2}{*}{50}  & $\hat{\tau}_U$    & 1215.9 & 462.6    & 610.4  & \multirow{6}{*}{50}  & \multirow{2}{*}{50}  & $\hat{\tau}_U$    & 1175.5 & 446.0    & 584.1  \\
         &       &       & $\hat{\tau}_{LS}$ & 1196.9 & 607.4    & 750.6  &       &       & $\hat{\tau}_{LS}$ & 1147.1 & 578.1    & 709.7  \\ \cline{4-7} \cline{10-13} 
         &       & \multirow{2}{*}{100} & $\hat{\tau}_U$    & 318.7  & 131.5    & 141.7  &       & \multirow{2}{*}{100} & $\hat{\tau}_U$    & 221.8  & 90.9     & 95.9   \\
         &       &       & $\hat{\tau}_{LS}$ & 299.5  & 177.5    & 186.4  &       &       & $\hat{\tau}_{LS}$ & 222.5  & 132.3    & 137.2  \\ \cline{4-7} \cline{10-13} 
         &       & \multirow{2}{*}{200} & $\hat{\tau}_U$    & 16.0   & 4.6      & 4.7    &       & \multirow{2}{*}{200} & $\hat{\tau}_U$    & 5.1    & 1.0      & 1.0    \\
         &       &       & $\hat{\tau}_{LS}$ & 10.4   & 4.3      & 4.3    &       &       & $\hat{\tau}_{LS}$ & 3.0    & 1.0      & 1.0    \\ \cline{2-13} 
         & \multirow{6}{*}{150} & \multirow{2}{*}{50}  & $\hat{\tau}_U$    & 254.0  & 95.8     & 102.3  & \multirow{6}{*}{150} & \multirow{2}{*}{50}  & $\hat{\tau}_U$    & 181.2  & 66.1     & 69.4   \\
         &       &       & $\hat{\tau}_{LS}$ & 209.0  & 113.0    & 117.3  &       &       & $\hat{\tau}_{LS}$ & 139.9  & 74.2     & 76.1   \\ \cline{4-7} \cline{10-13} 
         &       & \multirow{2}{*}{100} & $\hat{\tau}_U$    & 19.2   & 3.6      & 3.7    &       & \multirow{2}{*}{100} & $\hat{\tau}_U$    & 7.8    & 1.4      & 1.4    \\
         &       &       & $\hat{\tau}_{LS}$ & 11.4   & 3.8      & 3.8    &       &       & $\hat{\tau}_{LS}$ & 3.6    & 1.2      & 1.2    \\ \cline{4-7} \cline{10-13} 
         &       & \multirow{2}{*}{200} & $\hat{\tau}_U$    & 1.8    & 0.3      & 0.3    &       & \multirow{2}{*}{200} & $\hat{\tau}_U$    & 1.1    & 0.1      & 0.1    \\
         &       &       & $\hat{\tau}_{LS}$ & 0.6    & 0.3      & 0.3    &       &       & $\hat{\tau}_{LS}$ & 0.6    & 0.1      & 0.1    \\ \hline
\multirow{12}{*}{CS}       & \multirow{6}{*}{50}  & \multirow{2}{*}{50}  & $\hat{\tau}_U$    & 2359.6 & 706.3    & 1263.1 & \multirow{6}{*}{50}  & \multirow{2}{*}{50}  & $\hat{\tau}_U$    & 2334.9 & 701.0    & 1246.1 \\
         &       &       & $\hat{\tau}_{LS}$ & 2453.7 & 982.6    & 1584.6 &       &       & $\hat{\tau}_{LS}$ & 2446.0 & 983.2    & 1581.5 \\ \cline{4-7} \cline{10-13} 
         &       & \multirow{2}{*}{100} & $\hat{\tau}_U$    & 1623.2 & 597.2    & 860.6  &       & \multirow{2}{*}{100} & $\hat{\tau}_U$    & 1641.8 & 603.7    & 873.2  \\
         &       &       & $\hat{\tau}_{LS}$ & 1881.3 & 963.0    & 1316.9 &       &       & $\hat{\tau}_{LS}$ & 1848.8 & 953.0    & 1294.8 \\ \cline{4-7} \cline{10-13} 
         &       & \multirow{2}{*}{200} & $\hat{\tau}_U$    & 604.7  & 274.7    & 311.2  &       & \multirow{2}{*}{200} & $\hat{\tau}_U$    & 549.8  & 250.8    & 281.0  \\
         &       &       & $\hat{\tau}_{LS}$ & 813.5  & 528.5    & 594.7  &       &       & $\hat{\tau}_{LS}$ & 744.2  & 482.7    & 538.1  \\ \cline{2-13} 
         & \multirow{6}{*}{150} & \multirow{2}{*}{50}  & $\hat{\tau}_U$    & 2202.7 & 684.4    & 1169.6 & \multirow{6}{*}{150} & \multirow{2}{*}{50}  & $\hat{\tau}_U$    & 2194.8 & 684.6    & 1166.3 \\
         &       &       & $\hat{\tau}_{LS}$ & 2279.0 & 949.4    & 1468.7 &       &       & $\hat{\tau}_{LS}$ & 2301.1 & 957.4    & 1486.8 \\ \cline{4-7} \cline{10-13} 
         &       & \multirow{2}{*}{100} & $\hat{\tau}_U$    & 1508.2 & 568.9    & 796.3  &       & \multirow{2}{*}{100} & $\hat{\tau}_U$    & 1487.9 & 566.2    & 787.6  \\
         &       &       & $\hat{\tau}_{LS}$ & 1747.2 & 909.5    & 1214.7 &       &       & $\hat{\tau}_{LS}$ & 1757.0 & 913.6    & 1222.3 \\ \cline{4-7} \cline{10-13} 
         &       & \multirow{2}{*}{200} & $\hat{\tau}_U$    & 561.6  & 257.1    & 288.6  &       & \multirow{2}{*}{200} & $\hat{\tau}_U$    & 525.6  & 242.8    & 270.4  \\
         &       &       & $\hat{\tau}_{LS}$ & 775.7  & 503.5    & 563.7  &       &       & $\hat{\tau}_{LS}$ & 737.7  & 485.7    & 540.1  \\ \hline
	\end{tabular}
	\caption{Finite sample performance of location estimates ($\hat{\tau}_U$ and $\hat{\tau}_{LS}$) with $\tau_0 = 0.2$ (in $10^{-4}$)} \label{tab:point0.2}
\end{table}

\begin{table}
\footnotesize
	\centering

	\begin{tabular}{ccccccc|cccccc}
	\hline
	&\multicolumn{6}{c|}{Sparse}&\multicolumn{6}{c}{Dense}\\\cline{2-13}
	$\Sigma$ & $p$   & $n$   &             & Bias  & Variance & MSE   & $p$   & $n$   &             & Bias  & Variance & MSE   \\ \hline
\multirow{12}{*}{ID}       & \multirow{6}{*}{50}  & \multirow{2}{*}{50}  & $\hat{\tau}_U$    & 3.6   & 80.0     & 80.0  & \multirow{6}{*}{50}  & \multirow{2}{*}{50}  & $\hat{\tau}_U$    & 5.2   & 81.4     & 81.4  \\
         &       &       & $\hat{\tau}_{LS}$ & 2.8   & 110.5    & 110.5 &       &       & $\hat{\tau}_{LS}$ & -8.2  & 112.9    & 113.0 \\ \cline{4-7} \cline{10-13} 
         &       & \multirow{2}{*}{100} & $\hat{\tau}_U$    & 1.7   & 6.1      & 6.1   &       & \multirow{2}{*}{100} & $\hat{\tau}_U$    & -3.2  & 6.4      & 6.4   \\
         &       &       & $\hat{\tau}_{LS}$ & 2.1   & 6.6      & 6.6   &       &       & $\hat{\tau}_{LS}$ & -3.2  & 6.9      & 6.9   \\ \cline{4-7} \cline{10-13} 
         &       & \multirow{2}{*}{200} & $\hat{\tau}_U$    & 0.0   & 0.5      & 0.5   &       & \multirow{2}{*}{200} & $\hat{\tau}_U$    & 0.0   & 0.5      & 0.5   \\
         &       &       & $\hat{\tau}_{LS}$ & 0.0   & 0.5      & 0.5   &       &       & $\hat{\tau}_{LS}$ & -0.1  & 0.5      & 0.5   \\ \cline{2-13} 
         & \multirow{6}{*}{150} & \multirow{2}{*}{50}  & $\hat{\tau}_U$    & 2.3   & 3.2      & 3.2   & \multirow{6}{*}{150} & \multirow{2}{*}{50}  & $\hat{\tau}_U$    & -0.7  & 3.1      & 3.1   \\
         &       &       & $\hat{\tau}_{LS}$ & 2.0   & 3.3      & 3.3   &       &       & $\hat{\tau}_{LS}$ & -0.4  & 3.2      & 3.2   \\ \cline{4-7} \cline{10-13} 
         &       & \multirow{2}{*}{100} & $\hat{\tau}_U$    & 0.3   & 0.3      & 0.3   &       & \multirow{2}{*}{100} & $\hat{\tau}_U$    & 0.5   & 0.3      & 0.3   \\
         &       &       & $\hat{\tau}_{LS}$ & 0.2   & 0.3      & 0.3   &       &       & $\hat{\tau}_{LS}$ & 0.5   & 0.3      & 0.3   \\ \cline{4-7} \cline{10-13} 
         &       & \multirow{2}{*}{200} & $\hat{\tau}_U$    & -0.2  & 0.0      & 0.0   &       & \multirow{2}{*}{200} & $\hat{\tau}_U$    & -0.1  & 0.0      & 0.0   \\
         &       &       & $\hat{\tau}_{LS}$ & -0.2  & 0.0      & 0.0   &       &       & $\hat{\tau}_{LS}$ & -0.1  & 0.0      & 0.0   \\ \hline
\multirow{12}{*}{AR}       & \multirow{6}{*}{50}  & \multirow{2}{*}{50}  & $\hat{\tau}_U$    & -2.1  & 342.2    & 342.2 & \multirow{6}{*}{50}  & \multirow{2}{*}{50}  & $\hat{\tau}_U$    & 4.7   & 339.8    & 339.7 \\
         &       &       & $\hat{\tau}_{LS}$ & -8.9  & 521.1    & 521.1 &       &       & $\hat{\tau}_{LS}$ & 1.5   & 519.7    & 519.7 \\ \cline{4-7} \cline{10-13} 
         &       & \multirow{2}{*}{100} & $\hat{\tau}_U$    & -23.0 & 127.2    & 127.2 &       & \multirow{2}{*}{100} & $\hat{\tau}_U$    & 0.9   & 82.1     & 82.1  \\
         &       &       & $\hat{\tau}_{LS}$ & -25.2 & 202.1    & 202.2 &       &       & $\hat{\tau}_{LS}$ & 9.4   & 143.5    & 143.5 \\ \cline{4-7} \cline{10-13} 
         &       & \multirow{2}{*}{200} & $\hat{\tau}_U$    & 4.2   & 12.4     & 12.4  &       & \multirow{2}{*}{200} & $\hat{\tau}_U$    & 0.7   & 2.6      & 2.6   \\
         &       &       & $\hat{\tau}_{LS}$ & 4.1   & 14.6     & 14.6  &       &       & $\hat{\tau}_{LS}$ & 0.4   & 2.7      & 2.7   \\ \cline{2-13} 
         & \multirow{6}{*}{150} & \multirow{2}{*}{50}  & $\hat{\tau}_U$    & -1.2  & 90.4     & 90.4  & \multirow{6}{*}{150} & \multirow{2}{*}{50}  & $\hat{\tau}_U$    & 2.2   & 59.4     & 59.4  \\
         &       &       & $\hat{\tau}_{LS}$ & -11.5 & 133.7    & 133.7 &       &       & $\hat{\tau}_{LS}$ & -3.4  & 90.7     & 90.7  \\ \cline{4-7} \cline{10-13} 
         &       & \multirow{2}{*}{100} & $\hat{\tau}_U$    & -1.7  & 10.8     & 10.8  &       & \multirow{2}{*}{100} & $\hat{\tau}_U$    & 0.5   & 2.9      & 2.9   \\
         &       &       & $\hat{\tau}_{LS}$ & -1.2  & 12.1     & 12.1  &       &       & $\hat{\tau}_{LS}$ & 0.7   & 3.0      & 3.0   \\ \cline{4-7} \cline{10-13} 
         &       & \multirow{2}{*}{200} & $\hat{\tau}_U$    & -0.5  & 0.9      & 0.9   &       & \multirow{2}{*}{200} & $\hat{\tau}_U$    & -0.5  & 0.1      & 0.1   \\
         &       &       & $\hat{\tau}_{LS}$ & -0.5  & 0.9      & 0.9   &       &       & $\hat{\tau}_{LS}$ & -0.4  & 0.1      & 0.1   \\ \hline
\multirow{12}{*}{BD}   & \multirow{6}{*}{50}  & \multirow{2}{*}{50}  & $\hat{\tau}_U$    & -7.4  & 170.4    & 170.4 & \multirow{6}{*}{50}  & \multirow{2}{*}{50}  & $\hat{\tau}_U$    & 10.7  & 148.5    & 148.5 \\
         &       &       & $\hat{\tau}_{LS}$ & -0.8  & 245.5    & 245.5 &       &       & $\hat{\tau}_{LS}$ & -1.8  & 222.2    & 222.2 \\ \cline{4-7} \cline{10-13} 
         &       & \multirow{2}{*}{100} & $\hat{\tau}_U$    & 0.8   & 31.2     & 31.1  &       & \multirow{2}{*}{100} & $\hat{\tau}_U$    & 0.5   & 15.4     & 15.4  \\
         &       &       & $\hat{\tau}_{LS}$ & -0.4  & 38.5     & 38.5  &       &       & $\hat{\tau}_{LS}$ & -0.1  & 19.2     & 19.2  \\ \cline{4-7} \cline{10-13} 
         &       & \multirow{2}{*}{200} & $\hat{\tau}_U$    & 0.3   & 2.4      & 2.4   &       & \multirow{2}{*}{200} & $\hat{\tau}_U$    & -0.3  & 0.6      & 0.6   \\
         &       &       & $\hat{\tau}_{LS}$ & 0.3   & 2.4      & 2.4   &       &       & $\hat{\tau}_{LS}$ & -0.3  & 0.6      & 0.6   \\ \cline{2-13} 
         & \multirow{6}{*}{150} & \multirow{2}{*}{50}  & $\hat{\tau}_U$    & -2.7  & 15.6     & 15.6  & \multirow{6}{*}{150} & \multirow{2}{*}{50}  & $\hat{\tau}_U$    & 1.0   & 6.9      & 6.9   \\
         &       &       & $\hat{\tau}_{LS}$ & -1.6  & 18.8     & 18.8  &       &       & $\hat{\tau}_{LS}$ & 0.6   & 7.6      & 7.6   \\ \cline{4-7} \cline{10-13} 
         &       & \multirow{2}{*}{100} & $\hat{\tau}_U$    & 0.1   & 1.7      & 1.7   &       & \multirow{2}{*}{100} & $\hat{\tau}_U$    & 0.1   & 0.6      & 0.6   \\
         &       &       & $\hat{\tau}_{LS}$ & 0.0   & 1.7      & 1.7   &       &       & $\hat{\tau}_{LS}$ & 0.1   & 0.6      & 0.6   \\ \cline{4-7} \cline{10-13} 
         &       & \multirow{2}{*}{200} & $\hat{\tau}_U$    & 0.0   & 0.2      & 0.2   &       & \multirow{2}{*}{200} & $\hat{\tau}_U$    & 0.3   & 0.1      & 0.1   \\
         &       &       & $\hat{\tau}_{LS}$ & 0.1   & 0.2      & 0.2   &       &       & $\hat{\tau}_{LS}$ & 0.3   & 0.1      & 0.1   \\ \hline
\multirow{12}{*}{CS}       & \multirow{6}{*}{50}  & \multirow{2}{*}{50}  & $\hat{\tau}_U$    & -13.0 & 501.9    & 501.9 & \multirow{6}{*}{50}  & \multirow{2}{*}{50}  & $\hat{\tau}_U$    & -18.1 & 493.2    & 493.3 \\
         &       &       & $\hat{\tau}_{LS}$ & -13.8 & 767.3    & 767.3 &       &       & $\hat{\tau}_{LS}$ & -12.8 & 756.6    & 756.5 \\ \cline{4-7} \cline{10-13} 
         &       & \multirow{2}{*}{100} & $\hat{\tau}_U$    & 6.7   & 301.3    & 301.3 &       & \multirow{2}{*}{100} & $\hat{\tau}_U$    & -4.1  & 279.7    & 279.7 \\
         &       &       & $\hat{\tau}_{LS}$ & 2.9   & 572.7    & 572.7 &       &       & $\hat{\tau}_{LS}$ & 12.7  & 541.1    & 541.1 \\ \cline{4-7} \cline{10-13} 
         &       & \multirow{2}{*}{200} & $\hat{\tau}_U$    & 6.4   & 63.1     & 63.1  &       & \multirow{2}{*}{200} & $\hat{\tau}_U$    & 7.7   & 57.2     & 57.2  \\
         &       &       & $\hat{\tau}_{LS}$ & 10.9  & 132.1    & 132.1 &       &       & $\hat{\tau}_{LS}$ & 11.1  & 117.5    & 117.5 \\ \cline{2-13} 
         & \multirow{6}{*}{150} & \multirow{2}{*}{50}  & $\hat{\tau}_U$    & 19.3  & 437.2    & 437.2 & \multirow{6}{*}{150} & \multirow{2}{*}{50}  & $\hat{\tau}_U$    & -0.4  & 426.2    & 426.2 \\
         &       &       & $\hat{\tau}_{LS}$ & 12.4  & 684.4    & 684.3 &       &       & $\hat{\tau}_{LS}$ & -2.1  & 685.6    & 685.6 \\ \cline{4-7} \cline{10-13} 
         &       & \multirow{2}{*}{100} & $\hat{\tau}_U$    & -0.9  & 245.2    & 245.1 &       & \multirow{2}{*}{100} & $\hat{\tau}_U$    & 6.8   & 235.5    & 235.5 \\
         &       &       & $\hat{\tau}_{LS}$ & -7.1  & 471.3    & 471.2 &       &       & $\hat{\tau}_{LS}$ & 2.6   & 464.5    & 464.4 \\ \cline{4-7} \cline{10-13} 
         &       & \multirow{2}{*}{200} & $\hat{\tau}_U$    & 2.9   & 57.5     & 57.5  &       & \multirow{2}{*}{200} & $\hat{\tau}_U$    & 4.1   & 49.5     & 49.5  \\
         &       &       & $\hat{\tau}_{LS}$ & 7.4   & 114.3    & 114.3 &       &       & $\hat{\tau}_{LS}$ & 12.9  & 108.5    & 108.5 \\ \hline
	\end{tabular}
	\caption{Finite sample performance of location estimates ($\hat{\tau}_U$ and $\hat{\tau}_{LS}$) with $\tau_0 = 0.5$ (in $10^{-4}$)}\label{tab:point0.5}
\end{table}

\begin{table}[h!]
\footnotesize
\begin{tabular}{c|ccc|ccc|ccc|ccc}
\hline
               & \multicolumn{6}{c|}{$p = 50$}                           & \multicolumn{6}{c}{$p = 150$}                          \\ \cline{2-13}
               & \multicolumn{3}{c|}{Dense} & \multicolumn{3}{c|}{Sparse} & \multicolumn{3}{c|}{Dense} & \multicolumn{3}{c}{Sparse} \\ \hline
$\tau_0=0.2$   &     &     $n$    &        &    &     $n$      &        &    &    $n$       &        &     &     $n$     &        \\
               & 50      & $100$  & $200$  & $50$    & $100$   & $200$  & $50$    & $100$   & $200$  & $50$    & $100$   & $200$  \\ \hline
$U_1$          & 0.765   & 0.740  & 0.611  & 0.754   & 0.792   & 0.579  & 0.881   & 0.744   & 0.833  & 0.879   & 0.780   & 0.827  \\
Length         & 0.113   & 0.076  & 0.008  & 0.113   & 0.076   & 0.008  & 0.049   & 0.015   & 0.004  & 0.049   & 0.015   & 0.004  \\ \cline{2-13} 
$U_2$          & 0.797   & 0.768  & 0.685  & 0.798   & 0.798   & 0.661  & 0.862   & 0.800   & 0.833  & 0.846   & 0.833   & 0.827  \\
Length         & 0.172   & 0.098  & 0.009  & 0.164   & 0.096   & 0.009  & 0.067   & 0.018   & 0.004  & 0.068   & 0.018   & 0.004  \\ \cline{2-13} 
$U_3$          & 0.727   & 0.732  & 0.917  & 0.722   & 0.787   & 0.920  & 0.738   & 0.859   & 0.966  & 0.718   & 0.889   & 0.954  \\
Length         & 0.071   & 0.062  & 0.022  & 0.072   & 0.061   & 0.023  & 0.008   & 0.014   & 0.009  & 0.008   & 0.014   & 0.010  \\ \cline{2-13} 
$U_4$          & 0.886   & 0.902  & 0.961  & 0.894   & 0.923   & 0.961  & 0.959   & 0.975   & 0.977  & 0.955   & 0.967   & 0.969  \\
Length         & 0.460   & 0.345  & 0.035  & 0.268   & 0.223   & 0.035  & 0.174   & 0.037   & 0.012  & 0.141   & 0.037   & 0.012  \\ \cline{2-13} 
$U_5$          & 0.831   & 0.812  & 0.930  & 0.837   & 0.844   & 0.929  & 0.913   & 0.937   & 0.971  & 0.907   & 0.947   & 0.961  \\
Length         & 0.194   & 0.097  & 0.025  & 0.172   & 0.094   & 0.025  & 0.093   & 0.031   & 0.011  & 0.092   & 0.030   & 0.011  \\ \cline{2-13} 
$LS_1$         & 0.601   & 0.626  & 0.885  & 0.575   & 0.677   & 0.874  & 0.780   & 0.827   & 0.893  & 0.776   & 0.840   & 0.890  \\
Length         & 0.035   & 0.034  & 0.017  & 0.035   & 0.034   & 0.017  & 0.020   & 0.010   & 0.005  & 0.020   & 0.010   & 0.005  \\ \cline{2-13} 
$LS_2$         & 0.643   & 0.607  & 0.888  & 0.606   & 0.657   & 0.877  & 0.686   & 0.732   & 0.889  & 0.671   & 0.779   & 0.873  \\
Length         & 0.034   & 0.032  & 0.017  & 0.034   & 0.031   & 0.017  & 0.001   & 0.001   & 0.004  & 0.000   & 0.000   & 0.004  \\ \hline
 $\tau_0 = 0.5$ & \multicolumn{3}{c|}{$n$} & \multicolumn{3}{c|}{$n$} & \multicolumn{3}{c|}{$n$} & \multicolumn{3}{c}{$n$} \\
               & $50$    & $100$  & $200$  & $50$    & $100$   & $200$  & $50$    & $100$   & $200$  & $50$    & $100$   & $200$  \\ \hline
$U_1$          & 0.761   & 0.789  & 0.617  & 0.742   & 0.762   & 0.605  & 0.763   & 0.818   & 0.858  & 0.792   & 0.802   & 0.867  \\
Length         & 0.072   & 0.049  & 0.005  & 0.072   & 0.049   & 0.005  & 0.031   & 0.010   & 0.003  & 0.031   & 0.010   & 0.003  \\ \cline{2-13} 
$U_2$          & 0.820   & 0.755  & 0.625  & 0.772   & 0.754   & 0.611  & 0.830   & 0.818   & 0.858  & 0.851   & 0.803   & 0.867  \\
Length         & 0.098   & 0.056  & 0.006  & 0.093   & 0.058   & 0.006  & 0.037   & 0.011   & 0.003  & 0.037   & 0.010   & 0.003  \\ \cline{2-13} 
$U_3$          & 0.839   & 0.850  & 0.940  & 0.818   & 0.844   & 0.952  & 0.814   & 0.941   & 0.974  & 0.836   & 0.944   & 0.974  \\
Length         & 0.073   & 0.064  & 0.023  & 0.072   & 0.065   & 0.023  & 0.012   & 0.018   & 0.010  & 0.012   & 0.017   & 0.010  \\ \cline{2-13} 
$U_4$          & 0.950   & 0.958  & 0.967  & 0.957   & 0.945   & 0.970  & 0.968   & 0.972   & 0.974  & 0.978   & 0.972   & 0.976  \\
Length         & 0.250   & 0.166  & 0.030  & 0.240   & 0.170   & 0.029  & 0.071   & 0.026   & 0.010  & 0.070   & 0.026   & 0.010  \\ \cline{2-13} 
$U_5$          & 0.916   & 0.887  & 0.946  & 0.903   & 0.872   & 0.949  & 0.963   & 0.962   & 0.972  & 0.968   & 0.961   & 0.971  \\
Length         & 0.136   & 0.081  & 0.024  & 0.133   & 0.083   & 0.024  & 0.052   & 0.021   & 0.010  & 0.052   & 0.020   & 0.010  \\ \cline{2-13} 
$LS_1$         & 0.781   & 0.777  & 0.911  & 0.767   & 0.770   & 0.924  & 0.838   & 0.895   & 0.914  & 0.876   & 0.891   & 0.929  \\
Length         & 0.055   & 0.043  & 0.019  & 0.055   & 0.045   & 0.019  & 0.020   & 0.010   & 0.005  & 0.020   & 0.010   & 0.005  \\ \cline{2-13} 
$LS_2$         & 0.750   & 0.782  & 0.925  & 0.736   & 0.772   & 0.933  & 0.763   & 0.837   & 0.947  & 0.789   & 0.827   & 0.951  \\
Length         & 0.041   & 0.040  & 0.019  & 0.041   & 0.041   & 0.019  & 0.000   & 0.002   & 0.008  & 0.000   & 0.002   & 0.007  \\ \hline
\end{tabular}
\caption{Coverage probability and average length of seven  confidence intervals for the ID covariance model} \label{tab:CI_ID}
\end{table}

\begin{table}[h!]
\footnotesize
\begin{tabular}{c|ccc|ccc|ccc|ccc}
\hline
               & \multicolumn{6}{c|}{$p = 50$}                           & \multicolumn{6}{c}{$p = 150$}                          \\ \cline{2-13}
               & \multicolumn{3}{c|}{Dense} & \multicolumn{3}{c|}{Sparse} & \multicolumn{3}{c|}{Dense} & \multicolumn{3}{c}{Sparse} \\ \hline
$\tau_0=0.2$   &     &     $n$    &        &    &     $n$      &        &    &    $n$       &        &     &     $n$     &        \\
               & 50      & $100$  & $200$  & $50$    & $100$   & $200$  & $50$    & $100$   & $200$  & $50$    & $100$   & $200$  \\ \hline
$U_1$          & 0.740   & 0.751  & 0.820  & 0.773   & 0.713   & 0.622  & 0.801   & 0.843   & 0.898  & 0.780   & 0.791   & 0.725  \\
Length         & 0.490   & 0.330  & 0.036  & 0.315   & 0.212   & 0.023  & 0.253   & 0.077   & 0.013  & 0.141   & 0.044   & 0.012  \\ \cline{2-13} 
$U_2$          & 0.676   & 0.639  & 0.802  & 0.688   & 0.650   & 0.600  & 0.760   & 0.837   & 0.872  & 0.751   & 0.768   & 0.650  \\
Length         & 0.622   & 0.378  & 0.041  & 0.579   & 0.360   & 0.027  & 0.346   & 0.094   & 0.014  & 0.183   & 0.052   & 0.013  \\ \cline{2-13} 
$U_3$          & 0.627   & 0.644  & 0.910  & 0.753   & 0.729   & 0.911  & 0.564   & 0.781   & 0.931  & 0.712   & 0.845   & 0.890  \\
Length         & 0.494   & 0.344  & 0.061  & 0.381   & 0.287   & 0.086  & 0.090   & 0.044   & 0.014  & 0.092   & 0.055   & 0.029  \\ \cline{2-13} 
$U_4$          & 0.666   & 0.719  & 0.970  & 0.825   & 0.817   & 0.971  & 0.763   & 0.939   & 0.995  & 0.906   & 0.950   & 0.970  \\
Length         & 0.743   & 0.686  & 0.297  & 0.661   & 0.591   & 0.203  & 0.627   & 0.335   & 0.057  & 0.415   & 0.171   & 0.064  \\ \cline{2-13} 
$U_5$          & 0.645   & 0.682  & 0.943  & 0.784   & 0.783   & 0.933  & 0.726   & 0.888   & 0.970  & 0.858   & 0.910   & 0.922  \\
Length         & 0.570   & 0.489  & 0.084  & 0.456   & 0.397   & 0.108  & 0.344   & 0.099   & 0.023  & 0.261   & 0.090   & 0.040  \\ \cline{2-13} 
$LS_1$         & 0.220   & 0.274  & 0.720  & 0.357   & 0.352   & 0.605  & 0.367   & 0.537   & 0.816  & 0.464   & 0.531   & 0.610  \\
Length         & 0.032   & 0.028  & 0.018  & 0.038   & 0.035   & 0.020  & 0.020   & 0.020   & 0.005  & 0.020   & 0.010   & 0.005  \\ \cline{2-13} 
$LS_2$         & 0.430   & 0.429  & 0.852  & 0.584   & 0.594   & 0.861  & 0.421   & 0.662   & 0.898  & 0.586   & 0.774   & 0.960  \\
Length         & 0.110   & 0.088  & 0.034  & 0.140   & 0.122   & 0.064  & 0.034   & 0.022   & 0.010  & 0.046   & 0.037   & 0.023  \\ \hline
 $\tau_0 = 0.5$ & \multicolumn{3}{c|}{$n$} & \multicolumn{3}{c|}{$n$} & \multicolumn{3}{c|}{$n$} & \multicolumn{3}{c}{$n$} \\
               & $50$    & $100$  & $200$  & $50$    & $100$   & $200$  & $50$    & $100$   & $200$  & $50$    & $100$   & $200$  \\ \hline
$U_1$          & 0.780   & 0.771  & 0.817  & 0.723   & 0.720   & 0.614  & 0.857   & 0.824   & 0.904  & 0.769   & 0.779   & 0.742  \\
Length         & 0.315   & 0.212  & 0.023  & 0.313   & 0.212   & 0.023  & 0.141   & 0.044   & 0.012  & 0.141   & 0.044   & 0.012  \\ \cline{2-13} 
$U_2$          & 0.711   & 0.719  & 0.781  & 0.648   & 0.648   & 0.586  & 0.856   & 0.910   & 0.862  & 0.749   & 0.757   & 0.673  \\
Length         & 0.356   & 0.231  & 0.025  & 0.300   & 0.223   & 0.027  & 0.163   & 0.048   & 0.013  & 0.172   & 0.052   & 0.013  \\ \cline{2-13} 
$U_3$          & 0.754   & 0.766  & 0.938  & 0.694   & 0.738   & 0.898  & 0.780   & 0.912   & 0.941  & 0.705   & 0.829   & 0.898  \\
Length         & 0.363   & 0.245  & 0.049  & 0.333   & 0.266   & 0.087  & 0.069   & 0.034   & 0.017  & 0.090   & 0.055   & 0.029  \\ \cline{2-13} 
$U_4$          & 0.824   & 0.854  & 0.995  & 0.782   & 0.830   & 0.961  & 0.949   & 0.987   & 0.999  & 0.905   & 0.953   & 0.973  \\
Length         & 0.684   & 0.610  & 0.197  & 0.556   & 0.519   & 0.204  & 0.416   & 0.162   & 0.064  & 0.390   & 0.171   & 0.064  \\ \cline{2-13} 
$U_5$          & 0.779   & 0.812  & 0.949  & 0.721   & 0.785   & 0.923  & 0.910   & 0.956   & 0.970  & 0.863   & 0.886   & 0.928  \\
Length         & 0.454   & 0.374  & 0.058  & 0.397   & 0.361   & 0.110  & 0.189   & 0.059   & 0.023  & 0.221   & 0.090   & 0.040  \\ \cline{2-13} 
$LS_1$         & 0.386   & 0.445  & 0.786  & 0.321   & 0.343   & 0.601  & 0.576   & 0.738   & 0.786  & 0.467   & 0.553   & 0.611  \\
Length         & 0.042   & 0.037  & 0.020  & 0.038   & 0.035   & 0.020  & 0.020   & 0.010   & 0.005  & 0.020   & 0.010   & 0.005  \\ \cline{2-13} 
$LS_2$         & 0.561   & 0.595  & 0.901  & 0.539   & 0.596   & 0.852  & 0.679   & 0.852   & 0.904  & 0.595   & 0.771   & 0.882  \\
Length         & 0.115   & 0.097  & 0.032  & 0.139   & 0.118   & 0.065  & 0.036   & 0.020   & 0.010  & 0.046   & 0.038   & 0.023  \\ \hline
\end{tabular}
\caption{Coverage probability and average length of seven  confidence intervals for the AR covariance model} \label{tab:CI_AR}
\end{table}

\begin{table}[h!]
\footnotesize
\begin{tabular}{c|ccc|ccc|ccc|ccc}
\hline
               & \multicolumn{6}{c|}{$p = 50$}                           & \multicolumn{6}{c}{$p = 150$}                          \\ \cline{2-13}
               & \multicolumn{3}{c|}{Dense} & \multicolumn{3}{c|}{Sparse} & \multicolumn{3}{c|}{Dense} & \multicolumn{3}{c}{Sparse} \\ \hline
$\tau_0=0.2$   &     &     $n$    &        &    &     $n$      &        &    &    $n$       &        &     &     $n$     &        \\
               & 50      & $100$  & $200$  & $50$    & $100$   & $200$  & $50$    & $100$   & $200$  & $50$    & $100$   & $200$  \\ \hline
$U_1$          & 0.751   & 0.770  & 0.575  & 0.763   & 0.775   & 0.615  & 0.877   & 0.801   & 0.816  & 0.888   & 0.773   & 0.817  \\
Length         & 0.113   & 0.076  & 0.008  & 0.113   & 0.076   & 0.008  & 0.049   & 0.015   & 0.004  & 0.049   & 0.015   & 0.004  \\ \cline{2-13} 
$U_2$          & 0.801   & 0.796  & 0.648  & 0.785   & 0.785   & 0.665  & 0.856   & 0.861   & 0.816  & 0.866   & 0.827   & 0.817  \\
Length         & 0.179   & 0.097  & 0.009  & 0.170   & 0.099   & 0.009  & 0.067   & 0.017   & 0.004  & 0.066   & 0.018   & 0.004  \\ \cline{2-13} 
$U_3$          & 0.714   & 0.777  & 0.912  & 0.719   & 0.763   & 0.928  & 0.742   & 0.912   & 0.952  & 0.728   & 0.885   & 0.949  \\
Length         & 0.072   & 0.061  & 0.023  & 0.073   & 0.062   & 0.022  & 0.008   & 0.013   & 0.009  & 0.008   & 0.013   & 0.010  \\ \cline{2-13} 
$U_4$          & 0.886   & 0.919  & 0.961  & 0.882   & 0.925   & 0.971  & 0.951   & 0.978   & 0.962  & 0.965   & 0.977   & 0.972  \\
Length         & 0.468   & 0.338  & 0.035  & 0.273   & 0.225   & 0.034  & 0.168   & 0.037   & 0.001  & 0.138   & 0.038   & 0.012  \\ \cline{2-13} 
$U_5$          & 0.827   & 0.833  & 0.925  & 0.824   & 0.837   & 0.943  & 0.900   & 0.960   & 0.962  & 0.921   & 0.951   & 0.962  \\
Length         & 0.199   & 0.096  & 0.025  & 0.175   & 0.096   & 0.025  & 0.092   & 0.030   & 0.001  & 0.091   & 0.030   & 0.011  \\ \cline{2-13} 
$LS_1$         & 0.609   & 0.646  & 0.876  & 0.602   & 0.650   & 0.899  & 0.770   & 0.864   & 0.885  & 0.788   & 0.857   & 0.900  \\
Length         & 0.035   & 0.034  & 0.017  & 0.036   & 0.034   & 0.017  & 0.020   & 0.010   & 0.005  & 0.020   & 0.010   & 0.005  \\ \cline{2-13} 
$LS_2$         & 0.626   & 0.632  & 0.870  & 0.617   & 0.652   & 0.899  & 0.683   & 0.803   & 0.870  & 0.678   & 0.781   & 0.878  \\
Length         & 0.035   & 0.031  & 0.017  & 0.035   & 0.031   & 0.017  & 0.000   & 0.001   & 0.004  & 0.000   & 0.000   & 0.004  \\\hline $\tau_0 = 0.5$ & \multicolumn{3}{c|}{$n$} & \multicolumn{3}{c|}{$n$} & \multicolumn{3}{c|}{$n$} & \multicolumn{3}{c}{$n$} \\
               & $50$    & $100$   & $200$  & $50$    & $100$  & $200$  & $50$    & $100$   & $200$   & $50$    & $100$  & $200$  \\ \hline
$U_1$          & 0.755   & 0.758  & 0.632  & 0.738   & 0.755   & 0.645  & 0.769   & 0.801   & 0.851  & 0.758   & 0.796   & 0.859  \\
Length         & 0.072   & 0.049  & 0.005  & 0.072   & 0.049   & 0.005  & 0.031   & 0.010   & 0.003  & 0.031   & 0.010   & 0.003  \\ \cline{2-13} 
$U_2$          & 0.792   & 0.752  & 0.638  & 0.788   & 0.748   & 0.648  & 0.826   & 0.803   & 0.851  & 0.839   & 0.799   & 0.859  \\
Length         & 0.092   & 0.060  & 0.006  & 0.095   & 0.058   & 0.006  & 0.037   & 0.011   & 0.003  & 0.038   & 0.011   & 0.003  \\ \cline{2-13} 
$U_3$          & 0.814   & 0.839  & 0.943  & 0.826   & 0.833   & 0.945  & 0.818   & 0.942   & 0.972  & 0.825   & 0.935   & 0.977  \\
Length         & 0.072   & 0.065  & 0.023  & 0.072   & 0.065   & 0.023  & 0.012   & 0.018   & 0.010  & 0.013   & 0.018   & 0.010  \\ \cline{2-13} 
$U_4$          & 0.942   & 0.941  & 0.965  & 0.938   & 0.945   & 0.961  & 0.970   & 0.965   & 0.974  & 0.972   & 0.962   & 0.980  \\
Length         & 0.239   & 0.172  & 0.030  & 0.239   & 0.171   & 0.029  & 0.072   & 0.026   & 0.010  & 0.072   & 0.026   & 0.010  \\ \cline{2-13} 
$U_5$          & 0.903   & 0.882  & 0.946  & 0.902   & 0.878   & 0.947  & 0.955   & 0.956   & 0.973  & 0.958   & 0.951   & 0.976  \\
Length         & 0.132   & 0.084  & 0.024  & 0.133   & 0.084   & 0.024  & 0.052   & 0.020   & 0.010  & 0.053   & 0.021   & 0.010  \\ \cline{2-13} 
$LS_1$      & 0.766   & 0.776  & 0.933  & 0.769   & 0.758   & 0.913  & 0.858   & 0.888   & 0.911  & 0.848   & 0.870   & 0.922  \\
Length         & 0.055   & 0.044  & 0.020  & 0.056   & 0.044   & 0.019  & 0.020   & 0.010   & 0.005  & 0.020   & 0.010   & 0.005  \\ \cline{2-13} 
$LS_2$         & 0.750   & 0.762  & 0.925  & 0.741   & 0.750   & 0.931  & 0.770   & 0.817   & 0.948  & 0.758   & 0.824   & 0.956  \\
Length         & 0.041   & 0.041  & 0.019  & 0.041   & 0.041   & 0.019  & 0.000   & 0.002   & 0.008  & 0.000   & 0.003   & 0.007  \\ \hline
\end{tabular}
\caption{Coverage probability and average length of seven  confidence intervals for the BD covariance model}\label{tab:CI_BD}
\end{table}

\begin{table}[h!]
\footnotesize
\begin{tabular}{c|ccc|ccc|ccc|ccc}
\hline
               & \multicolumn{6}{c|}{$p = 50$}                           & \multicolumn{6}{c}{$p = 150$}                          \\ \cline{2-13}
               & \multicolumn{3}{c|}{Dense} & \multicolumn{3}{c|}{Sparse} & \multicolumn{3}{c|}{Dense} & \multicolumn{3}{c}{Sparse} \\ \hline
$\tau_0 = 0.2$ &         & $n$     &        &         & $n$    &        &         & $n$     &         &         & $n$    &        \\
               & 50      & $100$   & $200$  & $50$    & $100$  & $200$  & $50$    & $100$   & $200$   & $50$    & $100$  & $200$  \\ \hline
$U_1$          & 1.000   & 0.818   & 0.777  & 1.000   & 0.800  & 0.738  & 1.000   & 0.738   & 0.746   & 1.000   & 0.719  & 0.743  \\
Length         & 0.949   & 0.754   & 0.108  & 0.952   & 0.748  & 0.108  & 1.000   & 0.508   & 0.157   & 1.000   & 0.504  & 0.157  \\ \cline{2-13} 
$U_2$          & 0.591   & 0.573   & 0.717  & 0.588   & 0.575  & 0.689  & 0.583   & 0.588   & 0.682   & 0.574   & 0.595  & 0.659  \\
Length         & 0.451   & 0.428   & 0.113  & 0.491   & 0.448  & 0.118  & 0.481   & 0.364   & 0.151   & 0.487   & 0.381  & 0.156  \\ \cline{2-13} 
$U_3$          & 0.558   & 0.584   & 0.857  & 0.553   & 0.592  & 0.824  & 0.502   & 0.648   & 0.797   & 0.523   & 0.635  & 0.792  \\
Length         & 0.357   & 0.376   & 0.277  & 0.375   & 0.364  & 0.279  & 0.355   & 0.353   & 0.300   & 0.369   & 0.342  & 0.310  \\ \cline{2-13} 
$U_4$          & 0.565   & 0.601   & 0.871  & 0.561   & 0.598  & 0.837  & 0.579   & 0.668   & 0.809   & 0.537   & 0.649  & 0.808  \\
Length         & 0.372   & 0.689   & 0.315  & 0.389   & 0.376  & 0.313  & 0.374   & 0.371   & 0.321   & 0.390   & 0.359  & 0.330  \\ \cline{2-13} 
$U_5$          & 0.548   & 0.578   & 0.859  & 0.547   & 0.582  & 0.826  & 0.505   & 0.653   & 0.796   & 0.523   & 0.632  & 0.794  \\
Length         & 0.351   & 0.373   & 0.286  & 0.367   & 0.362  & 0.288  & 0.355   & 0.354   & 0.297   & 0.371   & 0.344  & 0.307  \\ \cline{2-13} 
$LS_1$         & 0.135   & 0.122   & 0.442  & 0.137   & 0.105  & 0.391  & 0.118   & 0.115   & 0.281   & 0.092   & 0.128  & 0.235  \\
Length         & 0.027   & 0.023   & 0.016  & 0.028   & 0.022  & 0.016  & 0.020   & 0.010   & 0.005   & 0.020   & 0.010  & 0.005  \\ \cline{2-13} 
$LS_2$         & 0.420   & 0.394   & 0.736  & 0.394   & 0.383  & 0.694  & 0.375   & 0.460   & 0.662   & 0.392   & 0.464  & 0.623  \\
Length         & 0.241   & 0.222   & 0.090  & 0.250   & 0.215  & 0.104  & 0.232   & 0.188   & 0.115   & 0.246   & 0.191  & 0.125  \\ \hline
$\tau_0 = 0.5$ & \multicolumn{3}{c|}{$n$} & \multicolumn{3}{c|}{$n$} & \multicolumn{3}{c|}{$n$} & \multicolumn{3}{c}{$n$} \\
               & $50$    & $100$   & $200$  & $50$    & $100$  & $200$  & $50$    & $100$   & $200$   & $50$    & $100$  & $200$  \\ \hline
$U_1$          & 1.000   & 0.848   & 0.835  & 1.000   & 0.865  & 0.767  & 1.000   & 0.770   & 0.811   & 1.000   & 0.778  & 0.804  \\
Length         & 0.831   & 0.602   & 0.069  & 0.831   & 0.604  & 0.069  & 0.911   & 0.368   & 0.101   & 0.915   & 0.369  & 0.101  \\ \cline{2-13} 
$U_2$          & 0.608   & 0.589   & 0.813  & 0.602   & 0.588  & 0.734  & 0.578   & 0.651   & 0.775   & 0.586   & 0.677  & 0.761  \\
Length         & 0.467   & 0.397   & 0.071  & 0.443   & 0.409  & 0.073  & 0.457   & 0.314   & 0.100   & 0.465   & 0.336  & 0.101  \\ \cline{2-13} 
$U_3$          & 0.652   & 0.663   & 0.950  & 0.648   & 0.656  & 0.922  & 0.590   & 0.750   & 0.911   & 0.600   & 0.766  & 0.910  \\
Length         & 0.505   & 0.518   & 0.198  & 0.502   & 0.526  & 0.240  & 0.473   & 0.494   & 0.342   & 0.476   & 0.503  & 0.357  \\ \cline{2-13} 
$U_4$          & 0.674   & 0.683   & 0.968  & 0.668   & 0.673  & 0.942  & 0.617   & 0.777   & 0.921   & 0.626   & 0.782  & 0.923  \\
Length         & 0.566   & 0.570   & 0.613  & 0.562   & 0.575  & 0.592  & 0.535   & 0.617   & 0.660   & 0.541   & 0.623  & 0.658  \\ \cline{2-13} 
$U_5$          & 0.642   & 0.659   & 0.950  & 0.630   & 0.650  & 0.927  & 0.586   & 0.758   & 0.908   & 0.602   & 0.770  & 0.910  \\
Length         & 0.496   & 0.511   & 0.248  & 0.493   & 0.518  & 0.291  & 0.479   & 0.521   & 0.341   & 0.485   & 0.529  & 0.355  \\ \cline{2-13} 
$LS_1$       & 0.218   & 0.190   & 0.667  & 0.219   & 0.206  & 0.574  & 0.136   & 0.216   & 0.456   & 0.157   & 0.258  & 0.423  \\
Length         & 0.032   & 0.027   & 0.019  & 0.031   & 0.027  & 0.019  & 0.020   & 0.010   & 0.005   & 0.020   & 0.010  & 0.005  \\ \cline{2-13} 
$LS_2$         & 0.527   & 0.513   & 0.877  & 0.560   & 0.536  & 0.861  & 0.499   & 0.627   & 0.825   & 0.497   & 0.618  & 0.840  \\
Length         & 0.248   & 0.229   & 0.064  & 0.277   & 0.248  & 0.090  & 0.255   & 0.171   & 0.083   & 0.247   & 0.165  & 0.094  \\ \hline
\end{tabular}
\caption{Coverage probability and average length of seven  confidence intervals for the CS covariance model}\label{tab:CI_CS}
\end{table}

\begin{table}[h!]
\footnotesize
\centering
\begin{tabular}{c|ccc|ccc|ccc}
\hline
$\tau_0 = 0.2$             & \multicolumn{3}{c|}{Weak} & \multicolumn{3}{c|}{Moderate} & \multicolumn{3}{c}{Strong} \\ \hline
\multirow{2}{*}{$p = 50$}  & \multicolumn{3}{c|}{$n$}  & \multicolumn{3}{c|}{$n$}      & \multicolumn{3}{c}{$n$}    \\ \cline{2-10} 
                           & 50      & $100$  & $200$  & $50$     & $100$    & $200$   & $50$    & $100$   & $200$   \\ \hline
$U_1$                      & 0.883   & 0.946  & 0.924  & 0.809    & 0.881    & 0.726   & 0.778   & 0.759   & 0.594   \\
Length                     & 0.087   & 0.022  & 0.005  & 0.087    & 0.022    & 0.005   & 0.087   & 0.022   & 0.005   \\ \cline{2-10} 
$U_2$                      & 0.859   & 0.909  & 0.924  & 0.797    & 0.814    & 0.730   & 0.770   & 0.668   & 0.614   \\
Length                     & 0.118   & 0.024  & 0.006  & 0.131    & 0.026    & 0.006   & 0.132   & 0.027   & 0.006   \\ \cline{2-10} 
$U_3$                      & 0.862   & 0.959  & 0.991  & 0.790    & 0.928    & 0.947   & 0.808   & 0.892   & 0.939   \\
Length                     & 0.067   & 0.022  & 0.010  & 0.091    & 0.037    & 0.017   & 0.109   & 0.053   & 0.027   \\ \cline{2-10} 
$U_4$                      & 0.943   & 0.991  & 1.000  & 0.891    & 0.963    & 0.962   & 0.888   & 0.907   & 0.918   \\
Length                     & 0.247   & 0.063  & 0.021  & 0.241    & 0.068    & 0.021   & 0.234   & 0.074   & 0.021   \\ \cline{2-10} 
$U_5$                      & 0.918   & 0.976  & 0.991  & 0.862    & 0.948    & 0.957   & 0.863   & 0.906   & 0.943   \\
Length                     & 0.139   & 0.031  & 0.010  & 0.168    & 0.047    & 0.019   & 0.181   & 0.064   & 0.028   \\ \cline{2-10} 
$LS_1$                     & 0.718   & 0.945  & 0.995  & 0.610    & 0.874    & 0.935   & 0.559   & 0.738   & 0.862   \\
Length                     & 0.024   & 0.025  & 0.015  & 0.027    & 0.023    & 0.015   & 0.027   & 0.022   & 0.014   \\ \cline{2-10} 
$LS_2$                     & 0.723   & 0.878  & 0.934  & 0.703    & 0.888    & 0.919   & 0.683   & 0.833   & 0.921   \\
Length                     & 0.020   & 0.006  & 0.000  & 0.041    & 0.021    & 0.011   & 0.048   & 0.035   & 0.021   \\ \hline
\multirow{2}{*}{$p = 150$} & \multicolumn{3}{c|}{$n$}  & \multicolumn{3}{c|}{$n$}      & \multicolumn{3}{c}{$n$}    \\ \cline{2-10} 
                           & $50$    & $100$  & $200$  & $50$     & $100$    & $200$   & $50$    & $100$   & $200$   \\ \hline
$U_1$                      & 0.802   & 0.916  & 0.916  & 0.795    & 0.861    & 0.839   & 0.786   & 0.822   & 0.738   \\
Length                     & 0.261   & 0.066  & 0.016  & 0.261    & 0.066    & 0.016   & 0.261   & 0.066   & 0.016   \\ \cline{2-10} 
$U_2$                      & 0.801   & 0.908  & 0.926  & 0.789    & 0.845    & 0.863   & 0.763   & 0.806   & 0.761   \\
Length                     & 0.305   & 0.076  & 0.017  & 0.283    & 0.079    & 0.018   & 0.281   & 0.082   & 0.018   \\ \cline{2-10} 
$U_3$                      & 0.523   & 0.804  & 0.923  & 0.484    & 0.746    & 0.907   & 0.523   & 0.717   & 0.860   \\
Length                     & 0.046   & 0.024  & 0.011  & 0.047    & 0.031    & 0.018   & 0.052   & 0.036   & 0.022   \\ \cline{2-10} 
$U_4$                      & 0.820   & 0.965  & 0.987  & 0.791    & 0.942    & 0.971   & 0.792   & 0.900   & 0.918   \\
Length                     & 0.321   & 0.169  & 0.034  & 0.313    & 0.168    & 0.036   & 0.311   & 0.164   & 0.036   \\ \cline{2-10} 
$U_5$                      & 0.728   & 0.893  & 0.949  & 0.690    & 0.850    & 0.935   & 0.707   & 0.838   & 0.904   \\
Length                     & 0.209   & 0.061  & 0.018  & 0.203    & 0.066    & 0.024   & 0.210   & 0.072   & 0.029   \\ \cline{2-10} 
$LS_1$                     & 0.421   & 0.651  & 0.831  & 0.379    & 0.577    & 0.744   & 0.401   & 0.550   & 0.658   \\
Length                     & 0.020   & 0.010  & 0.006  & 0.020    & 0.010    & 0.007   & 0.020   & 0.010   & 0.007   \\ \cline{2-10} 
$LS_2$                     & 0.330   & 0.631  & 0.836  & 0.302    & 0.692    & 0.839   & 0.352   & 0.655   & 0.788   \\
Length                     & 0.007   & 0.007  & 0.005  & 0.009    & 0.019    & 0.010   & 0.015   & 0.020   & 0.014   \\ \hline
\end{tabular}
\caption{Coverage probability and average length of seven  confidence intervals for different interactions}\label{tab:CI_int}
\end{table}

\begin{figure}[h!]
    \centering
    \includegraphics[scale = 0.4]{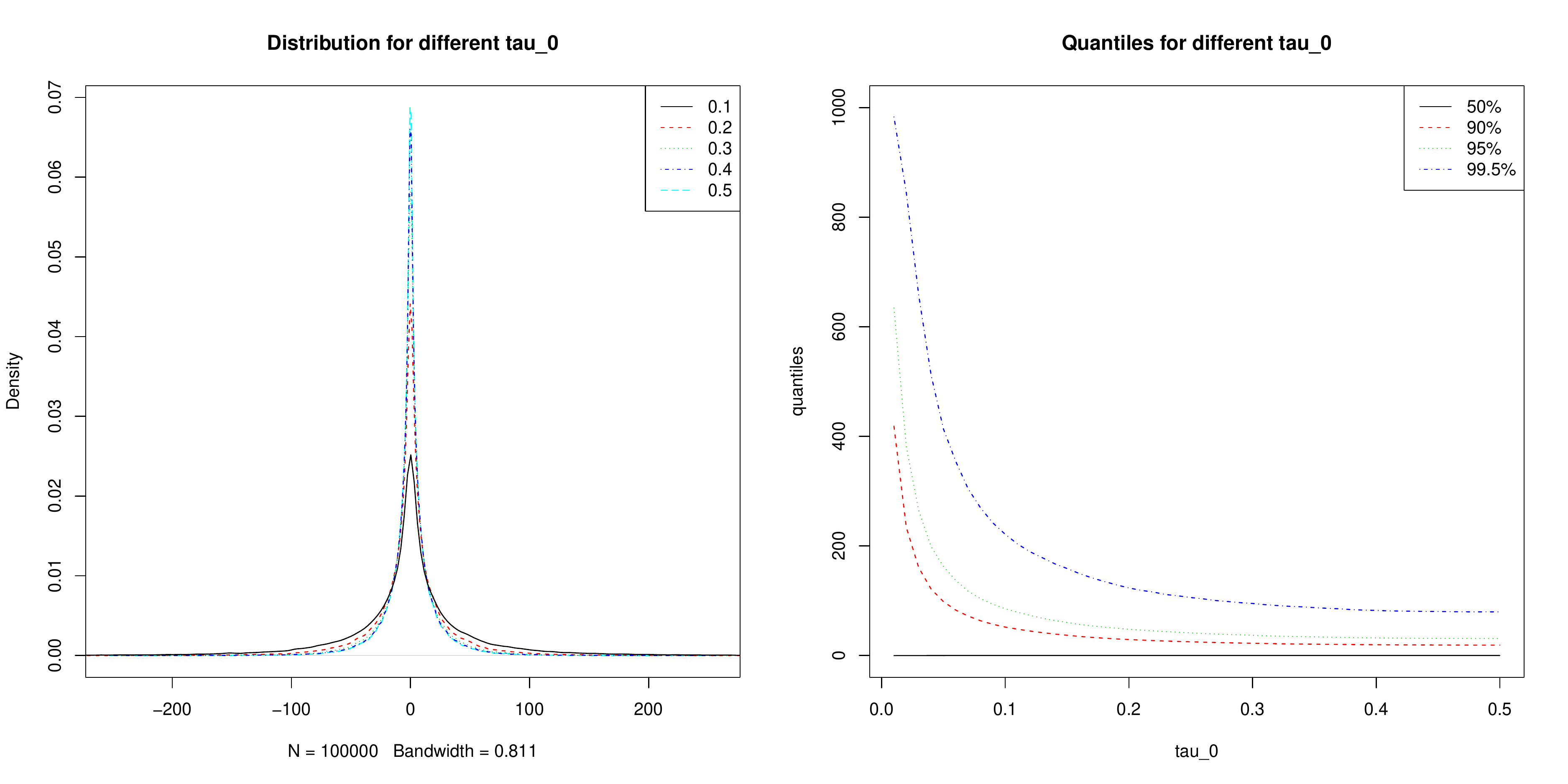}
    \caption{Density plot and quantile plot for different values of $\tau_0$}
    \label{fig:dist}
\end{figure}
\newpage

\begin{algorithm}[h!]
	\begin{enumerate}
		\item Estimate $k_0$ by $\hat{k}_{LS}$ in \cite{bai2010common}.
		\item Estimate the pre-mean and post-mean by sample average for the pre-sample and the post-sample using $\hat{k}_{LS}$ as the break point. Denote the estimation as $\hat{\mu}_{pre}$ and $\hat{\mu}_{post}$.
		\item Estimate the variance $\sigma_j^2$ by $\hat{\sigma}_j^2$ for each individual series:
		$$\hat{\sigma}_j^2 = \left(\sum_{i = 1}^{\hat{k}_{LS}}(X_{i,j} - \hat{\mu}_{pre,j})^2 +\sum_{i = \hat{k}_{LS}+1}^{n}(X_{i,j} - \hat{\mu}_{post,j})^2\right)/(n-2), $$
		for all $j = 1,2,...,p$.
		
		\item Estimate $A_p$ by $\hat{A}_p$ where 
		$$\hat{A}_p = \frac{\|\hat{\mu}_{post} - \hat{\mu}_{pre}\|_2^4}{\sum_{j = 1}^{p}(\hat{\mu}_{post,j} - \hat{\mu}_{pre,j})^2\hat{\sigma}_j^2}.$$
		
		\item A 95\% CI for $\tau_0$: $[\hat{\tau}_{LS} - \floor{11/(n\hat{A}_p)}, \hat{\tau}_{LS} + \lceil11/(n\hat{A}_p) \rceil]$, where $\hat{\tau}_{LS} = \hat{k}_{LS}/n$.
	\end{enumerate}
	\caption{Algorithm for constructing a confidence interval for $\tau_0$ in \cite{bai2010common}}\label{alg:CI_Bai}
\end{algorithm}
\begin{algorithm}[h!]
	\begin{enumerate}
		\item  Estimate $\tau_0$ by $\hat{\tau}_{BBM} = \hat{k}_{BBM}/n$ in \cite{bhattacharjee2019change}
		
		\item Estimate the pre-mean and post-mean by sample average for the pre-sample and the post-sample using $\hat{k}_{BBM}$ as the break point. Denote the estimation as $\hat{\mu}_{pre}$ and $\hat{\mu}_{post}$.
		
		\item Estimate $\Sigma$ by some positive semi-definite estimator $\hat{\Sigma}_X$.
		
		\item Generate random vectors $\epsilon_1$,...,$\epsilon_n$ in $\mathbb{R}^p$ from  distribution $\mathcal{N}(0,\hat{\Sigma}_X)$.
		
		\item Generate ${X}_t^* = \hat{\mu}_{pre} + \epsilon_t$ if $t \leq \hat{k}_{BBM}$ and ${X}_t^* = \hat{\mu}_{post} + \epsilon_t$ if $t > \hat{k}_{BBM}$.
		
		\item Estimate $\hat{h}$ by $\hat{h} = \argmin_{h \in (n(c^* - \hat{\tau}_{BBM}),n(1 - c^* - \hat{\tau}_{BBM}))}\hat{L}(h),$ where
		$$\hat{L}(h) = \frac{1}{n}\sum_{j = 1}^{p}\left[\sum_{t = 1}^{n\hat{\tau}_{BBM} + h}\left({X}^*_{t,j} - \hat{\mu}_{pre,j}\right)^2 + \sum_{t = n\hat{\tau}_{BBM} + h + 1}^{n}\left({X}^*_{t,j} - \hat{\mu}_{post,j}\right)^2\right].$$
		
		\item Repeat step 4-6 for $B$ times to generate $\hat{h}_1$,...,$\hat{h}_B$, and 95\% CI for $\tau_0$ is $[\hat{\tau}_{BBM}-q^*_{0.975}/n, \hat{\tau}_{BBM}-q^*_{0.025}/n]$, where ${q}^*_{0.025}$ and $q^*_{0.975}$ are the sample $2.5\%$ and $97.5\%$ quantiles based on $\hat{h}_1$,...,$\hat{h}_B$.
	\end{enumerate}
	\caption{Algorithm for constructing a confidence interval for $\tau_0$ in \cite{bhattacharjee2019change} (modified for independent data)}\label{alg:CI_BBM}
\end{algorithm}

\newpage
\clearpage
\bibliography{cploc}

\clearpage
\section{Technical Appendix A}
\label{sec:appendix}
In this section, we gather some auxiliary results in Section \ref{sec:result}, and present the proofs of all main theorems and corollaries. All the constants $C$, $C_1$, $C_2$ ... stated in the appendix are generic and their specific values may vary from line to line and are not important.
\subsection{Preliminary Results} \label{sec:result}
\begin{lemma}\label{lem:order}
	For any $1 \leq n_1 < n_2 \leq n$, under Assumption \ref{ass}\ref{ass:tr} and \ref{ass}\ref{ass:cum}, we have 
	$$\sum_{i_1,i_2,i_3,i_4 = n_1+1}^{n_2}\sqrt{\sum_{l_1,l_2,l_3,l_4 = 1}^{p}(\E[Z_{i_1,l_1}Z_{i_2,l_2}Z_{i_3,l_3}Z_{i_4,l_4}])^2} \leq C(n_2-n_1)^2\|\Sigma\|_F^2.$$
\end{lemma}

	\begin{lemma}\label{lem:Xi1Xi4}
	Under Assumption \ref{ass}\ref{ass:tr} and \ref{ass}\ref{ass:cum}, there exists a constant $C<\infty$ such that for all $j_1 \leq i_1, \cdots, j_4 \leq i_4$,
	\[
	\Big| \sum_{l_1,\cdots,l_4 = 1}^{p}\E[X_{i_1+1,l_1}X_{j_1,l_1}\cdots X_{i_4+1,l_4}X_{j_4,l_4}]  \Big| \leq C\|\Sigma\|_F^4.
	\]
\end{lemma}

The following identities will be used several times in the proof and are displayed in the following proposition.
\begin{proposition}\label{prop:terms} For any $k \leq k_0$,
\begin{enumerate}
    \item 
	\begin{flalign*}
	    G_n(k) &= G_n^Z(k) + \E[G_n(k)] - \frac{2(k-1)(n-k-1)(n-k_0)}{k(n-k)}\sum_{i = 1}^{k}\delta^TZ_i&\\
	    &+\frac{2(k-1)(n-k_0-1)}{n-k}\sum_{j = k_0+1}^{n}\delta^TZ_j + \frac{2(k-1)(n-k_0)}{n-k}\sum_{k+1}^{k_0}\delta^TZ_j.&
	\end{flalign*}
	\item \begin{flalign*}
&G_n^Z(k_0) - G_n^Z(k)&\\
= & 2\frac{(n-k-1)(n-1)}{(n-k)k}\sum_{i = k+1}^{k_0}\sum_{j = 1}^{k}Z_i^TZ_j - 2\frac{(k_0-1)(n-1)}{(n-k_0)k_0}\sum_{i = k+1}^{k_0}\sum_{j = k_0+1}^{n}Z_i^TZ_j&\\
+ &\frac{(n-1)(n-k-k_0)}{k(n-k_0)}\sum_{i,j = k+1, i\neq j}^{k_0}Z_i^TZ_j- \frac{(n-1)(k_0-k)}{kk_0}\sum_{i,j=1,i\neq j}^{k_0}Z_i^TZ_j&\\
+&\frac{(n-1)(k_0-k)}{(n-k_0)(n-k)}\sum_{i,j = k+1,i\neq j}^{n}Z_i^TZ_j - \frac{2(n-1)(n-k_0-k)(k_0-k)}{kk_0(n-k)(n-k_0)}\sum_{i = 1}^{k}\sum_{j = k_0+1}^{n}Z_i^TZ_j&\\
:=&S_{1,n}(k)+S_{2,n}(k)+S_{3,n}(k)+S_{4,n}(k)+S_{5,n}(k)+S_{6,n}(k).
\end{flalign*}
For the simplicity of the notations we denote $S_{i,n}(k)$ as $S_i(k)$ for $i = 1,2,...,6$.

\end{enumerate}
	
\end{proposition}

\begin{proposition} \label{prop:Q}
	Under Assumptions \ref{ass}\ref{ass:tr} and \ref{ass}\ref{ass:cum}, as $n \wedge p \rightarrow \infty$, for any $0 \leq a < b \leq 1$,
	$$\frac{\sqrt{2}}{n\|\Sigma\|_F}\sum_{i = \floor{na}+1}^{\floor{nb}-1}\sum_{j = \floor{na}+1}^{i}Z_{i+1}^TZ_j\rightsquigarrow Q(a,b)\text{ in } l_{\infty}([0,1]^2),$$
	where $Q(a,b)$ is a centered Gaussian process on $[0,1]^2$ with covariance structure given by
	$$Cov\{Q(a_1,b_1),Q(a_2,b_2)\} = (b_1\wedge b_2 - a_1 \vee a_2)^2\1{b_1\wedge b_2 > a_1 \vee a_2}.$$ 
\end{proposition}

\begin{remark}
	The centered Gaussian process $Q$ can be regarded as a 2-D analogue of the standard Brownian motion. Suppose $M_n$ is an $n$-by-$n$ matrix containing i.i.d. standard normal random variables, and we take $Q_n(a,b)$ as the standardized sum of all variables of $M_n$ in the region bounded by rows $\floor{na}+1$, $\floor{nb}$ and columns $\floor{na}+1$, $\floor{nb}$, for any $0 \leq a < b \leq 1$. As $n \rightarrow \infty$, $Q_n(a,b) \rightsquigarrow Q(a,b)$. The proof of Proposition \ref{prop:Q} can be found in \cite{wang2019a}.
\end{remark}

\begin{proposition}[Tightness]\label{prop:tightness}
	Define $$H_n(\gamma) = \frac{\sqrt{2}\sqrt{b_n}}{n\|\Sigma\|_F}\left\{G_n^Z({n\tau_0}) - G_n^Z(\floor{n\tau_0 + n\gamma/b_n})\right\}$$
	for all $\gamma \in [-M,M]$. For any $\gamma_1,\gamma_2 \in [-M,M]$ and $\gamma_1 \neq \gamma_2$, denote $k_1 = \floor{n\tau_0 + n\gamma_1/b_n}$ and $k_2 = \floor{n\tau_0+n\gamma_2/b_n}$. Let $b_n$ satisfy $1/b_n + b_n/n = O(1)$. Under Assumption \ref{ass}\ref{ass:tr} and \ref{ass}\ref{ass:cum}, we have
	$$\E\left[\left\{H_n(\gamma_2) - H_n(\gamma_1)\right\}^4\right] \leq Cb_n^2(k_2 - k_1)^2/n^2$$
	for all sufficiently large $n$ and some positive constant $C$.
\end{proposition}

\begin{remark}
	By Lemma 9.8 in \cite{wangshao19}, the assertion in Proposition \ref{prop:tightness} is a sufficient condition to show the tightness. 
\end{remark}

\begin{lemma}[H\'ajek-R\'enyi's inequality (\cite{birnbaum1961some} or \cite{bai1994least})]\label{HR1}
	Assume that $\{\epsilon_t\}$ is martingale difference sequence with variance $E(\epsilon_t^2) = \sigma^2_t$, and $\{c_k\}$ is a non-increasing positive sequence of constants. Then for $\alpha > 0$,
	$$P\left(\max_{m \leq k \leq n}c_k\left|\sum_{i = 1}^{k}\epsilon_i\right| \geq \alpha\right) \leq \frac{1}{\alpha^2}\left(c_m^2\sum_{i = 1}^{m}\sigma_i^2 + \sum_{i = m+1}^{n}c_i^2\sigma_i^2\right).$$
	Specifically, if $c_k = 1/k$, 
	$$P\left(\max_{m \leq k \leq n}\frac{1}{k}\left|\sum_{i = 1}^{k}\epsilon_i\right| \geq \alpha\right) \leq \frac{1}{\alpha^2}\left(\frac{1}{m^2}\sum_{i = 1}^{m}\sigma_i^2 + \sum_{i = m+1}^{n}\frac{\sigma_i^2}{i^2}\right).$$
\end{lemma}

\begin{proposition}\label{lem:S1}
    Under Assumption \ref{ass}, for any positive $\eta$ and $\epsilon$, there exists $M_0 > 0$ such that for all $M > M_0$ and sufficiently large $n$ and $p$,
$$P\left(\sup_{k \in [k_0/2, k_0-nM/a_n]}\frac{1}{(k_0-k)\|\Sigma\|_F}\left|\left(\sum_{i = k+1}^{k_0}Z_i^T\sum_{j = 1}^{k}Z_j\right)\right| > \eta \sqrt{a_n}\right) < \epsilon.$$
\end{proposition}

\begin{lemma}\label{lem:rateterms}
	Under Assumptions \ref{ass}, for any fixed positive constant $M$, we have the following results for some positive constants $C_1,C_2,C_3$,
	\begin{enumerate}[label = (\alph*)]
	    \item $\max_{1\leq k_1 < k_2 \leq n}\left|\sum_{i = k_1+1}^{k_2}\delta^TZ_j\right| = o_p(n\|\delta\|^2/\sqrt{a_n})$. \label{lem:sumdelta}
	    \item $\max_{1\leq k_1<k_2 \leq n}\left|\sum_{i = k_1}^{k_2}\sum_{j = k_1}^{i}Z_{i+1}^TZ_j\right| = O_p(n\|\Sigma\|_F)$.\label{lem:sum}
	    \item $\max_{1\leq k_1 < k_2 < k_3\leq n}\left|\sum_{i = k_1}^{k_2}\sum_{j = k_2+1}^{k_3}Z_{i}^TZ_j\right| = O_p(n\|\Sigma\|_F)$.\label{lem:sumcross}
	    \item $\max_{k_1\leq k \leq k_2}\frac{1}{k}\left|\sum_{i = 1}^{k}\delta^TZ_i\right| = o_p(\sqrt{n}\|\delta\|^2/\sqrt{a_nk_1})$, for any $1 \leq k_1 \leq k_2 \leq n$. \label{lem:avedelta1}
	    \item $\max_{1\leq k \leq n}\frac{1}{n-k}\left|\sum_{i = k+1}^{n}\delta^TZ_i\right| = o_p(\sqrt{n}\|\delta\|^2/\sqrt{a_n})$.\label{lem:avedelta2}
	    \item $P\left(\max_{1 \leq k \leq k_0 - nM/a_n}\frac{1}{k_0-k}|\sum_{i = k+1}^{k_0}\sum_{k_0+1}^{n}Z_i^TZ_j| > \lambda\right) \leq \frac{C_1(n-k_0)\|\Sigma\|_F^2}{\lambda^2}\left(\frac{a_n}{nM}\right)$.\label{lem:HRcross}
	    \item $P\left(\max_{1 \leq k \leq k_0 - nM/a_n}|\sum_{i = 1}^{k}\sum_{j = k_0+1}^{n}Z_i^TZ_j| > \lambda\right) \leq \frac{C_2k_0(n-k_0)\|\Sigma\|_F^2}{\lambda^2}$.\label{lem:HRcrosssum}
	    \item $P\left(\max_{1 \leq k \leq k_0 - nM/a_n}\frac{1}{k_0-k}|\sum_{i,j = k+1,i\neq j}^{k_0}Z_i^TZ_j| > \lambda\right) \leq \frac{C_3\log(k_0)\|\Sigma\|_F^2}{\lambda^2}$.
	    \label{lem:HRmid}
	\end{enumerate}

\end{lemma}

\subsection{Proof of Lemma \ref{lem:max}}
Consider the case $k \leq k_0$ first. By the definition of $G_n(k)$,
\begin{align*}
\E[G_n(k)] &= \frac{1}{k(n-k)}\sum_{i_1,i_2 = 1, i_1 \neq i_2}^k\sum_{j_1,j_2 = k_0+1}^{n}\|\delta\|^2 = \frac{k(k-1)(n-k_0)(n-k_0-1)}{k(n-k)}\|\delta\|^2 \\
&= \frac{(k-1)(n-k_0)(n-k_0-1)}{(n-k)}\|\delta\|^2,
\end{align*}
which is an increasing function of $k$. Thus it achieves its maximum at $k = k_0$ where $\E[G_n(k_0)] = (k_0-1)(n-k_0-1)\|\delta\|^2$. By similar arguments, we see that when $k \geq k_0$, $\E[G_n(k)] = (n-k-1)k_0(k_0-1)\|\delta\|^2/k$, which achives the maximum when $k = k_0$. This completes the proof.

\subsection{Proof of Theorem \ref{thm:rate}}
Assume $k \leq k_0$ first and we need to show for any $\epsilon>0$, there exists $M_0,N_0 > 0$, such that for any $M > \max(M_0,1)$, $n > N_0$,
$$P\left(\max_{k \in \Omega_n(M)}G_n(k) \geq G_n(k_0) \right) < \epsilon,$$
where $\Omega_n(M) = \{k:1 \leq k < k_0 - nM/a_n\}$. We further decompose $\Omega_n(M) = \bigcup_{i = 1}^{3}\Omega_n^{(i)}(M)$, where $\Omega_n^{(1)}(M) = \{k:1 \leq k < k_0/2\}$, $\Omega_n^{(2)}(M) = \{k:k_0/2 \leq k \leq k_0 - n/\sqrt{a_n}\}$ and $\Omega_n^{(3)}(M) = \{k: k_0 - n/\sqrt{a_n} < k < k_0 - nM/a_n\}$. It is easy to see that the three sets $\Omega_n^{(1)}(M), \Omega_n^{(2)}(M)$ and $\Omega_n^{(3)}(M)$ are disjoint for large enough $n$, say $n > N_0$.

For $\Omega_n^{(1)}(M)$,
\begin{align*}
&P\left(\max_{k \in \Omega_n^{(1)}(M)}G_n(k) \geq G_n(k_0) \right)\\
= &P\left(\max_{k \in \Omega_n^{(1)}(M)}\{G_n(k) - \E[G_n(k)] - G_n(k_0) + \E[G_n(k_0)] + \E[G_n(k)] - \E[G_n(k_0)]\} \geq 0\right)\\
\leq & P\left(\max_{k = 1,2,...,k_0}(|G_n(k) - \E[G_n(k)]| + |G_n(k_0) - \E[G_n(k_0)]|) + \max_{k \in \Omega_n^{(1)}(M)}(\E[G_n(k)] - \E[G_n(k_0)]) \geq 0\right)\\
\end{align*}

Since $\E[G_n(k)] = (k-1)(n-k_0)(n-k_0-1)\|\delta\|^2/(n-k)$ as stated in Lemma \ref{lem:max},
\begin{align*}
\max_{k \in \Omega_n^{(1)}(M)}\left(\E[G_n(k)] - \E[G_n(k_0)]\right) =& \E[G_n(k_0/2)] - \E[G_n(k_0)] = -\frac{(n-k_0-1)(n-1)k_0}{2(n-k_0/2)}\|\delta\|^2.
\end{align*}

Then
\begin{align*}
&P\left(\max_{k \in \Omega_n^{(1)}(M)}G_n(k) \geq G_n(k_0) \right)\\
\leq & P\left(2\max_{k = 1,2,...k_0}|G_n(k) - \E[G_n(k)]| \geq \frac{(n-k_0-1)(n-1)k_0}{2(n-k_0/2)}\|\delta\|^2\right).
\end{align*}

Notice that by Proposition \ref{prop:terms},
\begin{align*}
\max_{k = 1,2,...,k_0}|G_n(k) - \E[G_n(k)]| &\leq  \max_{k = 1,2,...,k_0}|G_n^Z(k)| + \max_{k = 1,...,k_0}\left|\frac{2(n-k_0)(n-k-1)(k-1)}{k(n-k)}\sum_{j = 1}^{k}\delta^TZ_j\right| \\
&+ \max_{k = 1,...,k_0}\left|\frac{2(k-1)(n-k_0-1)}{n-k}\sum_{j = k_0+1}^{n}\delta^TZ_j\right| \\
&+ \max_{k = 1,...,k_0}\left|\frac{2(k-1)(n-k_0)}{n-k}\sum_{j = k+1}^{k_0}\delta^TZ_{j}\right|.
\end{align*}

By Lemma \ref{lem:rateterms}\ref{lem:sumdelta}, the last three terms in the above inequalities are all $o_p(n^2\|\delta\|^2/\sqrt{a_n})$. In addition, 
\begin{align*}
\max_{k \leq k_0}|G_n^Z(k)| &\leq \max_{k \leq k_0}\left|\frac{2(n-k)(n-k-1)}{k(n-k)}\sum_{i = 1}^{k-1}\sum_{j = 1}^{i}Z_{i+1}^TZ_j\right| + \max_{k \leq k_0}\left|\frac{2k(k-1)}{k(n-k)}\sum_{i = k+1}^{n-1}\sum_{j = k+1}^{i}Z_{i+1}^TZ_j\right|\\
&+ \max_{k \leq k_0}\left|\frac{2(k-1)(n-k-1)}{k(n-k)}\sum_{i = 1}^{k}\sum_{j = k+1}^{n}Z_{i+1}^TZ_j\right|\\
&= O_p(n\sqrt{\log(n)}\|\Sigma\|_F) + O_p(n\|\Sigma\|_F) + O_p(n\|\Sigma\|_F) = O_p(n\sqrt{\log(n)}\|\Sigma\|_F),
\end{align*} 
where the bound for the first term is due to Lemma \ref{HR1} by letting $c_k = 1/k$ and $\epsilon_i = Z_{i+1}^T\sum_{j = 1}^iZ_j$, and the bounds for the second and third term are due to \ref{lem:sum} and \ref{lem:sumcross} in Lemma \ref{lem:rateterms}.

Hence
\begin{align*}
&P\left(\max_{k \in \Omega_n^{(1)}(M)}G_n(k) \geq G_n(k_0) \right)\\
\leq & P\left(2\max_{k = 1,2,...,k_0}|G_n(k) - \E[G_n(k)]| \geq \frac{(n-k_0-1)(n-1)k_0}{2(n-k_0/2)}\|\delta\|^2\right)\\
=& P\left(\frac{2\max_{k = 1,2,...,k_0}|G_n(k) - \E[G_n(k)]|}{\frac{(n-k_0-1)(n-1)k_0}{2(n-k_0/2)}\|\delta\|^2} \geq 1 \right)\rightarrow 0
\end{align*}
for sufficiently large $n$, since under Assumption \ref{ass}\ref{ass:delta} and \ref{ass}\ref{ass:rate},
\begin{align*}
    \frac{2\max_{k = 1,2,...,k_0}|G_n(k) - \E[G_n(k)]|}{\frac{(n-k_0-1)(n-1)k_0}{2(n-k_0/2)}\|\delta\|^2}
    &= \frac{\{O_p(n\sqrt{\log(n)}\|\Sigma\|_F) + o_p(n^{2}\|\delta\|^2/\sqrt{a_n})\}(n-k_0/2)}{(n-k_0-1)(n-1)k_0\|\delta\|^2}\\
    &= O_p\left(\frac{\sqrt{\log(n)}}{\sqrt{a_n}}\right) + o_p(a_n^{-1/2}) = o_p(1).
\end{align*}

For $k \in \Omega_n^{(2)}(M) \bigcup \Omega_n^{(3)}(M)$, we need to decompose $G_n(k_0) - G_n(k)$ further as 
\begin{align*}
&G_n(k_0) - G_n(k)\\
= & G_n^Z(k_0) - G_n^Z(k) + \E[G_n(k_0)] - \E[G_n(k)]- \frac{2(k_0-1)(n-k_0-1)}{k_0}\sum_{i = 1}^{k_0}\delta^TZ_i \\
&+ \frac{2(k_0-1)(n-k_0-1)}{n-k_0}\sum_{j = k_0+1}^{n}\delta^TZ_j +\frac{2(n-k-1)(k-1)(n-k_0)}{k(n-k)}\sum_{i = 1}^{k}\delta^TZ_i \\
&- \frac{2(k-1)(n-k_0-1)}{n-k}\sum_{j = k_0+1}^{n}\delta^TZ_j - \frac{2(k-1)(n-k_0)}{n-k}\sum_{j = k+1}^{k_0}\delta^TZ_j \\
= & G_n^Z(k_0) - G_n^Z(k) + (k_0-k)\frac{(n-k_0-1)(n-1)}{n-k}\|\delta\|^2+\frac{2(n-1)(k_0-k)(k_0+k-n)}{k_0(n-k)}\frac{1}{k}\sum_{i = 1}^{k}\delta^TZ_{i}\\
&+\frac{2(n-k_0-1)(n-1)(k-k_0)}{n-k}\frac{1}{n-k_0}\sum_{i = k_0+1}^{n}\delta^TZ_i\\
&+\frac{2(n-1)(k_0^2 - nk_0 - k + n)}{k_0(n-k)}\sum_{i = k+1}^{k_0}\delta^TZ_i\\
=&G_n^Z(k_0) - G_n^Z(k) + (k_0-k)\frac{(n-k_0-1)(n-1)}{n-k}\|\delta\|^2 + R_1(k) + R_2(k) + R_3(k).
\end{align*}

Observing that the third term is always nonnegative, we want to show that it dominates the other terms with probability converging to 1, for every $k \in \Omega_n^{(2)}(M) \bigcup \Omega_n^{(3)}(M)$. Then $G_n(k_0) - G_n(k)$ is nonnegative for every $k \in \Omega_n^{(2)}(M) \bigcup \Omega_n^{(3)}(M)$  with probability converging to 1. Specifically we want to show that for any fixed $\eta > 0$ and for any $\epsilon > 0$, when $n$,$p$ and $M$ are sufficiently large,
\begin{equation} \label{eq:R}
  P\left(\max_{k \in \Omega_n^{(2)}(M) \bigcup \Omega_n^{(3)}(M)}\frac{|R_i(k)|}{(k_0-k)(n-k_0-1)(n-1)\|\delta\|^2/(n-k)} > \eta\right) \leq \epsilon  
\end{equation} for all $i = 1,2,3$, and
\begin{align}\label{eq:Gz}
	&P\left(\max_{k \in \Omega_n^{(2)}(M) \bigcup \Omega_n^{(3)}(M)} \frac{|G_n^Z(k_0) - G_n^Z(k)|}{(k_0-k)(n-k_0-1)(n-1)\|\delta\|^2/(n-k)} > \eta\right) < \epsilon.
	\end{align}
To verify Equation (\ref{eq:R}), it follows from Lemma \ref{lem:rateterms}\ref{lem:sumdelta} that for $i = 1,2$, $$\max_{k \in \Omega_n^{(2)}(M) \bigcup \Omega_n^{(3)}(M)}R_i(k)/(k_0-k) = o_p(n\|\delta\|^2/\sqrt{a_n}).$$ Hence for $i = 1,2$,

\begin{align*}
    &\max_{k \in \Omega_n^{(2)}(M) \bigcup \Omega_n^{(3)}(M)}\frac{R_i(k)}{(k_0-k)(n-k_0-1)(n-1)\|\delta\|^2/(n-k)} \\
    \leq &\frac{n}{(n-k_0-1)(n-1)\|\delta\|^2}\max_{k \in \Omega_n^{(2)}(M) \bigcup \Omega_n^{(3)}(M)}R_i(k)/(k_0-k) = o_p(1/\sqrt{a_n}) = o_p(1).
\end{align*}

For $R_3(k)$, we apply Lemma \ref{HR1} by setting $c_k$ = $(k_0-k)^{-1}$, summation from ${nM}/{a_n}$ to $k_0/2$ and $\epsilon_i = \delta^TZ_i$,
\begin{align*}
    &P\left(\max_{k \in \Omega_n^{(2)}(M) \bigcup \Omega_n^{(3)}(M)}\frac{|R_3(k)|}{(k_0-k)(n-k_0-1)(n-1)\|\delta\|^2/(n-k)} > \eta\right)\\
    =&P\left(\max_{k \in \Omega_n^{(2)}(M) \bigcup \Omega_n^{(3)}(M)}\frac{2(n-1)|(k_0^2 - nk_0 - k + n)|}{(n-k_0-1)(n-1)k_0\|\delta\|^2}\frac{1}{k_0 - k}\left|\sum_{i = k+1}^{k_0}\delta^TZ_i\right| > \eta\right)\\
    =&P\left(\max_{k \in \Omega_n^{(2)}(M) \bigcup \Omega_n^{(3)}(M)}\frac{1}{k_0 - k}\left|\sum_{i = k+1}^{k_0}\delta^TZ_i\right| > C\eta\|\delta\|^2\right) \leq \frac{\delta^T\Sigma\delta}{C^2\eta^2\|\delta\|^4nM/a_n}\\
    &=o\left(\frac{n\|\delta\|^4}{\|\delta\|^4nM}\right) = o(1).
    \end{align*}
    Thus $R_3(k)$ is also uniformly dominated.
    
As for Equation (\ref{eq:Gz}), it is equivalent to show
	$$P\left(\sup_{k \in \Omega_n^{(2)}(M) \bigcup \Omega_n^{(3)}(M)}  \frac{1}{(k_0-k)}|G_n^Z(k) - G_n^Z(k_0)| \geq n\eta\|\delta\|^2\right) < \epsilon,$$
	for any positive constant $\eta$.

By Proposition \ref{prop:terms} we have
\begin{align*}
&G_n^Z(k_0) - G_n^Z(k) \\
=& 2\frac{(n-k-1)(n-1)}{(n-k)k}\sum_{i = k+1}^{k_0}\sum_{j = 1}^{k}Z_i^TZ_j - 2\frac{(k_0-1)(n-1)}{(n-k_0)k_0}\sum_{i = k+1}^{k_0}\sum_{j = k_0+1}^{n}Z_i^TZ_j\\
+ &\frac{(n-1)(n-k-k_0)}{k(n-k_0)}\sum_{i,j = k+1, i\neq j}^{k_0}Z_i^TZ_j- \frac{(n-1)(k_0-k)}{kk_0}\sum_{i,j=1,i\neq j}^{k_0}Z_i^TZ_j\\
+&\frac{(n-1)(k_0-k)}{(n-k_0)(n-k)}\sum_{i,j = k+1,i\neq j}^{n}Z_i^TZ_j - \frac{2(n-1)(n-k_0-k)(k_0-k)}{kk_0(n-k)(n-k_0)}\sum_{i = 1}^{k}\sum_{j = k_0+1}^{n}Z_i^TZ_j\\
=&S_1(k)+S_2(k)+S_3(k)+S_4(k)+S_5(k)+S_6(k).
\end{align*}

The remaining steps are to show $S_i(k)/(k_0-k)$ is dominated by $n\|\delta\|^2$ on the set $\Omega_n^{(2)}(M) \bigcup \Omega_n^{(3)}(M)$, for all $i = 1,\cdots,6$. Note that in this set, $k/n > \tau_0/2$. To show the result, by Lemma \ref{lem:rateterms}\ref{lem:sum}, $\sup_{k < k_0}|S_4(k)|/(k_0-k)$ and $\sup_{k < k_0}|S_5(k)|/(k_0-k)$ are all $O_p(\|\Sigma\|_F)$, hence $o_p(n\|\delta\|^2)$ since $\sqrt{a_n} \rightarrow \infty$ by Assumption \ref{ass}\ref{ass:rate}. So $S_4(k)$ and $S_5(k)$ are both asymptotically negligible.

For $S_2(k)$,
$$\frac{1}{k_0-k}|S_2(k)| = \frac{2(k_0-1)(n-1)}{(n-k_0)k_0}\left|\frac{1}{k_0-k}\sum_{i = k+1}^{k_0}\sum_{j = k_0+1}^{n}Z_i^TZ_j\right|. $$
By Lemma \ref{lem:rateterms}\ref{lem:HRcross}, for any positive constant $\eta$,
$$P\left(\max_{k_0/2 \leq k \leq k_0 - nM/a_n}\frac{1}{k_0-k}\left|\sum_{i = k+1}^{k_0}\sum_{k_0+1}^{n}Z_i^TZ_j\right| > \eta n\|\delta\|^2\right) \leq \frac{C_1(n-k_0)\|\Sigma\|_F^2}{\eta^2n^2\|\delta\|^4}\left(\frac{a_n}{nM}\right) \leq \frac{C_1'}{M} \leq \epsilon/6$$
for sufficiently large $M$. Hence $S_2(k)$ is dominated. 

For $S_3(k)$, by Lemma \ref{lem:rateterms}\ref{lem:HRmid},
$$P\left(\max_{k_0/2+1 \leq k \leq k_0 - nM/a_n}\frac{1}{k_0-k}\left|\sum_{i,j = k+1,i\neq j}^{k_0}Z_i^TZ_j\right| > \eta n \|\delta\|^2\right) \leq \frac{C_3\log(k_0-Mn/a_n)\|\Sigma\|_F^2}{\eta^2n^2\|\delta\|^4} \leq \epsilon/6$$
for all large $n$ under Assumption \ref{ass}\ref{ass:rate}.

For $S_6(k)$, by Lemma \ref{lem:rateterms}\ref{lem:HRcrosssum},
$$P\left(\max_{k_0/2+1 \leq k \leq k_0 - nM/a_n}\frac{1}{k}\left|\sum_{i = 1}^{k}\sum_{j = k_0+1}^{n}Z_i^TZ_j\right| > \eta n\|\delta\|^2\right) \leq \frac{C_2(n-k_0)\|\Sigma\|_F^2}{\eta^2n^2k_0\|\delta\|^4} \leq \frac{C_2'}{a_n} \leq \epsilon/6$$
for all large enough $n$.

It remains to deal with $S_1(k)$. Note that 
\begin{align*}
S_1(k) &= \frac{2(n-k-1)(n-1)}{(n-k)k}\left(\sum_{i = k+1}^{k_0}Z_i^T\sum_{j = 1}^{k}Z_j\right).
\end{align*}
For any $\eta > 0$, we want to show
\begin{align*}
P\left(\sup_{k \in \Omega_n^{(2)}(M) \bigcup \Omega_n^{(3)}(M)}\left|\frac{2(n-k-1)(n-1)}{(k_0-k)(n-k)k}\left(\sum_{i = k+1}^{k_0}Z_i^T\sum_{j = 1}^{k}Z_j\right)\right| > \eta n\|\delta\|^2\right) < \epsilon/6.
\end{align*}

This is equivalent to prove that for any positive $\eta$, 
$$P\left(\sup_{k \in \Omega_n^{(2)}(M)\bigcup\Omega_n^{(3)}(M)}\frac{1}{(k_0-k)\|\Sigma\|_F}\left|\left(\sum_{i = k+1}^{k_0}Z_i^T\sum_{j = 1}^{k}Z_j\right)\right| > \eta \sqrt{a_n}\right) < \epsilon,$$
which was proved in Proposition \ref{lem:S1}. The proof is thus complete.

\subsection{Proof of Theorem \ref{thm:weak}}
In view of Proposition \ref{prop:tightness} which shows the tightness, we shall only present the proof for the finite-dimensional convergence. For any $k < k_0$, it follows from Proposition \ref{prop:terms} that
\begin{align*}
&G_n^Z(k_0) - G_n^Z(k)\\
= & 2\frac{(n-k-1)(n-1)}{(n-k)k}\sum_{i = k+1}^{k_0}\sum_{j = 1}^{k}Z_i^TZ_j - 2\frac{(k_0-1)(n-1)}{(n-k_0)k_0}\sum_{i = k+1}^{k_0}\sum_{j = k_0+1}^{n}Z_i^TZ_j\\
+ &\frac{(n-1)(n-k-k_0)}{k(n-k_0)}\sum_{i,j = k+1, i\neq j}^{k_0}Z_i^TZ_j- \frac{(n-1)(k_0-k)}{kk_0}\sum_{i,j=1,i\neq j}^{k_0}Z_i^TZ_j\\
+&\frac{(n-1)(k_0-k)}{(n-k_0)(n-k)}\sum_{i,j = k+1,i\neq j}^{n}Z_i^TZ_j - \frac{2(n-1)(n-k_0-k)(k_0-k)}{kk_0(n-k)(n-k_0)}\sum_{i = 1}^{k}\sum_{j = k_0+1}^{n}Z_i^TZ_j\\
=&S_1(k)+S_2(k)+S_3(k)+S_4(k)+S_5(k)+S_6(k).
\end{align*}

Here we only need to consider $k = \floor{k_0-n\gamma/b_n}$ for any $\gamma \in (0,M]$. For simplicity we assume $n\gamma/b_n$ is an integer. Since $b_n \rightarrow \infty$ and $k - k_0 =  o(n)$, we have $k/n \rightarrow \tau_0$.

For $S_3(k)$,
\begin{align*}
&Var\left(\frac{\sqrt{2}\sqrt{b_n}}{n\|\Sigma\|_F}S_3(k)\right)\\
=&\frac{8b_n}{n^2\|\Sigma\|_F^2}\frac{(n-1)^2(n-k-k_0)^2(k_0-k+1)(k_0-k)}{k^2(n-k_0)^2}\|\Sigma\|_F^2\\
\leq &\frac{C\gamma^2}{b_n} + \frac{C\gamma}{n} \rightarrow 0
\end{align*}
for some positive constant $C$. Hence $\frac{\sqrt{2}\sqrt{b_n}}{n\|\Sigma\|_F}S_3(k) = o_p(1)$, for any fixed $\gamma \in (0,M]$.

For $S_4(k)$,
\begin{align*}
&\frac{\sqrt{2}\sqrt{b_n}}{n\|\Sigma\|_F}S_4(k) = -\frac{\sqrt{2}\sqrt{b_n}}{n\|\Sigma\|_F}\frac{(n-1)(k_0-k)}{kk_0}\sum_{i,j = 1, i\neq j}^{k_0}Z_i^TZ_j\\
=&-\frac{2\sqrt{b_n}n\gamma}{b_n}\frac{(n-1)}{kk_0}\left(\frac{\sqrt{2}}{n\|\Sigma\|_F}\sum_{i = 1}^{k_0-1}\sum_{j = 1}^{i}Z_{i+1}^TZ_j\right) = \frac{\gamma}{\sqrt{b_n}}O_p(1) = o_p(1)
\end{align*}
since $\frac{\sqrt{2}}{n\|\Sigma\|_F}\sum_{i = 1}^{k_0-1}\sum_{j = 1}^{i}Z_{i+1}^TZ_j \overset{\mathcal{D}}{\rightarrow} Q(0,\tau_0)$, which is a simple consequence of Proposition \ref{prop:Q}.

By a similar argument, for $S_5(k)$,
\begin{align*}
&\frac{\sqrt{2}\sqrt{b_n}}{n\|\Sigma\|_F}S_5(k) = \frac{\sqrt{2}\sqrt{b_n}}{n\|\Sigma\|_F}\frac{(n-1)(k_0-k)}{(n-k_0)(n-k)}\sum_{i,j = k+1, i\neq j}^{n}Z_i^TZ_j\\
=&\frac{2\sqrt{b_n}n\gamma}{b_n}\frac{(n-1)}{(n-k_0)(n-k)}\left(\frac{\sqrt{2}}{n\|\Sigma\|_F}\sum_{i = k+1}^{n-1}\sum_{j = k+1}^{i}Z_{i+1}^TZ_j\right) = \frac{\gamma}{\sqrt{b_n}}O_p(1) = o_p(1)
\end{align*}
since 
\begin{align*}
&Var\left(\frac{\sqrt{2}}{n\|\Sigma\|_F}\sum_{i = k+1}^{n-1}\sum_{j = k+1}^{i}Z_{i+1}^TZ_j - \frac{\sqrt{2}}{n\|\Sigma\|_F}\sum_{i = k_0+1}^{n-1}\sum_{j = k_0+1}^{i}Z_{i+1}^TZ_j\right)\\
=&Var\left(\frac{\sqrt{2}}{n\|\Sigma\|_F}\sum_{j = k+1}^{n-1}\sum_{i = j}^{n-1}Z_{i+1}^TZ_j - \frac{\sqrt{2}}{n\|\Sigma\|_F}\sum_{j = k_0+1}^{n-1}\sum_{i = j}^{n-1}Z_{i+1}^TZ_j\right)\\
=&\frac{(2n-k-k_0-1)(k_0-k)\|\Sigma\|_F^2}{n^2\|\Sigma\|_F^2} = \frac{\gamma(2n-k-k_0-1)}{b_nn} \rightarrow 0
\end{align*}
and $\frac{\sqrt{2}}{n\|\Sigma\|_F}\sum_{i = k_0+1}^{n}\sum_{j = k_0+1}^{i}Z_{i+1}^TZ_j \overset{\mathcal{D}}{\rightarrow} Q(\tau_0,1)$.

For $S_6(k)$, we note that 
\begin{align*}
&Var\left(\frac{\sqrt{2}\sqrt{b_n}}{n\|\Sigma\|_F}S_6(k)\right)\\
=&\frac{8b_n}{n^2\|\Sigma\|_F^2}\frac{(n-1)^2(k_0-k)^2(n-k_0-k)^2}{k^2k_0^2(n-k)^2(n-k_0)^2}k(n-k_0)\|\Sigma\|_F^2 \leq \frac{C\gamma^2}{b_nn^2} \rightarrow 0.
\end{align*}
Hence $\frac{\sqrt{2}\sqrt{b_n}}{n\|\Sigma\|_F}S_6(k) = o_p(1)$.

The only remaining two terms are $S_1(k)$ and $S_2(k)$. These two terms are not asymptotically negligible and we can employ martingale CLT to get the asymptotic distribution. 

By Corollary 3.1 in \cite{hall1980martingale}, for any square-integrable martingale difference triangular array $Y_{n,i}$ for $i = 1,2,...,k_n$ with $k_n \rightarrow \infty$ and $\mathcal{F}_{n,i}$ is the natural filtration for $Y_{n,i},Y_{n,i-1},\cdots$, if 
\begin{enumerate}
	\item $\sum_{i = 1}^{k_n}\E[Y_{n,i}^4] \rightarrow 0$ (Lyapunov's Condition), and
	\item $V_{nk_n} = \sum_{i = 1}^{k_n}\E[Y_{n,i}^2|\mathcal{F}_{n,i-1}] \rightarrow_p \sigma^2,$
\end{enumerate}
then we have $S_n = \sum_{i = 1}^{k_n}Y_{n,i} \overset{\mathcal{D}}{\rightarrow} N(0,\sigma^2)$.

Consider any $k_1 < k_2 < k$. Without loss of generality, assume there exists $\gamma_1,\gamma_2$ such that $k_0-k_1 = n\gamma_1/b_n$ and $k_0 - k_2 = n\gamma_2/b_n$, which means $\gamma_2 < \gamma_1$. For any $\alpha_1,\alpha_2,\alpha_3,\alpha_4 \in \mathbb{R}$, consider
\begin{align*}
&\frac{\sqrt{2}\sqrt{b_n}}{n\|\Sigma\|_F}\left(\alpha_1\sum_{i = k_1+1}^{k_0}\sum_{j = 1}^{k_1}Z_i^TZ_j + \alpha_2\sum_{i = k_2+1}^{k_0}\sum_{j = 1}^{k_2}Z_i^TZ_j + \alpha_3\sum_{i = k_1+1}^{k_0}\sum_{j = k_0+1}^{n}Z_i^TZ_j + \alpha_4\sum_{i = k_2+1}^{k_0}\sum_{j = k_0+1}^{n}Z_i^TZ_j\right)\\
=&\sum_{i = k_1+1}^{n}Y_{n,i}
\end{align*}
where we define $Y_{n,i}$ as
\begin{enumerate}
	\item $Y_{n,i} = \alpha_1\frac{\sqrt{2}\sqrt{b_n}}{n\|\Sigma\|_F}Z_{i}^T\sum_{j = 1}^{k_1}Z_j$, for $i = k_1+1,...,k_2$;
	
	\item $Y_{n,i} = \alpha_1\frac{\sqrt{2}\sqrt{b_n}}{n\|\Sigma\|_F}Z_{i}^T\sum_{j = 1}^{k_1}Z_j + \alpha_2\frac{\sqrt{2}\sqrt{b_n}}{n\|\Sigma\|_F}Z_{i}^T\sum_{j = 1}^{k_2}Z_j$, for $i = k_2+1,...,k_0$;
	
	\item $Y_{n,i} = \alpha_3\frac{\sqrt{2}\sqrt{b_n}}{n\|\Sigma\|_F}Z_{i}^T\sum_{i = k_1+1}^{k_0}Z_j + \alpha_4\frac{\sqrt{2}\sqrt{b_n}}{n\|\Sigma\|_F}Z_{i}^T\sum_{i = k_2+1}^{k_0}Z_j$, for $i = k_0+1,...,n$
\end{enumerate}
and define $\cF_{n,i}$ as the natural filtration of $Z_{i},Z_{i-1},...$. It is easy to verify that $Y_{n,i}$ is a martingale difference sequence adaptive to $\cF_{n,i}$.

Then for $i = k_1+1,...k_2$,
\begin{align*}
&\E[Y_{n,i}^4]\\
=&\frac{4\alpha_1^4b_n^2}{n^4\|\Sigma\|_F^4}\sum_{j_1,j_2,j_3,j_4 = 1}^{k_1}\sum_{l_1,\cdots,l_4 = 1}^{p}\E[Z_{i,l_1}Z_{i,l_2}Z_{i,l_3}Z_{i,l_4}]\E[Z_{j_1,l_1}Z_{j_2,l_2}Z_{j_3,l_3}Z_{j_4,l_4}]\\
=&\frac{4\alpha_1^4b_n^2}{n^4\|\Sigma\|_F^4}\sum_{j_1 = 1}^{k_1}\sum_{l_1,\cdots,l_4 = 1}^{p}\E[Z_{i,l_1}Z_{i,l_2}Z_{i,l_3}Z_{i,l_4}]\E[Z_{j_1,l_1}Z_{j_1,l_2}Z_{j_1,l_3}Z_{j_1,l_4}]\\
+&\frac{12\alpha_1^4b_n^2}{n^4\|\Sigma\|_F^4}\sum_{j_1,j_2 = 1,j_1 \neq j_2}^{k_1}\sum_{l_1,\cdots,l_4 = 1}^{p}\E[Z_{i,l_1}Z_{i,l_2}Z_{i,l_3}Z_{i,l_4}]\E[Z_{j_1,l_1}Z_{j_1,l_2}]\E[Z_{j_2,l_3}Z_{j_2,l_4}]\\
\leq&\frac{4\alpha_1^4b_n^2k_1}{n^4\|\Sigma\|_F^4}\sqrt{\sum_{l_1,\cdots,l_4 = 1}^{p}\E[Z_{i,l_1}Z_{i,l_2}Z_{i,l_3}Z_{i,l_4}]^2}\sqrt{\sum_{l_1,\cdots,l_4 = 1}^{p}\E[Z_{j_1,l_1}Z_{j_1,l_2}Z_{j_1,l_3}Z_{j_1,l_4}]^2}\\
+&\frac{12\alpha_1^4b_n^2k_1^2}{n^4\|\Sigma\|_F^4}\sqrt{\sum_{l_1,\cdots,l_4 = 1}^{p}\E[Z_{i,l_1}Z_{i,l_2}Z_{i,l_3}Z_{i,l_4}]^2}\sqrt{\sum_{l_1,\cdots,l_4 = 1}^{p}\Sigma_{l_1,l_2}^2\Sigma_{l_3,l_4}^2} \leq C\frac{b_n^2}{n^2},
\end{align*}
where we have applied the Cauchy-Schwartz inequality and Lemma \ref{lem:Xi1Xi4}. Furthermore we notice that
\begin{align*}
&\E[Y_{n,i}^2|\cF_{n,i-1}]=\frac{2b_n\alpha_1^2}{n^2\|\Sigma\|_F^2}\sum_{j_1,j_2 = 1}^{k_1}\sum_{l_1,l_2 = 1}^pZ_{j_1,l_1}Z_{j_2,l_2}\Sigma_{l_1,l_2}.
\end{align*}

For $i = k_2+1,\cdots,k_0$, by essentially the same arguments,
\begin{align*}
\E[Y_{n,i}^4]
=&\E\left[\left(\alpha_1\frac{\sqrt{2}\sqrt{b_n}}{n\|\Sigma\|_F}Z_{i}^T\sum_{j = 1}^{k_1}Z_j + \alpha_2\frac{\sqrt{2}\sqrt{b_n}}{n\|\Sigma\|_F}Z_{i}^T\sum_{j = 1}^{k_2}Z_j\right)^4\right]\\
\leq&8\E\left[\left(\alpha_1\frac{\sqrt{2}\sqrt{b_n}}{n\|\Sigma\|_F}Z_{i}^T\sum_{j = 1}^{k_1}Z_j\right)^4\right] + 8\E\left[\left( \alpha_2\frac{\sqrt{2}\sqrt{b_n}}{n\|\Sigma\|_F}Z_{i}^T\sum_{j = 1}^{k_2}Z_j\right)^4\right] = C\frac{b_n^2}{n^2}
\end{align*}
and
\begin{align*}
&\E[Y_{n,i}^2|\cF_{n,i-1}]\\
=&\frac{2b_n\alpha_1^2}{n^2\|\Sigma\|_F^2}\sum_{j_1,j_2 = 1}^{k_1}\sum_{l_1,l_2 = 1}^pZ_{j_1,l_1}Z_{j_2,l_2}\Sigma_{l_1,l_2}+\frac{2b_n\alpha_2^2}{n^2\|\Sigma\|_F^2}\sum_{j_1,j_2 = 1}^{k_2}\sum_{l_1,l_2 = 1}^pZ_{j_1,l_1}Z_{j_2,l_2}\Sigma_{l_1,l_2}\\
+&\frac{4b_n\alpha_1\alpha_2}{n^2\|\Sigma\|_F^2}\sum_{j_1 = 1}^{k_1}\sum_{j_2 = 1}^{k_2}\sum_{l_1,l_2 = 1}^{p}Z_{j_1,l_1}Z_{j_2,l_2}\Sigma_{l_1,l_2}.
\end{align*}

For $i = k_0+1,\cdots,n$, we have
\begin{align*}
\E[Y_{n,i}^4] = &\E\left[\left(\alpha_3\frac{\sqrt{2}\sqrt{b_n}}{n\|\Sigma\|_F}Z_{i}^T\sum_{i = k_1+1}^{k_0}Z_j + \alpha_4\frac{\sqrt{2}\sqrt{b_n}}{n\|\Sigma\|_F}Z_{i}^T\sum_{i = k_2+1}^{k_0}Z_j\right)^4\right]\\
\leq&8\E\left[\left(\alpha_3\frac{\sqrt{2}\sqrt{b_n}}{n\|\Sigma\|_F}Z_{i}^T\sum_{j = k_1+1}^{k_0}Z_j\right)^4\right] + 8\E\left[\left( \alpha_4\frac{\sqrt{2}\sqrt{b_n}}{n\|\Sigma\|_F}Z_{i}^T\sum_{j = k_2+1}^{k_0}Z_j\right)^4\right] \leq Cn^{-2},
\end{align*}
and
\begin{align*}
&\E[Y_{n,i}^2|\cF_{n,i-1}]\\
=&\frac{2b_n\alpha_3^2}{n^2\|\Sigma\|_F^2}\sum_{j_1,j_2 = k_1+1}^{k_0}\sum_{l_1,l_2 = 1}^pZ_{j_1,l_1}Z_{j_2,l_2}\Sigma_{l_1,l_2}+\frac{2b_n\alpha_4^2}{n^2\|\Sigma\|_F^2}\sum_{j_1,j_2 = k_2+1}^{k_0}\sum_{l_1,l_2 = 1}^pZ_{j_1,l_1}Z_{j_2,l_2}\Sigma_{l_1,l_2}\\
+&\frac{4b_n\alpha_3\alpha_4}{n^2\|\Sigma\|_F^2}\sum_{j_1 = k_1+1}^{k_0}\sum_{j_2 = k_2+1}^{k_0}\sum_{l_1,l_2 = 1}^{p}Z_{j_1,l_1}Z_{j_2,l_2}\Sigma_{l_1,l_2}.
\end{align*}

To verify the two conditions for the martingale CLT, we note that for the first condition,
\begin{align*}
\sum_{i = k_1+1}^{n}\E[Y_{n,i}^4]
\leq&(k_2-k_1)C\frac{b_n^2}{n^2} + (k_0-k_2)C\frac{b_n^2}{n^2} + (n-k_0)Cn^{-2}\\
=&(k_0-k_2)C\frac{b_n^2}{n^2} + (n-k_0)Cn^{-2} = O(b_n/n) + O(n^{-1}) \rightarrow 0.
\end{align*}
For the second condition, we write
\begin{align*}
&\sum_{i = k_1 +1}^{n}\E[Y_{n,i}^2|\cF_{n,i-1}]\\
=&\sum_{i = k_1 +1}^{k_2}\E[Y_{n,i}^2|\cF_{n,i-1}] + \sum_{i = k_2 +1}^{k_0}\E[Y_{n,i}^2|\cF_{n,i-1}] + \sum_{i = k_0 +1}^{n}\E[Y_{n,i}^2|\cF_{n,i-1}]\\
=&\frac{2b_n\alpha_1^2}{n^2\|\Sigma\|_F^2}\sum_{i = k_1+1}^{k_0}\sum_{j_1,j_2 = 1}^{k_1}\sum_{l_1,l_2 = 1}^pZ_{j_1,l_1}Z_{j_2,l_2}\Sigma_{l_1,l_2} + \frac{2b_n\alpha_2^2}{n^2\|\Sigma\|_F^2}\sum_{i = k_2+1}^{k_0}\sum_{j_1,j_2 = 1}^{k_2}\sum_{l_1,l_2 = 1}^pZ_{j_1,l_1}Z_{j_2,l_2}\Sigma_{l_1,l_2}\\
+&\frac{4b_n\alpha_1\alpha_2}{n^2\|\Sigma\|_F^2}\sum_{i = k_2+1}^{k_0}\sum_{j_1 = 1}^{k_1}\sum_{j_2 = 1}^{k_2}\sum_{l_1,l_2 = 1}^{p}Z_{j_1,l_1}Z_{j_2,l_2}\Sigma_{l_1,l_2}+\frac{2b_n\alpha_3^2}{n^2\|\Sigma\|_F^2}\sum_{i = k_0+1}^{n}\sum_{j_1,j_2 = k_1+1}^{k_0}\sum_{l_1,l_2 = 1}^pZ_{j_1,l_1}Z_{j_2,l_2}\Sigma_{l_1,l_2}\\
+&\frac{2b_n\alpha_4^2}{n^2\|\Sigma\|_F^2}\sum_{i = k_0+1}^{n}\sum_{j_1,j_2 = k_2+1}^{k_0}\sum_{l_1,l_2 = 1}^pZ_{j_1,l_1}Z_{j_2,l_2}\Sigma_{l_1,l_2}+\frac{4b_n\alpha_3\alpha_4}{n^2\|\Sigma\|_F^2}\sum_{i = k_0+1}^{n}\sum_{j_1 = k_1+1}^{k_0}\sum_{j_2 = k_2+1}^{k_0}\sum_{l_1,l_2 = 1}^{p}Z_{j_1,l_1}Z_{j_2,l_2}\Sigma_{l_1,l_2}\\
:=& U_{1,n}+U_{2,n}+U_{3,n}+U_{4,n}+U_{5,n}+U_{6,n}.
\end{align*}

Define $\sigma^2 = \sum_{i = 1}^{6}\sigma_i^2$ where $\sigma_i^2 = \lim\E[U_{i,n}]$, and we are going to show $U_{i,n} \rightarrow_p \sigma_i^2$. For $U_{1,n}$, $\sigma_1^2 = 2\alpha_1^2\gamma_1\tau_0$, and
\begin{align*}
&\E[(U_{1,n}-\sigma_1^2)^2]=\E[U_{1,n}^2] - 2\sigma_1^2\E[U_{1,n}] + \sigma_1^4\\
=&\frac{4b_n^2\alpha_1^4}{n^4\|\Sigma\|_F^4}\sum_{i_1,i_2 = k_1+1}^{k_0}\sum_{j_1,j_2,j_3,j_4 = 1}^{k_1}\sum_{l_1,l_2,l_3,l_4}^{p}\E[Z_{j_1,l_1}Z_{j_2,l_2}Z_{j_3,l_3}Z_{j_4,l_4}]\Sigma_{l_1,l_2}\Sigma_{l_3,l_4} - 2\sigma_1^2\E[U_{1,n}] + \sigma_1^4\\
=&\frac{4b_n^2\alpha_1^4(k_0-k_1)^2}{n^4\|\Sigma\|_F^4}\sum_{j_1 = 1}^{k_1}\sum_{l_1,l_2,l_3,l_4}^{p}\E[Z_{j_1,l_1}Z_{j_1,l_2}Z_{j_1,l_3}Z_{j_1,l_4}]\Sigma_{l_1,l_2}\Sigma_{l_3,l_4} - 2\sigma_1^2\E[U_{1,n}] + \sigma_1^4\\
+&\frac{4b_n^2\alpha_1^4(k_0-k_1)^2}{n^4\|\Sigma\|_F^4}\sum_{j_1,j_2 = 1,j_1 \neq j_2}^{k_1}\sum_{l_1,l_2,l_3,l_4}^{p}\E[Z_{j_1,l_1}Z_{j_1,l_2}]\E[Z_{j_2,l_3}Z_{j_2,l_4}]\Sigma_{l_1,l_2}\Sigma_{l_3,l_4}\\
+&\frac{4b_n^2\alpha_1^4(k_0-k_1)^2}{n^4\|\Sigma\|_F^4}\sum_{j_1,j_2 = 1,j_1 \neq j_2}^{k_1}\sum_{l_1,l_2,l_3,l_4}^{p}\E[Z_{j_1,l_1}Z_{j_1,l_3}]\E[Z_{j_2,l_2}Z_{j_2,l_4}]\Sigma_{l_1,l_2}\Sigma_{l_3,l_4}\\
+&\frac{4b_n^2\alpha_1^4(k_0-k_1)^2}{n^4\|\Sigma\|_F^4}\sum_{j_1,j_2 = 1,j_1 \neq j_2}^{k_1}\sum_{l_1,l_2,l_3,l_4}^{p}\E[Z_{j_1,l_1}Z_{j_1,l_4}]\E[Z_{j_2,l_2}Z_{j_2,l_3}]\Sigma_{l_1,l_2}\Sigma_{l_3,l_4}\\
\leq&\frac{4Cb_n^2\alpha_1^4(k_0-k_1)^2k_1}{n^4}- 2\sigma_1^2\E[U_{1,n}] + \sigma_1^4 + \frac{4b_n^2\alpha_1^4(k_0-k_1)^2k_1^2}{n^4}(1 + 2tr(\Sigma^4)/\|\Sigma\|_F^4)\\
=&\left(\frac{2\alpha_1^2\gamma_1k_1}{n}\right)^2 - 2\sigma_1^2\E[U_{1,n}] + \sigma_1^4 + O(1/n) + o(1) \rightarrow 0,
\end{align*}
since 
\begin{align*}
\sum_{l_1,...l_4=1}^{p}\Sigma_{l_1,l_3}\Sigma_{l_2,l_4}\Sigma_{l_1,l_2}\Sigma_{l_3,l_4} = \sum_{l_2,l_3 = 1}^{p}(\Sigma^2)_{l_2,l_3}(\Sigma^2)_{l_2,l_3} = \|\Sigma^2\|_F^4 = tr(\Sigma^4) = o(\|\Sigma\|_F^4),
\end{align*}
under Assumption \ref{ass}\ref{ass:tr} and $k_1/n = k_0/n - \gamma_1/b_n \rightarrow \tau_0$. Thus $U_1 \rightarrow_p \sigma_1^2$. By exactly the same derivation, we have $U_{2,n} \rightarrow_p \sigma_2^2 = 2\alpha_2^2\gamma_2\tau_0$. For $U_{3,n}$, $\sigma_3^2 = 4\alpha_1\alpha_2\tau_0\gamma_2$, and
\begin{align*}
&\E[(U_{3,n} - \sigma_3^2)^2]\\
=&\frac{16\alpha_1^2\alpha_2^2b_n^2(k_0-k_2)^2}{n^4\|\Sigma\|_F^4}\sum_{j_1,j_3 = 1}^{k_1}\sum_{j_2,j_4 = 1}^{k_2}\sum_{l_1,l_2,l_3,l_4=1}^{p}\E[Z_{j_1,l_1}Z_{j_2,l_2}Z_{j_3,l_3}Z_{j_4,l_4}]\Sigma_{l_1,l_2}\Sigma_{l_3,l_4} - 2\sigma_3^2\E[U_{3,n}] + \sigma_3^4\\
=&\frac{16\alpha_1^2\alpha_2^2b_n^2(k_0-k_2)^2}{n^4\|\Sigma\|_F^4}\sum_{j_1=1}^{k_1}\sum_{l_1,l_2,l_3,l_4=1}^{p}\E[Z_{j_1,l_1}Z_{j_1,l_2}Z_{j_1,l_3}Z_{j_1,l_4}]\Sigma_{l_1,l_2}\Sigma_{l_3,l_4}-2\sigma_3^2\E[U_{3,n}] + \sigma_3^4\\
+&\frac{16\alpha_1^2\alpha_2^2b_n^2(k_0-k_2)^2}{n^4\|\Sigma\|_F^4}\sum_{j_1,j_2=1}^{k_1}\sum_{l_1,l_2,l_3,l_4=1}^{p}\E[Z_{j_1,l_1}Z_{j_1,l_2}]\E[Z_{j_2,l_3}Z_{j_2,l_4}]\Sigma_{l_1,l_2}\Sigma_{l_3,l_4}\\
+&\frac{16\alpha_1^2\alpha_2^2b_n^2(k_0-k_2)^2}{n^4\|\Sigma\|_F^4}\sum_{j_1=1}^{k_1}\sum_{j_2 = 1}^{k_2}\sum_{l_1,l_2,l_3,l_4=1}^{p}\E[Z_{j_1,l_1}Z_{j_1,l_3}]\E[Z_{j_2,l_2}Z_{j_2,l_4}]\Sigma_{l_1,l_2}\Sigma_{l_3,l_4}\\
+&\frac{16\alpha_1^2\alpha_2^2b_n^2(k_0-k_2)^2}{n^4\|\Sigma\|_F^4}\sum_{j_1,j_2=1}^{k_1}\sum_{l_1,l_2,l_3,l_4=1}^{p}\E[Z_{j_1,l_1}Z_{j_1,l_4}]\E[Z_{j_2,l_2}Z_{j_2,l_3}]\Sigma_{l_1,l_2}\Sigma_{l_3,l_4}\\
=&O(1/n) + \left(\frac{4\alpha_1\alpha_2\gamma_2k_1}{n}\right)^2-2\sigma_3^2\E[U_{3,n}]+\sigma_3^4+\left(\frac{8\alpha_1\alpha_2\gamma_2k_1}{n}\right)^2\frac{tr(\Sigma^4)}{\|\Sigma\|_F^4}\rightarrow 0,
\end{align*}
under Assumption \ref{ass}\ref{ass:tr}. Thus $U_3 \rightarrow_p \sigma_3^2$. By a simple calculation we have $\sigma_4^2 = 2\alpha_3^2\gamma_1(1-\tau_0)$, $\sigma_5^2 = 2\alpha_4^2\gamma_2(1-\tau_0)$ and $\sigma_6^2 = 4\alpha_3\alpha_4\gamma_2(1-\tau_0)$. The proof for the consistency of $U_{4,n},U_{5,n}$ and $U_{6,n}$ are skipped since the arguments are exactly the same.

Therefore, we prove that $\sum_{i = k_1 +1}^{n}\E[Y_{n,i}^2|\cF_{n,i-1}] \rightarrow_p \sigma^2$. So
$\sum_{i = k_1+1}^{n}Y_{n,i} \overset{\mathcal{D}}{\rightarrow} N(0,\sigma^2)$, where $$\sigma^2 = \sum_{i = 1}^{6}\sigma_i^2 = 2\alpha_1^2\gamma_1\tau_0 + 2\alpha_2^2\gamma_2\tau_0+ 4\alpha_1\alpha_2\tau_0\gamma_2+ 2\alpha_3^2\gamma_1(1-\tau_0)+2\alpha_4^2\gamma_2(1-\tau_0)+4\alpha_3\alpha_4\gamma_2(1-\tau_0).$$

It pays to look into a few special cases. Let $\alpha_1 = 2/\tau_0$, $\alpha_3 = -2/(1-\tau_0)$ and $\alpha_2 = \alpha_4 = 0$, we have
$$\frac{\sqrt{2}\sqrt{b_n}}{n\|\Sigma\|_F}(S_1(k)+S_2(k)) \overset{\mathcal{D}}{\rightarrow} N(0,8\gamma_1/(\tau_0(1-\tau_0))) \overset{\mathcal{D}}{=}\frac{2\sqrt{2}}{\sqrt{\tau_0(1-\tau_0)}}W(\gamma_1),$$
which implies $H_n(\gamma_1)\overset{\mathcal{D}}{\rightarrow} \frac{2\sqrt{2}}{\sqrt{\tau_0(1-\tau_0)}}W(\gamma_1),$ where $W(r)$ is a standard Brownian Motion.
Let $\alpha_2 = 2/\tau_0$, $\alpha_4 = 2/(1-\tau_0)$ and $\alpha_1 = \alpha_3 = 0$,
we have
$$H_n(\gamma_2) \overset{\mathcal{D}}{\rightarrow} N(0,8\gamma_2/\tau_0(1-\tau_0))\overset{\mathcal{D}}{=}\frac{2\sqrt{2}}{\sqrt{\tau_0(1-\tau_0)}}W(\gamma_2).$$
Further, by letting $\alpha_1 = 2\beta_1/\tau_0$, $\alpha_3 = 2\beta_1/(1-\tau_0)$,$\alpha_2 = 2\beta_2/\tau_0$, $\alpha_4 = 2\beta_2/(1-\tau_0)$ for any $\beta_1,\beta_2 \in \mathbb{R}$, we have
\begin{align*}
\beta_1H_n(\gamma_1) + \beta_2H_n(\gamma_2) &
\overset{\mathcal{D}}{\rightarrow}N(0,8\beta_1^2\gamma_1/(\tau_0(1-\tau_0)) + 8\beta_2^2\gamma_2/(\tau_0(1-\tau_0)) + 16\beta_1\beta_2\gamma_2/((1-\tau_0)\tau_0))\\
&\overset{\mathcal{D}}{=} \frac{2\sqrt{2}}{\sqrt{\tau_0(1-\tau_0)}}(\beta_1W(\gamma_1) + \beta_2W(\gamma_2)).
\end{align*}

The case $k_2 > k_1 > k_0$ can be shown by exactly the same argument, so we skip the details here. For the case that $k_1 < k_0 < k_2$, i.e. $\gamma_1 > 0$ and $\gamma_2 < 0$, we can also employ similar arguments. Consider for any $\alpha_1,\alpha_2,\alpha_3,\alpha_4 \in \mathbb{R}$,
\begin{align*}
&\frac{\sqrt{2}\sqrt{b_n}}{n\|\Sigma\|_F}\left(\alpha_1\sum_{i = k_1+1}^{k_0}\sum_{j = 1}^{k_1}Z_i^TZ_j + \alpha_2\sum_{i = k_0+1}^{k_2}\sum_{j = 1}^{k_0}Z_i^TZ_j + \alpha_3\sum_{i = k_1+1}^{k_0}\sum_{j = k_0+1}^{n}Z_i^TZ_j + \alpha_4\sum_{i = k_0+1}^{k_2}\sum_{j = k_2+1}^{n}Z_i^TZ_j\right)\\
=&\sum_{i = k_1+1}^{n}Y_{n,i},
\end{align*}
where we define $Y_{n,i}$ as
\begin{enumerate}
	\item $Y_{n,i} = \alpha_1\frac{\sqrt{2}\sqrt{b_n}}{n\|\Sigma\|_F}Z_{i}^T\sum_{j = 1}^{k_1}Z_j$, for $i = k_1+1,...,k_0$;
	
	\item $Y_{n,i} = \alpha_2\frac{\sqrt{2}\sqrt{b_n}}{n\|\Sigma\|_F}Z_{i}^T\sum_{j = 1}^{k_0}Z_j + \alpha_3\frac{\sqrt{2}\sqrt{b_n}}{n\|\Sigma\|_F}Z_{i}^T\sum_{j = k_1+1}^{k_0}Z_j$, for $i = k_0+1,...,k_2$;
	
	\item $Y_{n,i} = \alpha_3\frac{\sqrt{2}\sqrt{b_n}}{n\|\Sigma\|_F}Z_{i}^T\sum_{i = k_1+1}^{k_0}Z_j + \alpha_4\frac{\sqrt{2}\sqrt{b_n}}{n\|\Sigma\|_F}Z_{i}^T\sum_{i = k_0+1}^{k_2}Z_j$,for $i = k_2+1,...,n$
\end{enumerate}
and define $\cF_{n,i}$ as the natural filtration of $Z_{i},Z_{i-1},...$. It is easy to verify that $Y_{n,i}$ is a martingale difference sequence adaptive to $\cF_{n,i}$.

By similar arguments, we can show $\sum_{i = k_1+1}^n\E[Y_{n,i}^4] \rightarrow 0$ and $\sum_{i = k_1+1}^n\E[Y_{n,i}^2|\cF_{n,i-1}] \rightarrow_p \sigma^2$, where $\sigma^2 = \lim\sum_{i = k_1+1}^n\E[Y_{n,i}^2]$. To see the specific expression of $\sigma^2$, notice that for $i = k_1+1,...,k_0$, $\E[Y_{n,i}^2] = {2\alpha_1^2b_nk_1}/{n^2}$. For $i = k_0+1,...,k_2$, $\E[Y_{n,i}^2] = {2\alpha_2^2b_nk_1}/{n^2} +{2(\alpha_2+\alpha_3)^2b_n(k_0-k_1)}/{n^2}$. And for $i = k_2+1,...,n$, $\E[Y_{n,i}^2] = {2\alpha_3^2b_n(k_0-k_1)}/{n^2} +{2\alpha_4^2b_n(k_2-k_0)}/{n^2}$. Thus,
\begin{align*}
    \sigma^2 =& \lim\sum_{i = k_1+1}^n\E[Y_{n,i}^2]\\
    =& \lim[(k_0-k_1){2\alpha_1^2b_nk_1}/{n^2}] + \lim[(k_2-k_0)\left\{{2\alpha_2^2b_nk_1}/{n^2} +{2(\alpha_2+\alpha_3)^2b_n(k_0-k_1)}/{n^2}\right\}]\\
    &+ \lim[(n-k_2)\left\{{2\alpha_3^2b_n(k_0-k_1)}/{n^2} +{2\alpha_4^2b_n(k_2-k_0)}/{n^2}\right\}]\\
    =&2\alpha_1^2\tau_0\gamma_1 + 2\alpha_2^2\tau_0(-\gamma_2) +2\alpha_3^2(1-\tau_0)\gamma_1 + 2\alpha_4^2(1-\tau_0)(-\gamma_2).
\end{align*}
 
 Letting $\alpha_1 = 2\beta_1/\tau_0$, $\alpha_3 = 2\beta_1/(1-\tau_0)$,$\alpha_2 = -2\beta_2/\tau_0$, $\alpha_4 = -2\beta_2/(1-\tau_0)$ for any $\beta_1,\beta_2 \in \mathbb{R}$, we have
\begin{align*}
\beta_1H_n(\gamma_1) + \beta_2H_n(\gamma_2) &
\overset{\mathcal{D}}{\rightarrow}N(0,8\beta_1^2\gamma_1/(\tau_0(1-\tau_0)) + 8\beta_2^2(-\gamma_2)/(\tau_0(1-\tau_0)))\\
&\overset{\mathcal{D}}{=} \frac{2\sqrt{2}}{\sqrt{\tau_0(1-\tau_0)}}(\beta_1W_1(\gamma_1) + \beta_2W_2(-\gamma_2)),
\end{align*}
where $W_1,W_2$ are two independent standard brownian motion defined on $[0,\infty)$.

Hence we have shown the finite dimensional convergence. Combining with Proposition \ref{prop:tightness}, we have the process convergence result.

\subsection{Proof of Corollary \ref{cor:asymdist}}
Essentially we want to apply the argmax continuous mapping theorem, i.e. Theorem 3.2.2 in Setion 3.2, \cite{vanweak}. To this end, we first show the weak convergence of the criterion function. Define 
$$L_n(\gamma;\tau_0) = \frac{\sqrt{2}\sqrt{a_n}}{n\|\Sigma\|_F}\left\{G_n(n\tau_0) - G_n(\floor{n\tau_0 + {n\gamma}/{a_n}})\right\}$$
as the criterion function on the real line with the parameter $\tau_0$ fixed. Let $\hat{\gamma}_n = a_n(\hat{\tau}_U - \tau_0) = \argmin_{\gamma \in (-\infty,\infty)}L_n(\gamma;\tau_0)$. We have obtained the consistency  and convergence rate for $\hat{\tau}_U$, i.e. $\hat{\tau}_U \rightarrow_p \tau_0$ in Theorem \ref{thm:rate}. For any fixed $M > 0$ and let $k = \floor{n\tau_0 + {n\gamma}/{a_n}}$,  on the set $[-M,0]$ we have 
\begin{align*}
    L_n(\gamma;\tau_0) &= H_n(\gamma) + \frac{\sqrt{2}\sqrt{a_n}(k_0-k)(n-k_0-1)(n-1)\|\delta\|^2}{n(n-k)\|\Sigma\|_F} + \frac{\sqrt{2}\sqrt{a_n}}{n\|\Sigma\|_F}R(k)
\end{align*}
where $R(k) = R_1(k) + R_2(k) + R_3(k)$ and $R_1,R_2,R_3$ are defined in the proof of Theorem \ref{thm:rate}.

By Theorem \ref{thm:weak} we have $H_n(\gamma) \rightsquigarrow \frac{2\sqrt{2}}{\sqrt{\tau_0(1 - \tau_0)}}W^*(\gamma)$. It is straightforward to see that 
$$ \frac{\sqrt{2}\sqrt{a_n}(k_0-k)(n-k_0-1)(n-1)\|\delta\|^2}{n(n-k)\|\Sigma\|_F}  \rightarrow \sqrt{2}|\gamma|,$$
for $\gamma \in [-M,0]$, i.e. $k \in [k_0 - nM/a_n-1,k_0]$. We shall prove that $\sup_{k \in [k_0-nM/a_n-1,k_0]}R(k) = o_p(1)$. To see this, for any $\eta > 0$,
\begin{align*}
    &P\left(\max_{k \in [k_0 - nM/a_n -1, k_0]}\frac{\sqrt{2}\sqrt{a_n}}{n\|\Sigma\|_F}\frac{2(n-1)(k_0-k)(k_0+k-n)}{k_0(n-k)k}\left|\sum_{i = 1}^k\delta^TZ_i\right| > \eta\right)\\
    \leq&P\left(\max_{k \in [k_0 - nM/a_n -1, k_0]}\frac{\sqrt{2}\sqrt{a_n}}{n\|\Sigma\|_F}\frac{2(n-1)nk_0}{k_0(n-k_0)(k_0-nM/a_n-1)}\left|\sum_{i = 1}^k\delta^TZ_i\right| > \eta\right)\\
    =&P\left(\max_{k \in [k_0 - nM/a_n -1, k_0]}\left|\sum_{i = 1}^k\delta^TZ_i\right| > \frac{k_0(n-k_0)(k_0-nM/a_n-1)n\|\Sigma\|_F}{2\sqrt{2}\sqrt{a_n}(n-1)nk_0}\eta\right)\\
    \leq & \frac{8a_n(n-1)^2n^2k_0^3\delta^T\Sigma\delta}{k_0^2(n-k_0)^2(k_0 - nM/a_n - 1)^2n^2\|\Sigma\|_F^2\eta^2} = O\left(\frac{n^7a_n}{n^8\|\Sigma\|_F^2}\right)\cdot o\left(\frac{\|\Sigma\|_F^2}{n}\right) = o(1) \rightarrow 0.
\end{align*}
Hence $\max_{k \in [k_0 - nM/a_n -1, k_0]}R_1(k) = o_p(1)$. We can use similar arguments to show $R_2(k)$ is uniformly negligible. For $R_3(k)$,
\begin{align*}
    &P\left(\max_{k \in [k_0 - nM/a_n -1, k_0]}\frac{\sqrt{2}\sqrt{a_n}}{n\|\Sigma\|_F}\frac{2(n-1)(k_0^2 - nk_0 - k + n)}{k_0(n-k)}\left|\sum_{i = k+1}^{k_0}\delta^TZ_i\right| > \eta\right)\\
    \leq&P\left(\max_{k \in [k_0 - nM/a_n -1, k_0]}\frac{\sqrt{2}\sqrt{a_n}}{n\|\Sigma\|_F}\frac{2(n-1)(k_0^2 - nk_0 + n)}{k_0(n-k_0)}\left|\sum_{i = k+1}^{k_0}\delta^TZ_i\right| > \eta\right)\\
    \leq& P\left(\max_{k \in [k_0 - nM/a_n -1, k_0]}\left|\sum_{i = k+1}^{k_0}\delta^TZ_i\right| > \frac{k_0(n-k_0)n\|\Sigma\|_F}{2\sqrt{2}\sqrt{a_n}(n-1)(k_0^2 - nk_0 + n)}\eta\right)\\
    \leq&\frac{8a_n(n-1)^2(k_0^2 - nk_0 + n)^2}{k_0^2(n-k_0)^2n^2\|\Sigma\|_F^2\eta^2}(nM/a_n + 1)\delta^T\Sigma\delta = \frac{a_nn}{\|\Sigma\|_F^2a_n}o\left(\frac{\|\Sigma\|_F^2}{n}\right) = o(1) \rightarrow 0.
\end{align*}
This implies that $\max_{k \in [k_0 - nM/a_n -1, k_0]}R_3(k) = o_p(1)$. Combining the above results we have $\sup_{\gamma \in [-M,0]}R(k) = o_p(1)$, and
$$L_n(\gamma;\tau_0) \rightsquigarrow L(\gamma;\tau_0) := \sqrt{2}|\gamma| + \frac{2\sqrt{2}}{\sqrt{\tau_0(1-\tau_0)}}W^*(\gamma),$$
on $[-M,0]$. By symmetry we can show
$$L_n(\gamma;\tau_0) \rightsquigarrow L(\gamma;\tau_0) := \sqrt{2}|\gamma| + \frac{2\sqrt{2}}{\sqrt{\tau_0(1-\tau_0)}}W^*(\gamma),$$
on $[0,M]$, hence on every compact set on the real line as well. By Theorem 3.2.2 in \cite{vanweak} we have 
$$a_n(\hat{\tau}_U - \tau_0) = \argmin_{\gamma\in(-\infty,\infty)}L_n(\gamma;\tau_0)\overset{\mathcal{D}}{\rightarrow}  \xi{(\tau_0)}.$$
This completes the proof.

\subsection{Proof of Theorem \ref{thm:3regime}}
 (a). For the regime that $a_n/n \rightarrow c \in (0,\infty)$, all proofs are identical to the first regime, i.e., the proofs of Theorems~\ref{thm:rate} and \ref{thm:weak} still hold, except for the arguments used to prove the finite dimensional convergence in the proof of Theorem~\ref{thm:weak}.
 Specifically, when $a_n/n \rightarrow c$, original arguments presented there need a modification to verify Lyapunov's condition.

Consider the case that $k_1 < k_2 < k$ first, where $k_0-k_1 = n\gamma_1/a_n$ and $k_0-k_2 = n\gamma_2/a_n$. For any $\alpha_1,\alpha_2,\alpha_3,\alpha_4 \in \mathbb{R}$, essentially we need to show that
\begin{align*}
    &\frac{\sqrt{2}\sqrt{b_n}}{n\|\Sigma\|_F}\left(\alpha_1\sum_{i = k_1 + 1}^{k_0}\sum_{j = 1}^{k_1}Z_i^TZ_j + \alpha_2\sum_{i = k_2 + 1}^{k_0}\sum_{j = 1}^{k_2}Z_i^TZ_j + \alpha_3\sum_{i = k_1 + 1}^{k_0}\sum_{j = k_0 + 1}^nZ_i^TZ_j + \alpha_4\sum_{i = k_2+1}^{k_0}\sum_{j = k_0+1}^nZ_i^TZ_j\right)
\end{align*}
converges to some normal distribution by applying martingale central limit theorem. Here we need a new way to define the martingale difference sequence by rearranging the observations. Under the new condition where $a_n/n \rightarrow c$, the formulation of the martingale difference sequence in the proof of Theorem \ref{thm:weak} no longer satisfies the Lyapunov condition. However, since the terms we need to work with are all double sums of the inner product of independent random vectors over a two-dimensional  array, we can choose either direction along the array as the martingale difference sequence. The previous formulation is the most natural one, but needs a modification to address the case $a_n/n\rightarrow c$. Now by considering to formulate the martingale difference sequence along the other direction, and after rearranging the terms and defining a new filtration, we can prove Lyapunov's condition under the new assumption. Below are the details.  

Let $\tilde{Z}_1 = Z_{k_1+1}, \tilde{Z}_2 = Z_{k_1+2},...,\tilde{Z}_{k_0-k_1} = Z_{k_0}$, $\tilde{Z}_{k_0-k_1+1} = Z_1,...,\tilde{Z}_{k_0} = Z_{k_1}$ and $\tilde{Z}_{k_0+1} = Z_{k_0+1}, ..., \tilde{Z}_{n} = Z_{n}$. Then 
\begin{align*}
    &\frac{\sqrt{2}\sqrt{b_n}}{n\|\Sigma\|_F}\left(\alpha_1\sum_{i = k_1 + 1}^{k_0}\sum_{j = 1}^{k_1}Z_i^TZ_j + \alpha_2\sum_{i = k_2 + 1}^{k_0}\sum_{j = 1}^{k_2}Z_i^TZ_j + \alpha_3\sum_{i = k_1 + 1}^{k_0}\sum_{j = k_0 + 1}^nZ_i^TZ_j + \alpha_4\sum_{i = k_2+1}^{k_0}\sum_{j = k_0+1}^nZ_i^TZ_j\right)\\
    &=\sum_{i = k_2-k_1+1}^{n}Y_{n,i},
\end{align*}
where 
\begin{enumerate}
    \item $Y_{n,i} = \alpha_2\tilde{Z}_{i}\sum_{j = 1}^{k_2-k_1}\widetilde{Z}_j$, for $i = k_2-k_1+1,...,k_0-k_1$,
    \item $Y_{n,i} = \alpha_1\tilde{Z}_{i}\sum_{j = 1}^{k_0-k_1}\widetilde{Z}_j + \alpha_2\tilde{Z}_{i}\sum_{j = 1}^{k_0-k_1}\widetilde{Z}_j$, for $i = k_0-k_1+1,k_0$,
    \item $Y_{n,i} = \alpha_3\tilde{Z}_{i}\sum_{j = 1}^{k_0-k_1}\widetilde{Z}_j + \alpha_4\tilde{Z}_{i}\sum_{j = k_2-k_1+1}^{k_0-k_1}\widetilde{Z}_j$, for $i = k_0+1,...,n$.
\end{enumerate}
It is easy to verify that $\{Y_{n,i}\}$ is a martingale difference sequence with respect to the natural filtration of $\{\tilde{Z}_i\}$.

Under this formulation, the Lyapunov condition can be verified under Assumption \ref{ass}\ref{ass:tr} and \ref{ass}\ref{ass:cum} and $a_n/n \rightarrow c$. To see this, by the same arguments used in the proof of Theorem \ref{thm:weak},
\begin{align*}
    \sum_{i = k_2-k_1+1}^n\E[Y_{n,i}^4]\leq C(k_0-k_2)n^{-2} + Ck_1n^{-2} + C(n-k_0)n^{-2} = O(1/n) \rightarrow 0.
\end{align*}

The conditional variance can be proved based on similar arguments. This is also true for any other relative orders between $k_1$, $k_2$ and $k_0$. Hence the finite dimensional convergence result still holds.

(b). For the third regime where $a_n/n \rightarrow \infty$, we want to show that $P(\hat{k}_U \neq  k_0) \rightarrow 0$. Assume $k \leq k_0$ and define $\Omega_{n}^{(1)} = \{k: 1 \leq k < k_0/2\}$, $\Omega_n^{(2)} = \{k: k_0/2 \leq k \leq k_0 - \sqrt{n}\}$ and $\Omega_n^{(3)} = \{k: k_0 - \sqrt{n} < k \leq k_0 - 1\}$. We want to show $P(k \in \Omega_n) \rightarrow 0$, where $\Omega = \bigcup_{i = 1}^3\Omega_n^{(i)}$. Under the new assumption $\sqrt{\delta^T\Sigma\delta} = o(\|\delta\|^2)$, according to the proof of Theorem \ref{thm:rate}, for $\Omega_n^{(1)}$ the arguments are basically the same, except for the three terms associated with $\delta^TZ_j$, for $j = 1,2,...,n$. Specifically, 
\begin{align*}
\max_{k = 1,2,...,k_0}|G_n(k) - \E[G_n(k)]| &\leq  \max_{k = 1,2,...,k_0}|G_n^Z(k)| + \max_{k = 1,...,k_0}\left|\frac{2(n-k_0)(n-k-1)(k-1)}{k(n-k)}\sum_{j = 1}^{k}\delta^TZ_j\right| \\
&+ \max_{k = 1,...,k_0}\left|\frac{2(k-1)(n-k_0-1)}{n-k}\sum_{j = k_0+1}^{n}\delta^TZ_j\right| \\
&+ \max_{k = 1,...,k_0}\left|\frac{2(k-1)(n-k_0)}{n-k}\sum_{j = k+1}^{k_0}\delta^TZ_{j}\right|,
\end{align*}
and the last three terms are $o_p(n^{1.5}\|\delta\|^2)$. Hence under Assumption \ref{ass}\ref{ass:delta} and \ref{ass}\ref{ass:rate},
\begin{align*}
    \frac{2\max_{k = 1,2,...,k_0}|G_n(k) - \E[G_n(k)]|}{\frac{(n-k_0-1)(n-1)k_0}{2(n-k_0/2)}\|\delta\|^2}
    &= \frac{\{O_p(n\sqrt{\log(n)}\|\Sigma\|_F) + o_p(n^{1.5}\|\delta\|^2)\}(n-k_0/2)}{(n-k_0-1)(n-1)k_0\|\delta\|^2}\\
    &= O_p\left(\frac{\sqrt{\log(n)}}{\sqrt{a_n}}\right) + o_p(n^{-1/2}) = o_p(1),
\end{align*} 
and
\begin{align*}
&P\left(\max_{k \in \Omega_n^{(1)}(M)}G_n(k) \geq G_n(k_0) \right)\\
\leq & P\left(2\max_{k = 1,2,...,k_0}|G_n(k) - \E[G_n(k)]| \geq \frac{(n-k_0-1)(n-1)k_0}{2(n-k_0/2)}\|\delta\|^2\right)\\
=& P\left(\frac{2\max_{k = 1,2,...,k_0}|G_n(k) - \E[G_n(k)]|}{\frac{(n-k_0-1)(n-1)k_0}{2(n-k_0/2)}\|\delta\|^2} \geq 1 \right)\rightarrow 0
\end{align*}
for sufficiently large $n$.

For $\Omega_n^{(2)} \bigcup \Omega_n^{(3)}$,
\begin{align*}
&G_n(k_0) - G_n(k)\\
= & G_n^Z(k_0) - G_n^Z(k) + (k_0-k)\frac{(n-k_0-1)(n-1)}{n-k}\|\delta\|^2+\frac{2(n-1)(k_0-k)(k_0+k-n)}{k_0(n-k)}\frac{1}{k}\sum_{i = 1}^{k}\delta^TZ_{i}\\
&+\frac{2(n-k_0-1)(n-1)(k-k_0)}{n-k}\frac{1}{n-k_0}\sum_{i = k_0+1}^{n}\delta^TZ_i\\
&+\frac{2(n-1)(k_0^2 - nk_0 - k + n)}{k_0(n-k)}\sum_{i = k+1}^{k_0}\delta^TZ_i\\
=&G_n^Z(k_0) - G_n^Z(k) + (k_0-k)\frac{(n-k_0-1)(n-1)}{n-k}\|\delta\|^2 + R_1(k) + R_2(k) + R_3(k).
\end{align*}

Under $\sqrt{\delta^T\Sigma\delta} = o(\|\delta\|^2)$, following similar arguments we can show $R_i(k)$ are dominated for $i = 1,2,3$. Further since 
\begin{align*}
&G_n^Z(k_0) - G_n^Z(k)\\
= & 2\frac{(n-k-1)(n-1)}{(n-k)k}\sum_{i = k+1}^{k_0}\sum_{j = 1}^{k}Z_i^TZ_j - 2\frac{(k_0-1)(n-1)}{(n-k_0)k_0}\sum_{i = k+1}^{k_0}\sum_{j = k_0+1}^{n}Z_i^TZ_j\\
+ &\frac{(n-1)(n-k-k_0)}{k(n-k_0)}\sum_{i,j = k+1, i\neq j}^{k_0}Z_i^TZ_j- \frac{(n-1)(k_0-k)}{kk_0}\sum_{i,j=1,i\neq j}^{k_0}Z_i^TZ_j\\
+&\frac{(n-1)(k_0-k)}{(n-k_0)(n-k)}\sum_{i,j = k+1,i\neq j}^{n}Z_i^TZ_j - \frac{2(n-1)(n-k_0-k)(k_0-k)}{kk_0(n-k)(n-k_0)}\sum_{i = 1}^{k}\sum_{j = k_0+1}^{n}Z_i^TZ_j\\
=&S_1(k)+S_2(k)+S_3(k)+S_4(k)+S_5(k)+S_6(k),
\end{align*}
the negligibility of all terms can be shown by similar arguments, except for $S_1(k)$. Thus we only provide the proof for $S_1(k)$ here. Essentially we want to show that 
$$P\left(\sup_{k \in \Omega_n^{(2)}\bigcup\Omega_n^{(3)}}\frac{1}{(k_0-k)\|\Sigma\|_F}\left|\left(\sum_{i = k+1}^{k_0}Z_i^T\sum_{j = 1}^{k}Z_j\right)\right| > \eta \sqrt{a_n}\right) < \epsilon.$$

For $\Omega_n^{(2)}$,
\begin{align*}
&P\left(\sup_{k \in \Omega_n^{(2)}}\frac{1}{(k_0-k)\|\Sigma\|_F}\left|\left(\sum_{i = k+1}^{k_0}Z_i^T\sum_{j = 1}^{k}Z_j\right)\right| > \eta \sqrt{a_n}\right)\\
=&P\left(\sup_{k_0-k \in [\sqrt{n},k_0/2]}\frac{1}{(k_0-k)\sqrt{n}\|\Sigma\|_F}\left|\left(\sum_{i = k+1}^{k_0}Z_i^T\sum_{j = 1}^{k}Z_j\right)\right| > \eta \sqrt{a_n}/\sqrt{n}\right)\\
\leq&P\left(\sup_{k_0-k \in [\sqrt{n},k_0/2]}\frac{1}{n\|\Sigma\|_F}\left|\left(\sum_{i = k+1}^{k_0}Z_i^T\sum_{j = 1}^{k}Z_j\right)\right| > \eta \sqrt{a_n}/\sqrt{n}\right) \rightarrow 0,
\end{align*}
since $$\sup_{k_0-k \in [\sqrt{n},k_0/2]}\frac{1}{n\|\Sigma\|_F}\left|\left(\sum_{i = k+1}^{k_0}Z_i^T\sum_{j = 1}^{k}Z_j\right)\right| \leq \sup_{k = 1,2,...,k_0}\frac{1}{n\|\Sigma\|_F}\left|\left(\sum_{i = k+1}^{k_0}Z_i^T\sum_{j = 1}^{k}Z_j\right)\right| = O_p(1),$$
according to Lemma \ref{lem:rateterms}\ref{lem:sum} and the assumption that $\sqrt{a_n/n} \rightarrow \infty$. The proof for $\Omega_n^{(3)}$ is similar to the original proof hence skipped here. 

This completes the proof.

\subsection{Proof of Proposition \ref{prop:compareBai}}

Assume $k \leq k_0$ first. Recall that the sum of square (SSR(k)) defined for $\hat{k}_{LS}$ is
\begin{align*}
    &\sum_{i = 1}^{k}\|X_i - \bar{X}_{1:k}\|^2 + \sum_{i = k+1}^{n}\|X_i - \bar{X}_{(k+1):n}\|^2\\
    =&\sum_{i = 1}^{n}X_i^TX_i - k\bar{X}_{1:k}^T\bar{X}_{1:k} - (n-k)\bar{X}_{(k+1):n}^T\bar{X}_{(k+1):n}\\
    =&\sum_{i = 1}^nZ_i^TZ_i + (n-k_0)\|\delta\|^2+2\sum_{i = k_0+1}^n\delta^TZ_i-\frac{1}{k}\sum_{i,j = 1}^{k}{X^T_{i}X_{j}} - \frac{1}{n-k}\sum_{i,j = k+1}^nX_{i}^TX_j\\
    =&\frac{k-1}{k}\sum_{i = 1}^kZ_i^TZ_i + \frac{n-k-1}{n-k}\sum_{i = k+1}^nZ_i^TZ_j + (n-k_0)\|\delta\|^2 + 2\sum_{i=k_0+1}^n\delta^TZ_i\\
    -&\frac{2(n-k_0)}{n-k}\sum_{i = k +1}^{n}\delta^TZ_i - \frac{1}{n-k}(n-k_0)^2\|\delta\|^2 - \frac{1}{k}\sum_{i,j = 1, i\neq j}^kZ_i^TZ_j - \frac{1}{n-k}\sum_{i,j = k+1,i\neq j}^nZ_i^TZ_j\\
    =&\frac{k-1}{k}\sum_{i = 1}^kZ_i^TZ_i + \frac{n-k-1}{n-k}\sum_{i = k+1}^nZ_i^TZ_j + \frac{(n-k_0)(k_0-k)}{n-k}\|\delta\|^2 + \frac{2(k_0-k)}{n-k}\sum_{i = k_0+1}^n\delta^TZ_i\\
    -&\frac{2(n-k_0)}{n-k}\sum_{i = k+1}^{k_0}\delta^TZ_i - \frac{1}{k}\sum_{i,j = 1, i\neq j}^kZ_i^TZ_j - \frac{1}{n-k}\sum_{i,j = k+1,i\neq j}^nZ_i^TZ_j.\\
\end{align*}
Then the objective function for $\hat{k}_{LS}$ is 
\begin{align*}
    &SSR(k) - SSR(k_0)\\
    =&\frac{k-k_0}{kk_0}\sum_{i = 1}^{k}Z_i^TZ_i + \frac{k_0+k-n}{k_0(n-k)}\sum_{i = k+1}^{k_0}Z_i^TZ_i + \frac{k_0-k}{(n-k_0)(n-k)}\sum_{i = k_0+1}^nZ_i^TZ_i\\
    +&\frac{(n-k_0)(k_0-k)}{n-k}\|\delta\|^2 + \frac{2(k_0-k)}{n-k}\sum_{i = k_0+1}^n\delta^TZ_i
    -\frac{2(n-k_0)}{n-k}\sum_{i = k+1}^{k_0}\delta^TZ_i\\
    -&\frac{1}{k}\sum_{i,j = 1, i\neq j}^kZ_i^TZ_j - \frac{1}{n-k}\sum_{i,j = k+1,i\neq j}^nZ_i^TZ_j+\frac{1}{k_0}\sum_{i,j = 1, i\neq j}^{k_0}Z_i^TZ_j + \frac{1}{n-k_0}\sum_{i,j = k_0+1,i\neq j}^nZ_i^TZ_j := \sum_{j = 1}^{10}I_j
\end{align*}
by a straightforward calculation. According to Theorem 4.2 and its proof in \cite{bai2010common}, $I_4$ and $I_6$ are leading terms, and the asymptotic distribution of $\hat{k}_{LS}$ is based on these two terms. For comparison, the objective function for $\hat{k}_U$ is
\begin{align*}
&G_n(k_0) - G_n(k)\\
= & G_n^Z(k_0) - G_n^Z(k) + (k_0-k)\frac{(n-k_0-1)(n-1)}{n-k}\|\delta\|^2+\frac{2(n-1)(k_0-k)(k_0+k-n)}{k_0(n-k)}\frac{1}{k}\sum_{i = 1}^{k}\delta^TZ_{i}\\
&+\frac{2(n-k_0-1)(n-1)(k-k_0)}{n-k}\frac{1}{n-k_0}\sum_{i = k_0+1}^{n}\delta^TZ_i\\
&+\frac{2(n-1)(k_0^2 - nk_0 - k + n)}{k_0(n-k)}\sum_{i = k+1}^{k_0}\delta^TZ_i\\
=&G_n^Z(k_0) - G_n^Z(k) + (k_0-k)\frac{(n-k_0-1)(n-1)}{n-k}\|\delta\|^2 + R_1(k) + R_2(k) + R_3(k).
\end{align*}

Following the arguments in the Proof of Theorem \ref{thm:3regime}, since under the assumptions proposed in this proposition we have $a_n/n \rightarrow \infty$, $G_n^Z(k_0) - G_n^Z(k)$, $R_1(k)$ and $R_2(k)$ are dominated by the third term above. However $\sqrt{\delta^T\Sigma\delta} = \|\delta\|$ is no longer $o(\|\delta\|^2)$ which makes $R_3(k)$ another leading term. This is equivalent to say that the distribution of $\hat{k}_U$ is mainly determined by the third term and $R_3(k)$, i.e.,  
$$(n-1)\left\{\frac{(n-k_0-1)(k_0-k)}{n-k}\|\delta\|^2 + \frac{2(k_0^2 - nk_0 - k + n)}{k_0(n-k)}\sum_{i = k+1}^{k_0}\delta^TZ_i\right\},$$
which is asymptotically equivalent to 
$$(n-1)\left\{\frac{(n-k_0)(k_0-k)}{n-k}\|\delta\|^2 + \frac{2(k_0 - k)}{n-k}\sum_{i = k+1}^{k_0}\delta^TZ_i\right\},$$
which is the same as the leading term (up to a constant $(n-1)$) in \cite{bai2010common}.

So the limiting distribution for $\hat{k}_U$ is identical to that in Theorem 4.2 in \cite{bai2010common}, which is the minimizer of a two-sided standard brownian motion. For the case $k > k_0$, the conclusion is the same due to similar arguments. This completes the proof.

\subsection{Proof of Theorem \ref{thm:boots}}
To prove the bootstrap consistency, we define $O^*_p(1)$ and $o_p^*(1)$ as the bootstrap stochastic order as done in \cite{chang2003sieve}. They have similar behaviors as traditional $O_p(1)$ and $o_p(1)$, for example, $O_p^*(1)o_p^*(1) = o_p^*(1)$. Note that $\hat{\tau}_U$ is a consistent estimator of $\tau_0$ and $\hat{a}_n$ is a ratio consistent estimator of $a_n$. Thus $|\hat{\tau}_U - \tau_0|$ and $|\hat{a}_n/a_n - 1|$ are $o_p(1)$.

First we are going to show $\hat{a}_n(\hat{\tau}_U^* - \hat{\tau}_U) \overset{\mathcal{D}}{\rightarrow} \xi(\tau_0) \text{ in P}$ by verifying Assumptions \ref{ass} for the bootstrap sample. Conditions (a) and (b) are bootstrap counterparts of Assumptions \ref{ass}\ref{ass:tr} and \ref{ass}\ref{ass:delta}. Since the bootstrap data are Gaussian, any joint cumulants with order higher than 2 are zero. Thus Assumption \ref{ass}\ref{ass:cum} is automatically satisfied for bootstrap data. Because $\hat{a}_n$ is a ratio consistent estimator of $a_n$, Assumption \ref{ass}\ref{ass:rate} is also satisfied stochastically. According to Theorem \ref{thm:weak} and its proof, for any $M >0$, given the data
$$\hat{\tau}_U(1-\hat{\tau}_U)H_n^*(\gamma;\hat{\tau}_U) := \frac{\sqrt{2}\sqrt{\hat{a}_n}(\hat{\tau}_U(1-\hat{\tau}_U))}{n\|\hat{\Sigma}\|_F}\left\{G_n^{\epsilon}(n\widehat{\tau}_U) - G_n^{\epsilon}(\floor{n\widehat{\tau}_U + {n\gamma}/{\hat{a}_n}})\right\} \rightsquigarrow 2\sqrt{2}W^*(\gamma),$$
on the set $[-M,M]$. Furthermore according to Theorem \ref{thm:rate} and its proof, for any $\epsilon > 0$, there exists $M$ and $\Omega_n^*(M) := \{\tau: \hat{a}_n|\tau - \tau_0| > M\}$ such that $P^*(\hat{\tau}^*_U \in \Omega^*(M)) < \epsilon$ with probability converging to 1. This is equivalent to $\hat{a}_n(\hat{\tau}^*_U - \tau_0) = O^*_p(1)$. Thus by Theorem 3.2.3 in \cite{vanweak} (or the proof of Corollary \ref{cor:asymdist}), and since $\hat{\tau}_U$ is a consistent estimator of $\tau_0$, 
$$L_n^*(\gamma;\hat{\tau}_U) := \frac{\sqrt{2}\sqrt{\hat{a}_n}}{n\|\hat{\Sigma}\|_F}\left\{G_n^{*}(n\widehat{\tau}_U) - G_n^{*}(\floor{n\widehat{\tau}_U + {n\gamma}/{\hat{a}_n}})\right\} \rightsquigarrow  L(\gamma;\tau_0) \text{ in P},$$
on any compact set $[-M,M]$, and 
$$\hat{a}_n(\hat{\tau}_U^* - \hat{\tau}_U) = \argmin_{\gamma \in (-\infty,\infty)}L^*_n(\gamma;\hat{\tau}_U) \overset{\mathcal{D}}{\rightarrow} \xi(\tau_0) \text{ in P}.$$

Since $|\hat{a}_n/a_n - 1| = o_p(1)$, we have 
$$a_n(\hat{\tau}_U^* - \hat{\tau}_U) = \hat{a}_n(\hat{\tau}_U^* - \hat{\tau}_U) + (a_n/\hat{a}_n - 1)\hat{a}_n(\hat{\tau}_U^* - \hat{\tau}_U) = \hat{a}_n(\hat{\tau}_U^* - \hat{\tau}_U) + o_p(1)O_p^*(1) = \hat{a}_n(\hat{\tau}_U^* - \hat{\tau}_U) + o_p^*(1).$$
This implies that $a_n(\hat{\tau}^*_U - \hat{\tau}_U)$ converge to the same limit as $\hat{a}_n(\hat{\tau}^*_U - \hat{\tau}_U)$, i.e.
$${a_n}(\hat{\tau}_U^* - \hat{\tau}_U) \overset{\mathcal{D}}{\rightarrow} \xi(\tau_0) \text{ in P}.$$
This completes the proof.

\section{Technical Appendix B}\label{sec:appB}
In this section, we provide the proofs for several lemmas presented in Section \ref{sec:result}.
\begin{proof}[Proof of Lemma \ref{lem:order}]
	If one of $i_1,i_2,i_3,i_4$ is distinct to all other three, then $\E[Z_{i_1,l_1}Z_{i_2,l_2}Z_{i_3,l_3}Z_{i_4,l_4}] = 0$, so it is trivially satisfied. When $i_1 = i_2 < i_3 = i_4$, where we have no more than $(n_2-n_1)^2$ distinct pairs, then 
	$$\sum_{l_1,l_2,l_3,l_4 = 1}^{p}(\E[Z_{i_1,l_1}Z_{i_2,l_2}Z_{i_3,l_3}Z_{i_4,l_4}])^2 = \|\Sigma\|_F^4.$$
	Hence the only remaining case we need to deal with is $i_1 = i_2 = i_3 = i_4$, where we have $n$ distinct cases,
	\begin{align*}
	&\sum_{l_1,l_2,l_3,l_4 = 1}^{p}(\E[Z_{i_1,l_1}Z_{i_1,l_2}Z_{i_1,l_3}Z_{i_1,l_4}])^2\\
	=&\sum_{l_1,l_2,l_3,l_4 = 1}^{p}\{cov(Z_{i_1,l_1},Z_{i_1,l_2})cov(Z_{i_1,l_3},Z_{i_1,l_4})+cov(Z_{i_1,l_1},Z_{i_1,l_3})cov(Z_{i_1,l_2},Z_{i_1,l_4})\\
	&+cov(Z_{i_1,l_1},Z_{i_1,l_4})cov(Z_{i_1,l_2},Z_{i_1,l_3})+cum(Z_{i_1,l_1},Z_{i_1,l_2},Z_{i_1,l_3},Z_{i_1,l_4})\}^2\\
	\leq&C\left\{\sum_{l_1,l_2,l_3,l_4 = 1}^{p}\Sigma_{l_1,l_2}^2\Sigma_{l_3,l_4}^2 + \sum_{l_1,l_2,l_3,l_4 = 1}^{p}cum(Z_{i_1,l_1},Z_{i_1,l_2},Z_{i_1,l_3},Z_{i_1,l_4})^2\right\}\\
	\leq &C\|\Sigma\|_F^4,
	\end{align*}
	under Assumption \ref{ass}\ref{ass:cum} for some generic constant $C$. Hence the conclusion holds.
\end{proof}

\begin{proof}[Proof of Lemma \ref{lem:Xi1Xi4}]
	The arguments are identical to the proof of Lemma 8.1 in \cite{wang2019a}, so are omitted.
\end{proof}

\begin{proof}[Proof of Proposition \ref{prop:terms}]
To show (1), we write
\begin{align*}
    &G_n(k) = \frac{1}{k(n-k)}\sum_{i_1,i_2 = 1, i_1 \neq i_2}^k\sum_{j_1,j_2 = k+1, j_1 \neq j_2}^n(X_{i_1} - X_{j_1})^T(X_{i_2} - X_{j_2}) \\
     =&\frac{1}{k(n-k)}\sum_{i_1,i_2 = 1, i_1 \neq i_2}^k\sum_{j_1,j_2 = k_0+1, j_1 \neq j_2}^n(Z_{i_1} - Z_{j_1} - \delta)^T(Z_{i_2} - Z_{j_2} - \delta)\\
    +& \frac{1}{k(n-k)}\sum_{i_1,i_2 = 1, i_1 \neq i_2}^k\sum_{j_1 = k+1}^{k_0}\sum_{j_2 = k_0+1}^n(Z_{i_1} - Z_{j_1})^T(Z_{i_2} - Z_{j_2} - \delta)\\
    +&\frac{1}{k(n-k)}\sum_{i_1,i_2 = 1, i_1 \neq i_2}^k\sum_{j_1 = k_0 + 1}^n\sum_{j_2 = k+1}^{k_0}(Z_{i_1} - Z_{j_1} - \delta)^T(Z_{i_2} - Z_{j_2})\\
=& G_n^Z(k) + \frac{k(k-1)(n-k_0)(n-k_0-1)}{k(n-k)}\|\delta\|^2 - \frac{1}{k(n-k)}\sum_{i_1,i_2 = 1, i_1 \neq i_2}^k\sum_{j_1,j_2 = k_0+1, j_1 \neq j_2}^n(Z_{i_1} - Z_{j_1})^T\delta\\
-&\frac{1}{k(n-k)}\sum_{i_1,i_2 = 1, i_1 \neq i_2}^k\sum_{j_1,j_2 = k_0+1, j_1 \neq j_2}^n\delta^T(Z_{i_2} - Z_{j_2})-\frac{1}{k(n-k)}\sum_{i_1,i_2 = 1, i_1 \neq i_2}^k\sum_{j_1 = k+1}^{k_0}\sum_{j_2 = k_0+1}^n(Z_{i_1} - Z_{j_1})^T\delta\\
-&\frac{1}{k(n-k)}\sum_{i_1,i_2 = 1, i_1 \neq i_2}^k\sum_{j_1 = k_0 + 1}^n\sum_{j_2 = k+1}^{k_0}\delta^T(Z_{i_2} - Z_{j_2}).\\
\end{align*}
By a straightforward calculation, we can prove the desired result.

To show (2), by calculation we have
\begin{align*}
&G_n^Z(k_0) - G_n^Z(k) \\
= &\frac{1}{k_0}\sum_{i,j = 1, i\neq j}^{k_0}(n-k_0-1)Z_i^TZ_j + \frac{1}{n-k_0}\sum_{i,j = k_0+1, i\neq j}^{n}(k_0-1)Z_i^TZ_j\\
-&2\frac{(k_0-1)(n-k_0-1)}{k_0(n-k_0)}\sum_{i = 1}^{k_0}\sum_{j = k_0+1}^{n}Z_i^TZ_j - \frac{1}{k}\sum_{i,j = 1, i\neq j}^{k}(n-k-1)Z_i^TZ_j\\
-&\frac{1}{n-k}\sum_{i,j = k+1, i\neq j}^{n}(k-1)Z_i^TZ_j + 2\frac{(k-1)(n-k-1)}{k(n-k)}\sum_{i = 1}^{k}\sum_{j = k+1}^{n}Z_i^TZ_j\\
&=A_1+A_2-A_3 - B_1 - B_2 + B_3.
\end{align*}

For each term we have
\begin{align*}
A_1 - B_1 =& \frac{1}{k_0}\sum_{i,j = 1, i\neq j}^{k_0}(n-k_0-1)Z_i^TZ_j - \frac{1}{k}\sum_{i,j = 1, i\neq j}^{k}(n-k-1)Z_i^TZ_j\\
=&2\frac{(n-k-1)}{k}\sum_{i = k+1}^{k_0}\sum_{j = 1}^{k}Z_i^TZ_j + \frac{(n-k-1)}{k}\sum_{i,j=k+1,i\neq j}^{k_0}Z_i^TZ_j\\
-&\frac{(n-1)(k-k_0)}{k_0k}\sum_{i,j = 1, i\neq j}^{k_0}Z_i^TZ_j,
\end{align*}
and
\begin{align*}
A_2 - B_2=&\frac{k_0-1}{n - k_0}\sum_{i,j = k_0+1, i\neq j}^{n}Z_i^TZ_j-\frac{k-1}{n-k}\sum_{i,j = k+1, i\neq j}^{n}Z_i^TZ_j\\
=&-\frac{2(k_0-1)}{n-k_0}\sum_{i = k+1}^{k_0}\sum_{j = k_0+1}^{n}Z_i^TZ_j - \frac{k_0-1}{n-k_0}\sum_{i,j = k+1,i\neq j}^{k_0}Z_i^TZ_j\\ +&\frac{(n-1)(k_0-k)}{(n-k_0)(n-k)}\sum_{i,j = k+1,i\neq j}^{n}Z_i^TZ_j,
\end{align*}
and
\begin{align*}
A_3 - B_3=&\frac{2(k_0-1)(n-k_0-1)}{k_0(n-k_0)}\sum_{i = 1}^{k_0}\sum_{j = k_0+1}^{n}Z_i^TZ_j - \frac{2(k-1)(n-k-1)}{k(n-k)}\sum_{i = 1}^{k}\sum_{j = k+1}^{n}Z_i^TZ_j\\
=&\frac{2(k_0-1)(n-k_0-1)}{k_0(n-k_0)}\sum_{i = k+1}^{k_0}\sum_{j = k_0+1}^{n}Z_i^TZ_j - \frac{2(k-1)(n-k-1)}{k(n-k)}\sum_{i = k+1}^{k_0}\sum_{j = 1}^{k}Z_i^TZ_j \\
+&\frac{2(n-1)(n-k-k_0)(k_0-k)}{k_0(n-k_0)k(n-k)}\sum_{i = 1}^{k}\sum_{j = k_0+1}^{n}Z_i^TZ_j.
\end{align*}

Thus by combining all terms above, we have
\begin{align*}
&G_n^Z(k_0) - G_n^Z(k) = A_1 + A_2 - A_3 - B_1 - B_2 + B_3 \\
=& 2\frac{(n-k-1)(n-1)}{(n-k)k}\sum_{i = k+1}^{k_0}\sum_{j = 1}^{k}Z_i^TZ_j - 2\frac{(k_0-1)(n-1)}{(n-k_0)k_0}\sum_{i = k+1}^{k_0}\sum_{j = k_0+1}^{n}Z_i^TZ_j\\
+ &\frac{(n-1)(n-k-k_0)}{k(n-k_0)}\sum_{i,j = k+1, i\neq j}^{k_0}Z_i^TZ_j- \frac{(n-1)(k_0-k)}{kk_0}\sum_{i,j=1,i\neq j}^{k_0}Z_i^TZ_j\\
+&\frac{(n-1)(k_0-k)}{(n-k_0)(n-k)}\sum_{i,j = k+1,i\neq j}^{n}Z_i^TZ_j - \frac{2(n-1)(n-k_0-k)(k_0-k)}{kk_0(n-k)(n-k_0)}\sum_{i = 1}^{k}\sum_{j = k_0+1}^{n}Z_i^TZ_j.
\end{align*}
This completes the proof. 

\end{proof}

\begin{proof}[Proof of Proposition \ref{prop:tightness}]
	Without loss of generality, we can assume $\gamma_1 < \gamma_2$. By the definition of $H_n(\gamma)$, we have
	\begin{align*}
	&H_n(\gamma_2) - H_n(\gamma_1) \\
	=&\frac{\sqrt{2}\sqrt{b_n}}{n\|\Sigma\|_F}\left\{G_n^Z(k_0) - G_n^Z(k_2) - G_n^Z(k_0) + G_n^Z(k_1)\right\} = -\frac{\sqrt{2}\sqrt{b_n}}{n\|\Sigma\|_F}\left\{G_n^Z(k_2) - G_n^Z(k_1)\right\}\\
	=&\frac{\sqrt{2}\sqrt{b_n}}{n\|\Sigma\|_F}\left(\frac{-2(n-k_1-1)(n-1)}{k_1(n-k_1)}\sum_{i = k_1+1}^{k_2}\sum_{j = 1}^{k_1}Z_i^TZ_j + \frac{2(k_2-1)(n-1)}{(n-k_2)k_2}\sum_{i = k_1+1}^{k_2}\sum_{j = k_2+1}^{n}Z_i^TZ_j\right.\\
	-&\frac{(n-1)(n-k_1-k_2)}{k_1(n-k_2)}\sum_{i,j = k_1+1, i\neq j}^{k_2}Z_i^TZ_j+ \frac{(n-1)(k_2-k_1)}{k_2k_1}\sum_{i,j=1,i\neq j}^{k_2}Z_i^TZ_j\\
	-&\left.\frac{(n-1)(k_2-k_1)}{(n-k_2)(n-k_1)}\sum_{i,j = k_1+1,i\neq j}^{n}Z_i^TZ_j + \frac{2(n-1)(n-k_1-k_2)(k_2-k_1)}{k_1k_2(n-k_1)(n-k_2)}\sum_{i = 1}^{k_1}\sum_{j = k_2+1}^{n}Z_i^TZ_j\right)\\
	:=&H_n^{(1)} + H_n^{(2)}+H_n^{(3)}+H_n^{(4)}+H_n^{(5)}+H_n^{(6)}.
	\end{align*}
	It suffices to show $\E[(H_n^{(i)})^4] \leq Cb_n^2(k_2 - k_1)^2/n^2$ for some generic positive constant $C$, for every $i = 1,2,..,6$.
	
	For $H_n^{(1)}$,
	\begin{align*}
	\E[(H_n^{(1)})^4] &= \frac{64(n-k_1-1)^4(n-1)^4b_n^2}{n^{4}k_1^4(n-k_1)^4\|\Sigma\|_F^4}\sum_{i_1,\i_2,i_3,i_4 = k_1+1}^{k_2}\sum_{j_1,j_2,j_3,j_4 = 1}^{k_1}\\&\sum_{l_1,l_2,l_3,l_4 = 1}^{p}\E[Z_{i_1,l_1}Z_{i_2,l_2}Z_{i_3,l_3}Z_{i_4,l_4}]\E[Z_{j_1,l_1}Z_{j_2,l_2}Z_{j_3,l_3}Z_{j_4,l_4}]\\
	\leq & \frac{Cb_n^2}{n^4\|\Sigma\|_F^4}\sum_{i_1,...,i_4 = k_1+1}^{k_2}\sqrt{\sum_{l_1,...l_4 = 1}^{p}\E[Z_{i_1,l_1}Z_{i_2,l_2}Z_{i_3,l_3}Z_{i_4,l_4}]^2}\\&\sum_{j_1,j_2,j_3,j_4 = 1}^{k_1}\sqrt{\sum_{l_1,...l_4 = 1}^{p}\E[Z_{j_1,l_1}Z_{j_2,l_2}Z_{j_3,l_3}Z_{j_4,l_4}]^2}\\
	\leq & \frac{Cb_n^2}{n^{4}\|\Sigma\|_F^4}(k_2-k_1)^2k_1^2\|\Sigma\|_F^4\leq  C\frac{b_n^2(k_2-k_1)^2}{n^2},
	\end{align*}
	by using Lemma \ref{lem:order} and Cauchy-Schwartz inequality.
	
	Similarly for $H_n^{(2)}$, 
	\begin{align*}
	\E[(H_n^{(2)})^4]
	=& \frac{64b_n^2(k_2-1)^4(n-1)^4}{n^{4}(n-k_2)^4k_2^4|\Sigma\|_F^4}\sum_{i_1,...,i_4 = k_1+1}^{k_2}\sum_{j_1,...,j_4 = k_2 + 1}^{n}\\
	&\sum_{l_1,..,l_4 = 1}^{p}\E[Z_{i_1,l_1}Z_{i_2,l_2}Z_{i_3,l_3}Z_{i_4,l_4}]\E[Z_{j_1,l_1}Z_{j_2,l_2}Z_{j_3,l_3}Z_{j_4,l_4}]\\
	\leq & \frac{Cb_n^2}{n^4\|\Sigma\|_F^4}\sum_{i_1,...,i_4 = k_1+1}^{k_2}\sqrt{\sum_{l_1,...l_4 = 1}^{p}\E[Z_{i_1,l_1}Z_{i_2,l_2}Z_{i_3,l_3}Z_{i_4,l_4}]^2}\\&\sum_{j_1,j_2,j_3,j_4 = k_2+1}^{n}\sqrt{\sum_{l_1,...l_4 = 1}^{p}\E[Z_{j_1,l_1}Z_{j_2,l_2}Z_{j_3,l_3}Z_{j_4,l_4}]^2}\\
	\leq & \frac{Cb_n^2}{n^{4}\|\Sigma\|_F^4}(k_2-k_1)^2(n-k_2)^2\|\Sigma\|_F^4 \leq C\frac{b_n^2(k_2-k_1)^2}{n^2}.
	\end{align*}
	By the same argument we have 
	\begin{align*}
	\E[(H_n^{(6)})^4]
	=&\frac{64(k_2-k_1)^4(n-k_2-k_1)^4(n-1)^4b_n^2}{n^{4}k_1^4k_2^4(n-k_1)^4(n-k_2)^4\|\Sigma\|_F^4}\sum_{i_1,...,i_4 = 1}^{k_1}\sum_{j_1,...j_4 = k_2+1}^{n}\sum_{l_1,..,l_4 = 1}^{p}\\
	&\qquad\E[Z_{i_1,l_1}Z_{i_2,l_2}Z_{i_3,l_3}Z_{i_4,l_4}]\E[Z_{j_1,l_1}Z_{j_2,l_2}Z_{j_3,l_3}Z_{j_4,l_4}]\\
	\leq&\frac{Cb_n^2}{n^{12}\|\Sigma\|_F^4}(k_2-k_1)^4k_1^2(n-k_2)^2\|\Sigma\|_F^4\\
	\leq & \frac{Cb_n^2}{n^{8}}(k_2-k_1)^4 \leq \frac{Cb_n^2}{n^2}(k_2 - k_1)^2.
	\end{align*}
	
	For $H_n^{(3)},H_n^{(4)}$ and $H_n^{(5)}$, we observe that for any $a,b \in [0,1]$, 
	\begin{align*}
	&\E\left[\left(\sum_{i = \floor{na}+1}^{\floor{nb}-1}\sum_{j = \floor{na}+1}^{i}Z_{i+1}^TZ_j\right)^4\right]  \\
	=& 16\sum_{i_1,...i_4 = \floor{na}+1}^{\floor{nb}-1}\sum_{j_1= \floor{na}+1}^{i_1}\sum_{j_2= \floor{na}+1}^{i_2}\sum_{j_3= \floor{na}+1}^{i_3}\sum_{j_4= \floor{na}+1}^{i_4}\sum_{l_1,...l_4 = 1}^{p}\\
	&\E[Z_{i_1+1,l_1}Z_{i_2+1,l_2}Z_{i_3+1,l_3}Z_{i_4+1,l_4}Z_{j_1,l_1}Z_{j_2,l_2}Z_{j_3,l_3}Z_{j_4,l_4}].
	\end{align*}
	First, there are at most $C(\floor{nb} - \floor{na})^4$ terms, for which the corresponding summand is nonzero, since we need at most four distinct values for the indices $i_1,...,i_4,j_1,...j_4$ to make it nonzero. Second, based on Lemma \ref{lem:Xi1Xi4} we know
	$$\left|\sum_{l_1,...l_4 = 1}^{p}\E[Z_{i_1+1,l_1}Z_{i_2+1,l_2}Z_{i_3+1,l_3}Z_{i_4+1,l_4}Z_{j_1,l_1}Z_{j_2,l_2}Z_{j_3,l_3}Z_{j_4,l_4}]\right| \leq C\|\Sigma\|_F^4.$$
	Hence
	$$\E\left[\left(\sum_{i = \floor{na}+1}^{\floor{nb}-1}\sum_{j = \floor{na}+1}^{i}Z_{i+1}^TZ_j\right)^4\right] \leq C(\floor{nb} - \floor{na})^4\|\Sigma\|_F^4.$$
	Applying the above results, we can analyze $H_n^{(3)}$,$H_n^{(4)}$ and $H_n^{(5)}$ as follows,
	\begin{align*}
	\E[(H_n^{(3)})^4]
	\leq& \frac{Cb_n^2(n-k_1-k_2)^4(n-1)^4}{n^{4}k_1^4(n-k_2)^4\|\Sigma\|_F^4}(k_2-k_1)^4\|\Sigma\|_F^4\\
	\leq&\frac{Cb_n^2}{n^4}(k_2-k_1)^4 \leq \frac{Cb_n^2}{n^2}(k_2-k_1)^2.
	\end{align*}
	Similarly for $H_n^{(4)}$,
	\begin{align*}
\E[(H_n^{(4)})^4]
	\leq& \frac{Cb_n^2(n-1)^4(k_2-k_1)^4}{n^{4}k_2^4k_1^4\|\Sigma\|_F^4}k_2^4\|\Sigma\|_F^4\\
	\leq&\frac{Cb_n^2}{n^4}(k_2-k_1)^4 \leq \frac{Cb_n^2}{n^2}(k_2-k_1)^2.
	\end{align*}
	Finally for $H_n^{(5)}$,
	\begin{align*}
	\E[(H_n^{(5)})^4]
	\leq& \frac{Cb_n^2(n-1)^4(k_2-k_1)^4}{n^{4}(n-k_2)^4(n-k_1)^4\|\Sigma\|_F^4}(n-k_1)^4\|\Sigma\|_F^4\\
	\leq&\frac{Cb_n^2}{n^4}(k_2-k_1)^4 \leq \frac{Cb_n^2}{n^2}(k_2-k_1)^2.
	\end{align*}
	Combining all these results, the proof is complete.	
\end{proof}

\begin{proof}[Proof of Proposition \ref{lem:S1}]
    Following the definition in the proof of Theorem \ref{thm:rate}, for $k \in \Omega_n^{2}(M)$, $n/\sqrt{a_n}\leq k_0 - k \leq k_0/2$. Let $b_n = \sqrt{a_n}$ and $\gamma = (k_0 - k)b_n/n$, then

\begin{align*}
&P\left(\sup_{k \in \Omega_n^{(2)}(M)}\frac{1}{(k_0-k)\|\Sigma\|_F}\left|\left(\sum_{i = k+1}^{k_0}Z_i^T\sum_{j = 1}^{k}Z_j\right)\right| > \eta \sqrt{a_n}\right)\\
=&P\left(\sup_{\gamma \in [1,b_n\tau_0/2]}\frac{b_n}{n\gamma\|\Sigma\|_F}\left|\left(\sum_{i = k_0 - n\gamma/b_n+1}^{k_0}Z_i^T\sum_{j = 1}^{k_0 - n\gamma/b_n}Z_j\right)\right| > \eta b_n\right)\\
\leq&P\left(\sup_{\gamma \in [1,M]}\frac{\sqrt{b_n}}{n\|\Sigma\|_F}\left|\left(\sum_{i = k_0 - n\gamma/b_n+1}^{k_0}Z_i^T\sum_{j = 1}^{k_0 - n\gamma/b_n}Z_j\right)\right| > \eta \sqrt{b_n}\right) \\
+& P\left(\sup_{\gamma \in [M,b_n\tau_0/2]}\frac{\sqrt{b_n}}{nM\|\Sigma\|_F}\left|\left(\sum_{i = k_0 - n\gamma/b_n+1}^{k_0}Z_i^T\sum_{j = 1}^{k_0 - n\gamma/b_n}Z_j\right)\right| > \eta \sqrt{b_n}\right)\\
\leq&P\left(\sup_{\gamma \in [-M,M]}\frac{\sqrt{b_n}}{n\|\Sigma\|_F}\left|\left(\sum_{i = k_0 - n\gamma/b_n+1}^{k_0}Z_i^T\sum_{j = 1}^{k_0 - n\gamma/b_n}Z_j\right)\right| > \eta \sqrt{b_n}\right)\\
+&P\left(\sup_{\gamma \in [M,b_n\tau_0/2]}\frac{1}{nM\|\Sigma\|_F}\left|\left(\sum_{i = k_0 - n\gamma/b_n+1}^{k_0}Z_i^T\sum_{j = 1}^{k_0 - n\gamma/b_n}Z_j\right)\right| > \eta\right)\\
\leq&O(1/\sqrt{b_n}) + C/\eta M,
\end{align*}	
for some constant $C$, due to Theorem \ref{thm:weak} and its proof, and the fact that $b_n = \sqrt{a_n} = o(\sqrt{n})$ and 
$$\sup_{\gamma \in [-M,M]}\frac{\sqrt{b_n}}{n\|\Sigma\|_F}\left|\left(\sum_{i = k_0 - n\gamma/b_n+1}^{k_0}Z_i^T\sum_{j = 1}^{k_0 - n\gamma/b_n}Z_j\right)\right| = O_p(1
).$$

Finally for $k \in \Omega_n^{(3)}(M)$, notice that 
\begin{align*}
&\left|\left(\sum_{i = k+1}^{k_0}Z_i^T\sum_{j = 1}^{k}Z_j\right)\right|\\
=&\left|\left(\sum_{i = k+1}^{k_0}Z_i^T\sum_{j = 1,i\neq j}^{n}Z_j\right) - \left(\sum_{i = k+1}^{k_0}Z_i^T\sum_{j = k_0+1}^{n}Z_j\right) - 2\left(\sum_{i = k+1}^{k_0}Z_i^T\sum_{j = i+1}^{k_0}Z_j\right)\right|\\
\leq&\left|\left(\sum_{i = k+1}^{k_0}Z_i^T\sum_{j = 1,i\neq j}^{n}Z_j\right)\right| + \left|\left(\sum_{i = k+1}^{k_0}Z_i^T\sum_{j = k_0+1}^{n}Z_j\right)\right|+\left| 2\left(\sum_{i = k+1}^{k_0}Z_i^T\sum_{j = i+1}^{k_0}Z_j\right)\right|.
\end{align*}

We have shown that, by Lemma \ref{lem:rateterms}\ref{lem:HRcross}, 
$$P\left(\sup_{k \in \Omega_n^{(3)}(M)}\frac{1}{(k_0 - k)\|\Sigma\|_F}\left|\left(\sum_{i = k+1}^{k_0}Z_i^T\sum_{j = k_0+1}^{n}Z_j\right)\right| > \eta\sqrt{a_n}\right) \leq C/M,$$
and
$$P\left(\sup_{k \in \Omega_n^{(3)}(M)}\frac{1}{(k_0 - k)\|\Sigma\|_F}\left| 2\left(\sum_{i = k+1}^{k_0}Z_i^T\sum_{j = i+1}^{k_0}Z_j\right)\right| > \eta\sqrt{a_n}\right) \leq C\log(n)/a_n.$$

For the first term, it is easy to see that $\left\{\left|\left(\sum_{i = k+1}^{k_0}Z_i^T\sum_{j = 1,i\neq j}^{n}Z_j\right)\right|\right\}_{k \in \Omega_n^{(3)}(M)}$ has exactly the same joint distribution as
$\left\{\left|\left(\sum_{i = n/\sqrt{a_n} - (k_0-k)+1}^{n/\sqrt{a_n}}Z_i^T\sum_{j = 1,i\neq j}^{n}Z_j\right)\right|\right\}_{k \in \Omega_n^{(3)}(M)}$ by shifting the indices, and 
\begin{align*}
&P\left(\sup_{nM/a_n < k_0 - k < n/\sqrt{a_n}}\frac{1}{(k_0-k)\|\Sigma\|_F}\left|\left(\sum_{i = k+1}^{k_0}Z_i^T\sum_{j = 1,i\neq j}^{n}Z_j\right)\right| >\eta\sqrt{a_n} \right)\\
=&P\left(\sup_{nM/a_n < l < n/\sqrt{a_n}}\frac{1}{l\|\Sigma\|_F}\left|\left(\sum_{i = n/\sqrt{a_n} - l+1}^{n/\sqrt{a_n}}Z_i^T\sum_{j = 1,i\neq j}^{n}Z_j\right)\right| >\eta\sqrt{a_n} \right).
\end{align*}

According to the same decomposition, 
\begin{align*}
&\left|\left(\sum_{i = n/\sqrt{a_n} - l+1}^{n/\sqrt{a_n}}Z_i^T\sum_{j = 1,i\neq j}^{n}Z_j\right)\right|\\
\leq & \left|\left(\sum_{i = n/\sqrt{a_n} - l+1}^{n/\sqrt{a_n}}Z_i^T\sum_{j = i+1}^{n}Z_j\right)\right| +\left|\left(\sum_{i = n/\sqrt{a_n} - l+1}^{n/\sqrt{a_n}}Z_i^T\sum_{j = 1}^{i-1}Z_j\right)\right| \\
\leq & \left|\left(\sum_{i = n/\sqrt{a_n} - l+1}^{n/\sqrt{a_n}}Z_i^T\sum_{j = i+1}^{n}Z_j\right)\right| +\left|\left(\sum_{i = 1}^{n/\sqrt{a_n}}Z_i^T\sum_{j = 1}^{i-1}Z_j\right)\right|+\left|\left(\sum_{i = 1}^{n/\sqrt{a_n} - l}Z_i^T\sum_{j = 1}^{i-1}Z_j\right)\right|\\
\leq &|U^{(1)}_l| + |U^{(2)}| + |U^{(3)}_{n/\sqrt{a_n} - l}|.
\end{align*}

It is easy to show that $U^{(1)}_l$ is a martingale sequence adaptive to the filtration $\cF^{(1)}_l = \sigma(Z_{n/\sqrt{a_n} - l + 1}, ..., Z_n)$. And  $U^{(3)}_{l'}$ is also a martingale sequence adaptive to $\cF^{(3)}_{l'} = \sigma(Z_1,...,Z_{l'})$. 

Hence, by Lemma \ref{HR1},
\begin{align*}
&P\left(\sup_{nM/a_n < l < n/\sqrt{a_n}} \frac{1}{l\|\Sigma\|_F}\left|U_l^{(1)}\right| > \eta\sqrt{a_n}\right) \leq \sum_{l = nM/a_n}^{n/\sqrt{a_n}}\frac{n}{\eta^2a_nl^2} + \frac{n}{\eta^2a_n}\frac{a_n^2}{n^2M^2} \leq C/M.
\end{align*}
For $U^{(2)}$,
\begin{align*}
P\left(\sup_{nM/a_n < l < n/\sqrt{a_n}} \frac{1}{l\|\Sigma\|_F}\left|U^{(2)}\right| > \eta\sqrt{a_n}\right) &= P\left( \frac{a_n}{nM\|\Sigma\|_F}\left|U^{(2)}\right| > \eta\sqrt{a_n}\right) \\
&\lesssim \frac{a_n^2n^2/a_n}{a_n\eta^2n^2M^2} \leq C/M^2
\end{align*}

And for $U^{(3)}_{n/\sqrt{a_n} - l}$,
\begin{align*}
&P\left(\sup_{nM/a_n < l < n/\sqrt{a_n}} \frac{1}{l\|\Sigma\|_F}\left|U^{(3)}_{n/\sqrt{a_n}-l}\right| > \eta\sqrt{a_n}\right)\\
=&P\left(\sup_{nM/a_n < l < n/\sqrt{a_n}} \frac{1}{l\|\Sigma\|_F}\sup_{nM/a_n < l < n/\sqrt{a_n}}\left|U^{(3)}_{n/\sqrt{a_n}-l}\right| > \eta\sqrt{a_n}\right)\\
=&P\left( \frac{a_n}{nM\|\Sigma\|_F}\sup_{1 \leq  l' < n/\sqrt{a_n} - nM/a_n}\left|U^{(3)}_{l'}\right| > \eta\sqrt{a_n}\right) \leq \frac{a_n^2}{n^2M^2\eta^2a_n}\sum_{l = 1}^{n/\sqrt{a_n} - nM/a_n}l \leq C/M^2.
\end{align*}

Thus by combining all the results above, we have $$P\left(\sup_{k \in \Omega_n^{(2)}(M)\bigcup\Omega_n^{(3)}(M)}\frac{1}{(k_0-k)\|\Sigma\|_F}\left|\left(\sum_{i = k+1}^{k_0}Z_i^T\sum_{j = 1}^{k}Z_j\right)\right| > \eta \sqrt{a_n}\right) < \epsilon,$$
by selecting $M = C/\epsilon$ for some constant $C$. 
\end{proof}

\begin{proof}[Proof of Lemma \ref{lem:rateterms}]
\ref{lem:sumdelta} is a direct consequence of Kolmogorov's inequality. To see this, note that
$$\max_{1 \leq k_1 < k_2 \leq n}\left|\sum_{i = k_1 + 1}^{k_2}\delta^TZ_i\right| = \max_{1 \leq k_1 < k_2 \leq n}\left|\sum_{i = 1}^{k_2}\delta^TZ_i - \sum_{i = 1}^{k_1}\delta^TZ_i\right| \leq 2\max_{1 \leq k_1 \leq n}\left|\sum_{i = 1}^{k_1}\delta^TZ_i\right|.$$
By Kolmogorov's inequaility, since $\delta^TZ_i$ are i.i.d. with mean zero and finite variance, for any positive $\lambda$
$$P\left(\max_{1 \leq k_1 \leq n}\left|\sum_{i = 1}^{k_1}\delta^TZ_i\right| > \lambda\right) \leq \frac{n\delta^T\Sigma\delta}{\lambda^2},$$
which implies \ref{lem:sumdelta} since $\delta^T\Sigma\delta = o(n\|\delta\|^4/a_n)$ according to Assumption \ref{ass}\ref{ass:delta}.

To prove \ref{lem:sum}, by applying the continuous mapping theorem together with Proposition \ref{prop:Q}, we have 
$$\frac{1}{n\|\Sigma\|_F}\max_{1 \leq k_1 < k_2 \leq n}\left|\sum_{i = k_1}^{k_2}\sum_{j = k_1}^iZ_{i+1}^TZ_j\right| \overset{\mathcal{D}}{\rightarrow} \sup_{0 \leq a < b \leq 1}|Q(a,b)|,$$
which is bounded in probability. 

For \ref{lem:sumcross}, notice that
$$\sum_{i = k_1}^{k_2}\sum_{j = k_2+1}^{k_3}Z_{i}^TZ_j = \sum_{i = k_1}^{k_3}\sum_{j = k_1}^iZ_{i+1}^TZ_j - \sum_{i = k_1}^{k_2}\sum_{j = k_1}^iZ_{i+1}^TZ_j - \sum_{i = k_2}^{k_3}\sum_{j = k_2}^iZ_{i+1}^TZ_j.$$
So
$$\max_{1 \leq k_1 < k_2 < k_3 \leq n}\left|\sum_{i = k_1}^{k_2}\sum_{j = k_2+1}^{k_3}Z_{i}^TZ_j\right| \leq 3\max_{1 \leq k_1 < k_2 \leq n}\left|\sum_{i = k_1}^{k_2}\sum_{j = k_1}^iZ_{i+1}^TZ_j\right| = O_p(n\|\Sigma\|_F),$$
according to \ref{lem:sum}.

\ref{lem:avedelta1} and \ref{lem:avedelta2} can be proved by essentially the same arguments so we only prove \ref{lem:avedelta1} here. By Lemma \ref{HR1}, 
$$P\left(\max_{k_1\leq k \leq k_2}\frac{1}{k}\left|\sum_{i = 1}^k \delta^TZ_i\right| > \lambda \right) \leq \frac{1}{\lambda^2}\sum_{i = k_1+1}^{k_2}\frac{1}{i^2}Var(\delta^TZ_i) + \frac{1}{\lambda^2k_1^2}\sum_{i = 1}^{k_1}Var(\delta^TZ_i) \leq \frac{C\delta^T\Sigma\delta}{\lambda^2k_1},$$
for any $\lambda > 0$. Thus $\max_{k_1\leq k \leq k_2}\frac{1}{k}\left|\sum_{i = 1}^k \delta^TZ_i\right|$ is $o_p(\sqrt{n}\|\delta\|^2/\sqrt{a_nk_1})$, under Assumption \ref{ass}\ref{ass:delta}.

	To show \ref{lem:HRcross}, note that $$\frac{1}{k_0-k}\sum_{i = k+1}^{k_0}\sum_{k_0+1}^{n}Z_i^TZ_j \overset{\mathcal{D}}{=} \frac{1}{k_0-k}\sum_{i = 1}^{k_0-k}\sum_{j = -n+k_0+1}^{0}Z_i'^TZ_j' = \frac{1}{k_0-k}\sum_{i = 1}^{k_0-k}s^{(1)}_i,$$
	where $Z_i'$ are i.i.d. coupled copy of $Z_i$ and $s_i^{(1)} =Z_i'^T(\sum_{j = -n+k_0+1}^{0}Z_j') $. By defining $\mathcal{F}^{(1)}_t$ as the natural filtration of $Z_i'$, $s_i^{(1)}$ is a martingale difference sequence with variance $Var(s_i^{(1)}) = (n-k_0)\|\Sigma\|_F^2$. Directly applying  H\'ajek-R\'enyi's inequality, i.e. Lemma \ref{HR1}, we have
	
	\begin{align*}
	P\left(\max_{nM/a_n \leq l \leq k_0-1}\frac{1}{l}\left|\sum_{i = 1}^{l}s_i^{(1)}\right| > \lambda\right) &\leq \frac{(n-k_0)\|\Sigma\|_F^2}{\lambda^2}\left(\frac{a_n}{nM} + \sum_{i = nM/a_n+1}^{k_0-1}i^{-2}\right) \\
	&\leq \frac{C_1(n-k_0)\|\Sigma\|_F^2}{\lambda^2}\left(\frac{a_n}{nM}\right).
	\end{align*}
	
	By similar arguments, we can prove \ref{lem:HRcrosssum}, by choosing $c_k = 1$. For \ref{lem:HRmid}, observe that 
	$$\sum_{i,j = k+1, i\neq j}^{k_0}Z_i^TZ_j = 2\sum_{l = 1}^{k_0-k-1}\sum_{j = 1}^{l}Z_{k_0-l}^TZ_{k_0-l+j} \overset{\mathcal{D}}{=} 2\sum_{l = 1}^{k_0-k-1}\sum_{j = 1}^{l}Z_{-l}^TZ_{j-l} = 2\sum_{l = 1}^{k_0-k-1}s_l^{(2)},$$ where $s_{l}^{(2)} = Z_{-l}^T\sum_{j = 1}^{l}Z_{j-l}$, and $s_l^{(2)}$ is a martingale difference sequence adapted to the filtration $\cF^{(2)}_l$ defined as $\cF^{(2)}_l = \sigma(Z_{-l},Z_{-l+1},Z_{-l+2},\cdots)$. It is very easy to see that $Var(s_l^{(2)}) = l\|\Sigma\|_F^2$. Then by Lemma \ref{HR1} we have
	\begin{align*}
	P\left(\max_{nM/a_n\leq h \leq k_0}\left|\frac{1}{h}\sum_{l = 1}^{h}s_l^{(2)}\right|\geq \lambda\right) &\leq \frac{\|\Sigma\|_F^2}{\lambda^2}\left(\left(\frac{nM}{a_n}\right)^{-2}\sum_{i = 1}^{nM/a_n}i + \sum_{i = nM/a_n}^{k_0}i^{-1}\right)\\
	&\leq \frac{C_3\|\Sigma\|_F^2}{\lambda^2}\log(k_0).
	\end{align*}
\end{proof}

\end{document}